\documentclass[format=acmsmall, review=false]{acmart}
\usepackage{acm-ec-26-proc}
\usepackage{booktabs} %
\usepackage[ruled]{algorithm2e} %
\usepackage{natbib} %

\SetAlFnt{\small}
\SetAlCapFnt{\small}
\SetAlCapNameFnt{\small}
\SetAlCapHSkip{0pt}
\IncMargin{-\parindent}
\setcitestyle{authoryear}

\usepackage{natbib}  %
\usepackage{caption} %

\acmConference{}{}{} %
\AtBeginDocument{
\fancypagestyle{firstpagestyle}{%
	\fancyhf{}%
	\fancyfoot[L]{\footnotesize{}}%
}
\fancypagestyle{ourstyle}{%
	\fancyhf{}%
	\fancyhead[C]{\footnotesize{\mbox{\shortauthors}}}
	\fancyhead[R]{\thepage}%
}
\pagestyle{ourstyle}
\renewcommand{\footnotetextcopyrightpermission}[1]{} %
\renewcommand{\footnotetextauthorsaddresses}[1]{} %
}

\usepackage[dvipsnames]{xcolor}
\usepackage{tabularx}
\usepackage{amsmath}
\usepackage{amsthm}
\usepackage{amsfonts,textcomp}
\usepackage{pifont}
\usepackage{latexsym}
\usepackage{graphicx} 
\usepackage{wrapfig} 
\usepackage{multirow}
\usepackage{microtype}
\usepackage{rotating}
\usepackage{url}
\usepackage{color}
\usepackage{comment}
\usepackage{paralist}
\usepackage{booktabs}
\usepackage{bm}
\usepackage[]{units}
\usepackage{tikz}
\usepackage{natbib}
\usepackage{xfrac}
\usetikzlibrary {patterns.meta} 
\tikzdeclarepattern{
  name=mylines,
  parameters={
      \pgfkeysvalueof{/pgf/pattern keys/size},
      \pgfkeysvalueof{/pgf/pattern keys/angle},
      \pgfkeysvalueof{/pgf/pattern keys/line width},
  },
  bounding box={
    (0,-0.5*\pgfkeysvalueof{/pgf/pattern keys/line width}) and
    (\pgfkeysvalueof{/pgf/pattern keys/size},
0.5*\pgfkeysvalueof{/pgf/pattern keys/line width})},
  tile size={(\pgfkeysvalueof{/pgf/pattern keys/size},
\pgfkeysvalueof{/pgf/pattern keys/size})},
  tile transformation={rotate=\pgfkeysvalueof{/pgf/pattern keys/angle}},
  defaults={
    size/.initial=5pt,
    angle/.initial=45,
    line width/.initial=.4pt,
  },
  code={
      \draw [line width=\pgfkeysvalueof{/pgf/pattern keys/line width}]
        (0,0) -- (\pgfkeysvalueof{/pgf/pattern keys/size},0);
  },
}

\usepackage[autostyle, english = american]{csquotes}
\MakeOuterQuote{"}
\usepackage{enumitem}

\usepackage[nameinlink]{cleveref}
\usepackage{algpseudocode}
\usepackage{pgfplots}
\pgfplotsset{compat=1.18}
\usepackage{subcaption}
\usetikzlibrary{shapes.geometric, arrows, patterns, patterns.meta}
\usepackage{tabularray}
\UseTblrLibrary{booktabs}
\usepackage{thmtools, thm-restate}
\usepackage{wrapstuff}
\usepackage{arydshln}

\newtheorem{definition}{Definition}
\newtheorem{proposition}{Proposition}

\newtheorem{claim}{Claim}
\crefname{claim}{claim}{claims}
\theoremstyle{definition}

\newtheorem{example}{Example}

\newtheorem{method}{Method}

\allowdisplaybreaks

\RequirePackage{ifthen,calc}
\newsavebox{\ffbox}\newlength{\ffboxlen}
\newcommand{\todo}[1]{%
  {\sbox{\ffbox}{\textbf{TODO:}\ \textit{{#1}}\ \textbf{:ODOT}}
    \settowidth{\ffboxlen}{\usebox{\ffbox}}
		\addtolength{\ffboxlen}{-5mm}
    \ifthenelse{\ffboxlen>\linewidth}{%
      \noindent\marginpar{$>>>>$}\textbf{TODO:}\ \textit{{#1}}\ \textbf{:ODOT}\marginpar{$<<<<$}}{%
      \noindent\marginpar{$>><<$}\textbf{TODO:}\ \textit{{#1}}\ \textbf{:ODOT}}}}

\newcommand{\naturals}{{{\mathbb{N}}}}

\newcommand{\x}{\ding{51}}

\newcommand{\alg}[1]{\textproc{#1}}
\newcommand{\pay}[1]{\textproc{#1}}
\newcommand{\mix}{\mathcal{R}}

\newcommand{\mesincrease}[1]{\alg{MES}^{\pay{#1}+}}
\newcommand{\seqPhrag}{\alg{SeqPhragmén}\xspace}

\newif\ifcomments %
\commentstrue
\commentsfalse

\definecolor{dark-gray}{gray}{0.4}

\newcommand{\anton}[1]{{\color{olive}{\ifcomments AB: #1 \fi}}}

\newcommand{\mut}[1]{{\color{Mahogany}{\ifcomments MU: #1 \fi}}}

\title{Mixed Voting Rules for Participatory Budgeting}

\author{Anton Baychkov}
\affiliation{%
\institution{University of Warwick}
\country{United Kingdom}
}

\author{Markus Brill}
\affiliation{%
\institution{University of Oxford}
\country{United Kingdom}
}

\author{Markus Utke}
\affiliation{%
\institution{TU Eindhoven}
\country{The Netherlands}
}

\begin{abstract}
     Designing and analyzing voting rules for Participatory Budgeting (PB) elections is an active research area in computational social choice. Many PB voting rules aim to optimize a specific objective. For instance, the ubiquitous \textit{Greedy rule} attempts to maximize utilitarian welfare, while the \textit{Method of Equal Shares (MES)} aims to achieve proportional representation. However, it is often desirable to achieve good outcomes on multiple objectives rather than a close-to-perfect outcome for one. Inspired by mixed-member systems for parliamentary elections, we introduce \textit{mixed voting rules} for PB. These are composed of a sequence of two or more rules that can each spend some fraction of the overall budget in order to add projects to the set selected by earlier rules. We develop a theoretical framework for formulating and analyzing mixed PB voting rules, and explore how existing rules can be adapted to this framework. We particularly focus on MES and its potential to address imbalances in representation created by earlier rules. We propose different ways to adjust MES voter budgets based on how satisfied voters are with previously chosen projects, and examine how well the resulting rules approximate well-known proportionality axioms such as EJR+. In particular, we show that one of these methods improves upon a natural proportionality baseline. We also extend our main positive result to  general additive satisfaction functions. 
     We complement our theoretical results with an extensive empirical analysis of real-world PB elections. Our experiments show that mixed rules can achieve favorable trade-offs between utilitarian welfare and proportionality. We identify several refinements that further improve their performance, and apply our framework to PB rules beyond Greedy and MES. 
\end{abstract}

\begin{document}

\maketitle

\anton{COMMENTS ARE ON}

\section{Introduction}

\noindent Participatory Budgeting (PB) is a democratic innovation that allows residents to vote on how public money is spent \citep{WMT21a}. Typically, community members suggest projects, each with a specific cost. Voters are then asked to express preferences over these projects, and a voting rule selects a subset of the projects to fund, making sure that their total cost stays within a given budget. Designing and analyzing PB voting rules is an active research area in computational social choice \citep{ReMa23a}.

PB voting rules have different strengths and weaknesses. For example, the \textit{Greedy} rule simply ranks projects by the number of votes they receive and funds them in that order until the budget runs out. While this method is straightforward to implement and explain, it may fail to represent minority interests. In contrast, the \textit{Method of Equal Shares (MES)} \citep{PPS21a} satisfies strong proportionality axioms, but computing it requires more complex calculations that are harder to explain.
Instead of choosing between different rules, %
we propose a framework that combines multiple rules. Ideally, this approach preserves the advantages of the constituent rules while mitigating their downsides. In our framework, a \emph{mixed voting rule} is defined as a sequence of rules, and each rule in that sequence is allocated a specific portion of the overall budget. The process can be visualized as an assembly line: %
the first rule selects projects using its assigned budget share, then passes both its selection and any unused budget to the next rule. Subsequent rules cannot alter the projects already chosen by earlier rules but can only add new projects using their budget allocation. For example, we might let the Greedy rule spend 60\% of the budget first and then allow MES to spend the remaining 40\%. %

This idea is inspired by \textit{mixed-member electoral systems}, which are used for parliamentary elections around the world \citep{ShWa03a}. 
Prominent examples include Germany's \textit{Mixed-Member Proportional Representation} system and Scotland's \textit{Additional Member System}. 
In these electoral systems, parliamentary seats are allocated according to two distinct selection methods: Some representatives are elected directly (usually through first-past-the-post voting in local districts), while additional seats are allocated to ensure\,---\,or at least approximate\,---\,the proportional representation of political parties in the parliament as a whole. 
Thus, these systems effectively divide the total resource (i.e., the parliamentary seats) between two complementary selection methods, with the second method specifically designed to enhance proportionality.
Similarly, our mixed voting rules for PB allocate portions of the total budget to different selection methods, and later rules may enhance the proportionality of the outcome.\footnote{The analogy has limitations, as PB elections often (but not always) lack the concept of geographic districts that is central to most mixed-member electoral systems. Moreover, voters often (but not always) submit two ballots in a mixed-member~system.}

To be able to apply existing PB voting rules in our mixed framework, we must adapt them to handle scenarios where some projects have already been selected. For the Greedy rule, this adaptation is straightforward: We can simply restrict attention to the remaining unselected projects and iteratively choose those with the highest vote counts. 
Adapting MES, on the other hand, is more subtle. Since the rule is defined via individual voter budgets, we need to decide how to divide the remaining budget among the voters. One trivial way of doing that is to split the budget equally, essentially ignoring the set of previously selected projects. 
In order to enhance proportionality, %
however, we need more sophisticated methods that explicitly take the previously selected projects into account. In this paper, we propose several approaches to adapting MES (and other PB rules) to this context, and we analyze\,---\,theoretically and empirically\,---\,how these methods perform in terms of proportional representation, among other metrics.

Beyond their role in our mixed framework, these adaptations might be of independent interest, as they allow PB voting rules to be applied to situations where certain projects must be included (maybe due to administrative or legal constraints). For example, a municipality might have ongoing projects that need to continue or legally required initiatives that must be funded. 
Another special case that is covered by our framework is the class of PB ``completion methods,'' i.e., rules that are designed to extend an outcome of another rule and make it exhaustive, in the sense that no more projects can be afforded with the unused budget.

\paragraph{Our Contribution} In this paper, we extend PB voting rules to work with a set of \textit{pre-selected projects}. We then use these extended rules to define \textit{mixed voting rules}, composed of a collection of voting rules that are executed in sequence (\Cref{sec:mixed}). We adapt MES to the mixed rules framework by defining ``pre-allocation'' methods to account for pre-selected projects (\Cref{sec:adapting_mes}) and establish proportionality guarantees for these methods using parameterized variants of EJR+ (\Cref{sec:proportionality}). We complement these theoretical guarantees with an empirical analysis of mixed rules on real-world PB instances (\Cref{sec:experiments}). Finally, we briefly discuss mixed rules for the special case of multiwinner voting, as well as a more general PB framework (\Cref{sec:extensions}). 
Omitted proofs can be found in \Cref{app:omitted_proofs}. %

\paragraph{Related Work}

Recent years have witnessed significant attention from the (computational) social choice community to multiwinner elections \citep{FSST17a,LaSk22a} and PB~\citep{AzSh21a,ReMa23a}. 
By contrast, mixed-member electoral systems are mostly studied within the political science literature~\citep{ShWa03a}. %
Proportionality in the PB setting has been a key research direction \citep{ALT18a,LCG22a,BFL+23a}. Prime examples of proportionality notions include \textit{Extended Justified Representation} \citep{ABC+16a}, its ``up to one'' variant \citep{PPS21a}, and its strengthenings \textit{EJR+}~\citep{BrPe23a} and \textit{Strong EJR+}~\citep{Skow26a}. 
Several proportional rules have been proposed, most prominently the \textit{Method of Equal Shares} \citep{PeSk20b}. 
Parameterized proportionality axioms, which play a central role in our paper, have been considered by \citet{BBMP25a} and \citet{Waru21a}. 
Recently, there has been some work on analyzing the performance of (non-mixed) multiwinner voting rules according to competing objectives \citep{KKE+19a,EFI+24a,BrPe24a} and on best-of-both-worlds approaches \citep{GPS+24a}. Extending an already selected set of projects has been studied in the context of PB completion methods \citep{PPS21a}. More generally, the issue of extending partially specified solutions also features, at least implicitly, in inter-temporal fairness notions where repeated decisions need to be made \citep{FZC17a, Lack20a,IsBr25a}.

\section{Preliminaries}\label{sec:preliminaries}

Let $P$ be a set of projects and $N=\{1,...,n\}$ a set of voters. %
We assume approval preferences and let $A_i \subseteq P$ denote the set of projects approved by voter $i \in N$. For a project $p \in P$, we let $N_p=\{i\in N: p\in A_i\}$ denote the \textit{supporters of $p$}, i.e., the set of voters approving $p$. We make the technical assumption that $N_p\neq \emptyset$ for all $p\in P$.

An (approval-based PB) \emph{instance} $I=(B,P,A,c)$ consists of
(i) a budget limit $B\in \mathbb{R}_{>0}$; 
(ii) a finite set of projects $P$;
(iii) an approval profile $A = (A_1, \dots, A_n)$; and
(iv) a cost function $c: P\rightarrow \mathbb{R}_{>0}$.
We assume that voters have \emph{cost satisfaction} functions \citep{BFL+23a}, i.e., voter~$i$'s satisfaction (or utility) from project $p$ is $\mu_i(p)=c(p)$ if $p\in A_i$ and $\mu_i(p)=0$ if $p\notin A_i$. 
For a subset of projects $P' \subseteq P$, we write $c(P')=\sum_{p\in P'} c(p)$ and $\mu_i(P')=\sum_{p\in P'} \mu_i(p)=c(P'\cap A_i)$. We define the \emph{utilitarian welfare} of a set of projects $P'\subseteq P$ as the sum of voter satisfaction, $uw(P')=\sum_{i\in N} \mu_i(P')=\sum_{p\in P'} |N_p|c(p)$. 
The utilitarian welfare of an outcome is a measure of its overall (utilitarian) efficiency.

Any subset of projects is called an \emph{outcome}. An outcome $P'\subseteq P$ is \emph{feasible} for instance $I = (B,P,A,c)$ if $c(P')\leq B$. For an outcome $P'$, we say that a project $p\in P\setminus P'$ is \emph{affordable} if $c(P')+c(p)\leq B$. An outcome is \emph{exhaustive} if there are no unchosen affordable projects.

\textit{Voting rules} map each PB instance to a feasible outcome. In order to facilitate the definition of mixed voting rules in \Cref{sec:mixed}, we define voting rules to take two additional inputs: a budget $B_R$ (that is upper bounded by the instance budget $B$ but can be strictly smaller) and a set $P_0$ of pre-selected projects that are required to be in the output of the rule. 

\begin{definition}
A \emph{voting rule} $R$ takes as an input 
(i) a PB instance $I=(B,P,A,c)$, 
(ii) a rule budget $B_R \in \mathbb{R}_{>0}$ with $B_R\leq B$, and 
(iii) a set $P_0 \subseteq P$ of pre-selected projects with $c(P_0) \le B_R$; 
it outputs an outcome $R(I,B_R,P_0)=P^*$ with $P_0 \subseteq P^*$ and $c(P^*)\le B_R$. 
\end{definition}

Whenever $P_0=\emptyset$ and $B_R=B$, this definition reduces to the standard definition of voting rules in the PB literature. When the instance $I$ is clear from the context, we often write $R(B_R, P_0)$ instead of $R(I, B_R, P_0)$. We say that a voting rule is \textit{exhaustive} if it always produces an exhaustive outcome (w.r.t.~$B_R$).

We introduce two voting rules that are widely used in the PB literature and real-world elections, describing them in the standard setting.

\paragraph{Greedy Rule.}

Given a budget $B$, $\alg{Greedy}(B,\emptyset)$ iteratively selects an affordable project $p$ with the largest number of supporters $|N_p|$, with arbitrary tie-breaking, until no affordable projects remain, greedily maximizing utilitarian welfare. 

\medskip

We measure the relative welfare of a particular outcome using the following notion.

\begin{definition}\label{def:utilitarian_ratio}
    The \emph{utilitarian ratio} of outcome  $P^*\subseteq P$ is given by $uw(P^*)/uw(\hat{P})$, where $\hat{P}$ is the feasible outcome maximizing utilitarian welfare.
\end{definition}

We say that $P^*\subseteq P$ is \emph{efficient up to one project} %
if there exists $p\in P\setminus P^*$ such that the utilitarian ratio of $P^*\cup \{p\}$ is at least $1$. Like the simplest greedy algorithm for the knapsack problem, \alg{Greedy} always produces an outcome that is efficient up to one project. %
When $P_0\neq \emptyset$ and $B_{\alg{G}}\leq B$, $\alg{Greedy}(B_{\alg{G}},P_0)$ can be defined in the exact same way, greedily maximizing welfare (given that $P_0$ must be included in its outcome) and checking project affordability with respect to $B_{\alg{G}}$. %

\paragraph{Method of Equal Shares (MES) \emph{\citep{PPS21a}}}
$\alg{MES}(B,\emptyset)$ assigns each voter $i\in N$ an initial budget of $b_i=\frac{B}{n}$ and iteratively selects projects as follows. Let $P^{(k-1)}$ be the set of projects chosen after step $k-1$ of \alg{MES}. During step $k$, for each affordable project $p \in P\setminus P^{(k-1)}$, we try to find $\rho(p)$ such that
$\sum_{i\in N_p} \min(b_i,\rho(p) c(p))=c(p)$.
We select $p_{k}=\arg \min \{\rho(p)\mid p\in P\setminus P^{(k-1)}\}$, with ties broken arbitrarily. This is the project that can be bought by its supporters while minimizing the maximum payment per unit satisfaction. We add $p_{k}$ to our selection %
and update the budgets of its supporters to $b_i-\min(b_i, \rho(p_{k})c(p_{k}))$ after every round. The algorithm terminates when no more projects can be afforded by their supporters, i.e., when $c(p)>\sum_{i\in N_p}b_i$ for all $p\in P\setminus P^{(k-1)}$.

\medskip

By giving voters equal budgets and allowing them to spend these on projects they approve of, \alg{MES} aims to make its outcome as proportional as possible. This has been formalized in the following proportionality notion, which the outcome of $\alg{MES}(B,\emptyset)$ always satisfies~\citep{BrPe23a}.

\begin{definition}\label{def:EJR+}
    An outcome $P^*\subseteq P$ satisfies \emph{EJR+ up to any project} if for every group of voters $N'\subseteq N$ and every project $p\in \bigcap_{i \in N'} A_i \setminus P^*$, there is a voter $i\in N'$ with
    $c(A_i\cap P^*)+c(p) > \frac{|N'|B}{n}$. 
\end{definition}

\section{Mixed Voting Rules}
\label{sec:mixed}

In this section, we formally introduce a general framework for combining voting rules sequentially. We let $m \in \mathbb{N}$ denote the number of constituent rules of a mixed rule $\mix =[R_1, R_2, \dots, R_m]$ and we use $[m]$ to denote the set $\{1,\dots,m\}$.

\begin{definition}
A \emph{mixed voting rule} $\mix=[R_k]_{k \in [m]}$ takes as an input 
(i)~a PB instance $I=(B,P,A,c)$, 
(ii)~a sequence of rule budgets $[B_k]_{k \in [m]}$, with each $B_k \in \mathbb{R}_{>0}$ and \mbox{$0< B_1\leq B_2\leq ...\leq B_m\leq B$}, and
(iii)~a set $P_0 \subseteq P$ of pre-selected projects with $c(P_0) \le B_1$; 
it outputs an outcome $P^* = \mix(I,[B_k]_{k \in [m]},P_0)$ with $P_0 \subseteq P^*$ and $c(P^*)\leq B_m$. 
The mixed rule $\mix$ produces its output $P^*$ using a series of intermediate outcomes $(P_k)_{k \in [m]}$, created by the rules it contains, with $P_0\subseteq P_1\subseteq\dots\subseteq P_{m-1}\subseteq P_m=P^*\subseteq P$. The rules are resolved in sequence: Each intermediate outcome is iteratively defined as $P_k=R_k(I,B_k,P_{k-1})$ and is used as the next rule's set of pre-selected projects. 
\end{definition}

\tikzstyle{standard} = [rectangle, minimum width=1cm, minimum height=1cm, text centered, draw=black, fill=teal!20]
\tikzstyle{etc} = [rectangle, minimum width=1cm, minimum height=1cm, text centered, draw=black, dashed, 
]
\tikzstyle{arrow} = [thick,->,>=stealth]
\tikzstyle{big} = [rectangle, minimum width=7.5cm, minimum height=2.5cm, text centered, draw=black, thick, fill=violet!10]
\begin{wrapfigure}[11]{R}{0.6\textwidth}
    \centering
    \begin{tikzpicture}[node distance=2cm,xscale=0.9,yscale=0.8, transform shape]
         \node (Mixed) [big, text depth = 1.5 cm]  {$\mix=[R_k]_{k \in [m]}$};
         \node (R1) [standard, left of=Mixed, xshift=-1cm, yshift=-0.5cm] {$R_1$};
         \node (R2) [standard, right of=R1] {$R_2$};
         \node (etc) [etc, right of=R2] {...};
         \node (Rm) [standard, right of=etc] {$R_m$};

         \draw [<-] (R1) -- (-4.5,-0.5) node[near end, above] {$P_0$};
         \draw [->] (R1) -- (R2) node[midway,above] {$P_1$};
         \draw [->] (R2) -- (etc) node[midway,above] {$P_2$};
         \draw [->] (etc) -- (Rm) node[midway,above] {$P_{m-1}$};
         \draw [->] (Rm) -- (4.5,-0.5) node[near end,above] {$P_{m}$};
         
         \draw [<-] (R1) -- (-3,-2) node[near end,left] {$B_1$};
         \draw [<-] (R2) -- (-1,-2) node[near end,left] {$B_2$};
         \draw [<-] (Rm) -- (3,-2) node[near end,left] {$B_m$};     
    \end{tikzpicture}
    
\caption{\centering Illustration of a mixed voting rule $\mix$ as a sequence of rules with inputs and outputs.}
\label{fig:mixed_rule_diagram}  
\end{wrapfigure}
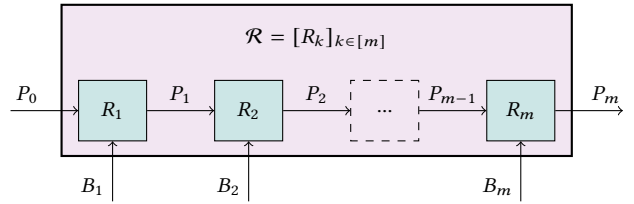

This process is illustrated in \Cref{fig:mixed_rule_diagram}.
When the instance $I$ is clear from the context, we often drop it from the notation and write $P^*=\mix([B_k],P_0)$.

We can think of a mixed rule $\mix([B_k]_{k},P_0)=[R_k]_{k}([B_k]_{k},P_0)$ as splitting up the instance budget among the voting rules it contains, giving rule $R_k$ a budget of at least $B_k-B_{k-1}$.
We allow each rule to use any budget left unspent by the previous rules, giving $R_k$ an \emph{available budget} of $B_k-c(P_{k-1})$. %

\begin{definition}
    Consider a mixed voting rule $\mix=[R_k]_{k \in [m]}$ applied to $(I,[B_k]_k,P_0)$ and let the outcome of $R_{k-1}$ be $P_{k-1}$. We define the \emph{available budget share} of $R_k$ during the execution of $\mix(I,[B_k]_{k \in [m]},P_0)$ as 
        $\alpha_k=\frac{B_k-c(P_{k-1})}{B}.$
\end{definition}

The available budget share $\alpha_k$ represents the proportion of the overall budget available to $R_k$ to spend on remaining projects. The available budget share of the first rule in a mixed rule is $\alpha_1=\frac{B_1-c(P_0)}{B}$. For $k>1$, $\alpha_k$ depends on the set of projects chosen by earlier rules in the mix, and is unaffected by later rules. The available budget share of $R_k$ is bounded from above by its ``gross'' budget share $\frac{B_k}{B}\geq \alpha_k$ (with equality if and only if $P_{k-1}=\emptyset$). We will refer to the available budget share as simply $\alpha$ when it is clear from the context which rule we are considering, writing "$R\in \mix$ with available budget share $\alpha$."

A subset\footnote{We only consider completion methods that add projects to the outcome of the rule they are completing, without modifying the already selected set.} of "completion methods" \citep{PPS21a} that are widely used in PB can be easily defined in the mixed voting rule framework. We call a rule $R_k \in \mix$ a \emph{completion rule} if $B_k=B_{k-1}$. That is, $R_k$ is not allocated any extra budget, but can only spend budget that was left over from the previous rule. 
For instance, `\alg{MES} completed by \alg{Greedy}' can be defined as a mixed rule, $[\alg{MES}, \alg{Greedy}]([B,B], \emptyset) = {\textproc{\alg{Greedy}}(B,\textproc{\alg{MES}}(B,\emptyset))}$.

\begin{wrapstuff}[type=table,width=.34\textwidth]
        \centering
        \scalebox{0.8}{
        \begin{tabular}{c c c c c c c}
            \toprule
             & $p_1$ & $p_2$ & $p_3$ & $p_4$ & $p_5$ & $p_6$ \\
            \midrule
            Cost  & 28 & 12 & 45 & 12 & 8 & 6 \\
            \midrule
            $A_1$ & \x &    &    & \x &    & \x \\
            $A_2$ & \x &    & \x & \x &    & \\
            $A_3$ & \x &    & \x & \x & \x \\
            $A_4$ & \x & \x & \x &    & \x \\
            $A_5$ &    & \x & \x &    &    \\
            \bottomrule
        \end{tabular}
        }
    \end{wrapstuff} 

\begin{example}\label{ex:mixed_method}
    Consider $\mix=[R_1,R_2,R_3]$ where $R_1=\alg{MES}$, $R_2=\alg{Greedy}$, and $R_3=\alg{Spend}$ is a rule that picks the set of projects that maximizes total spending. Let the instance budget be $B=100$, the rule budgets be $[B_1,B_2,B_3]=[50,90,100]$, and the project set be $P=\{p_1,\ldots,p_6\}$. The cost function and approval profile are given in the table on the right. %

    \smallskip
    
    We compute $\mix([50,90,100], \emptyset)$ %
    in three steps:

\begin{enumerate}[label={(\arabic*)}]
    \item $\alg{MES}(50,\emptyset)$ with available budget $B_1-c(\emptyset) = 50$ and available budget share $\alpha_\alg{MES} = 0.5$ selects set {$P_1=\{p_1,p_2\}$} and terminates, as the remaining voter budgets 
    (i.e., $(3,3,3,0,1)$)
    are not sufficient to afford any of the other projects. Note that $B_1-c(P_1)=10$ units of budget are left unspent.

    \item $\alg{Greedy}(90,P_1)$ with available budget $B_2-c(P_1)=50$ %
    and available budget share $\alpha_{\alg{Greedy}}=0.5$ selects $p_3$ and terminates with outcome $P_2=P_1\cup \{p_3\}$, as no other projects can be afforded with \alg{Greedy}'s remaining budget of $B_2-c(P_2)=5$.

    \item$\alg{Spend}(100,P_2)$ with available budget $B_3-c(P_2)=15$ and available budget share $\alpha_{\alg{Spend}}=\frac{B_3-c(P_2)}{100}=0.15$ selects $P_3=P_2\cup \{p_5,p_6\}$, which is the final outcome of our mixed rule $\mix$. 
    \hfill $\diamond$
\end{enumerate}
\end{example}

We will analyze to what extent a mixed voting rule inherits the properties satisfied by its constituent rules. The following definition applies to a wide range of axiomatic~properties.

\begin{definition}\label{def:alpha-budget}
    Let $\mathcal{I}_P$ be the set of all instances with project set $P$. A \emph{monotone property} is a function $X: \mathcal{I}_P \times 2^P \rightarrow \{0,1\}$ such that, for any $I \in \mathcal{I}_P$, 
    $P'\subseteq P''\subseteq P$ implies $X(I, P')\leq X(I, P'')$.
    We say that a set of projects $P^*\subseteq P$ \emph{satisfies} $X$ for instance $I=(B,P,A,c)$ if $X(I,P^*)=1$. For $\alpha\in \mathbb{R}_{>0}$, we say that $P^*$ instead satisfies \emph{${\alpha}$-budget} $X$ for instance $I=(B,P,A,c)$ if $X((\alpha B,P,A,c),P^*)=1$.
\end{definition}

We can use \Cref{def:alpha-budget} to capture the monotone properties \textit{efficiency up to one project} and \textit{$EJR+$ up to any project}, as described in \Cref{sec:preliminaries}.\footnote{\Cref{def:alpha-budget} is tailored to \textit{intra-profile} properties of outcomes (i.e., properties that can be evaluated given a single instance and outcome). \textit{Inter-profile} properties such as strategyproofness fall outside its scope.}
Note that project set $P^*$ need not be a feasible outcome for the instance $(\alpha B,P,A,c)$.

Every voting rule from the standard setting can be adapted to the mixed framework in a trivial way, by ignoring the pre-selected projects and setting $R(I,B,P_0) = P_0 \cup R(I,B-c(P_0),\emptyset)$. In other words, the rule runs on the full project set with the remaining budget, and its output is combined with the pre-selected projects. Note that this differs from the perhaps more natural approach of removing pre-selected projects from the instance; in particular, the rule may re-select projects that are already in $P_0$, potentially leaving large portions of the budget unspent. This ``trivial adaptation'' allows us to establish the following baseline result. %

\begin{restatable}{observation}{obsone}\label{obs:1}
    Suppose there exists a voting rule $R$ whose outcome $P^*=R(B,\emptyset)$ always satisfies some monotone property $X$. Then, we can construct a voting rule $R'$ such that if $R'\in \mix$ with available budget share $\alpha$, the outcome of $\mix$ always satisfies $\alpha$-budget $X$.
\end{restatable}

\Cref{obs:1} helps establish a \emph{theoretical baseline} for the axiomatic results of this paper. When we adapt a rule $R$ whose outcome always satisfies property $X$ to the mixed framework, it is desirable that, when $R\in \mix$ with an available budget share of $\alpha$, the outcome of $\mix$ satisfies $\alpha$-budget $X$. %

However, exceeding the theoretical baseline, i.e., always achieving $\alpha'$-budget $X$ for some $\alpha'>\alpha$, is not possible for the properties we have introduced so far. We can easily show that no adaptation of \alg{Greedy} or \alg{MES} can, in general, perform better than their corresponding baseline. 

\begin{restatable}{proposition}{ingeneralbaseline}\label{prop:in_general_baseline}
    Let property $X$ be either "efficiency up to one project" or "EJR+ up to any project".
    There exists no rule $\alg{R}$ such that when $\alg{R}\in \mix$ with available budget share $\alpha$, the outcome of $\mix$ satisfies $\alpha'$-budget $X$ for any $\alpha'>\alpha$, for all instances and sets of pre-selected projects.
\end{restatable}

We argue this in \Cref{app:omitted_proofs} by considering an instance in which the pre-selected set $P_0$ consists of a single unpopular project.

Given the impossibility from \Cref{prop:in_general_baseline}, our goal in adapting \alg{MES} to the mixed rules framework will be twofold. 
We want our adaptation to closely resemble the original rule while taking into account the pre-selected projects, rather than ignoring them as in the trivial adaptation discussed above. At the same time, we aim to match the theoretical baseline from \Cref{obs:1}.
We discuss several possible ways to do so in \Cref{sec:adapting_mes}. In \Cref{sec:proportionality}, we will consider a measure of proportionality as a function of the pre-selected set, and consider how our adaptations perform with respect to this measure.

Most of our adaptations in this paper deal with adjusting a voting rule to account for projects that have already been selected by earlier rules in the sequence. For \alg{Greedy}, this is straightforward, as discussed above. However, the mixed rules framework also opens up a different possibility: adapting a rule in light of knowing that other rules will follow. In the standard setting, an exhaustive rule like \alg{Greedy} spends as much of the budget as possible, since any unspent funds go to waste. In a mixed rule, however, unspent budget is passed on to later rules. This motivates a variant we call \emph{\alg{Greedy} with early stopping}, which terminates as soon as the most supported unselected project is not affordable, rather than continuing to select increasingly inefficient projects. Despite potentially generating lower utilitarian welfare, we can show that both the standard \alg{Greedy} rule and \alg{Greedy} with early stopping meet our theoretical baseline with respect to efficiency up to one project.

\begin{restatable}{proposition}{greedyefficiency}\label{prop:greedy_efficiency}
    Both variants of \alg{Greedy} achieve $\alpha$-efficiency up to one project when provided with a budget share of $\alpha$.
\end{restatable}

We use the standard variant as our default, but consider both in our experiments (see \Cref{sec:exp_early_stopping}).

\section{The Method of Equal Shares in the Mixed Rules Framework} \label{sec:adapting_mes}

In this section we generalize \alg{MES} to account for a set of pre-selected projects $P_0$. We define several pre-allocation methods to account for the differences in voter satisfaction from $P_0$, each of which provides a different profile of initial voter budgets to \alg{MES}.

Consider $\alg{MES}(B_{\alg{MES}},P_0)$, with $P_0\neq \emptyset$. Our goal is to determine how to initialize voter budgets $b_i$ for each voter $i \in N$ to account for %
$P_0$, in order to maximize the proportionality achieved by subsequently running \alg{MES}. We might no longer desire to equally split \alg{MES}'s available budget, $\alpha B = B_\alg{MES} - c(P_0)$, as the set of pre-selected projects $P_0$ need not be equally liked by all voters. Intuitively, voters that are more satisfied with $P_0$ should be provided with a smaller individual budget.

We formalize methods for determining voter budgets $(b_i)_{i\in N}$ (with $\sum b_i=\alpha B$) as \emph{pre-allocation methods}, and write $\alg{MES}^{M}$ to refer to \alg{MES} with pre-allocation method~$\pay{M}$. All our pre-allocation methods follow a two-step process.

\begin{definition}\label{def:rebalancing_process}

A \emph{pre-allocation method} takes as input an instance $I$, a set $P_0\subseteq P$ of pre-selected projects, and an available budget share $\alpha$,\footnote{This information can be used to find $B_{MES}=c(P_0)+\alpha B$.} and proceeds in two stages:
\begin{enumerate}[label={(\arabic*)}]
    \item It determines \emph{voter payments} $(\pi_i)_{i\in N}$ for projects from $P_0$, such that $\pi_i\geq 0$ for all $i\in N$ and $\sum_{i\in N}\pi_i \leq c(P_0)$.
    
    \item It applies a \emph{rebalancing step} to determine
    \emph{voter budgets} $(b_i)_{i\in N}$ in order to make voters' endowments $\{\pi_i+b_i\}_{i\in N}$ as equal as possible.\footnote{This can be thought of as pouring $\alpha B$ square units of water into a 2D bucket (see \Cref{fig:pre_allocation_example}), where the floor (hatched, violet) is a bar chart, where the height of each bar is $\pi_i$ (and the width is $1$), and the water (filled, blue) represents $(b_i)_{i \in N}$.}
    Formally, $(b_i)_{i\in N}$ are chosen to maximize $\min_{ i\in N} (\pi_i+b_i)$ under the constraints $\sum_{i\in N} b_i=\alpha B$ and $b_i\geq 0$ for all $i \in N$.%
\end{enumerate}
\end{definition}

\noindent Note that we allow pre-allocation methods to have some voters pay more than their ``fair share'' of the \alg{MES} budget, i.e., $\pi_i > \frac{B_{\alg{MES}}}{n}$.
It is also possible for the voter payments to only partially fund the projects in $P_0$, or not fund them at all. The rebalancing step (2) is the same for all pre-allocation methods. To motivate this, it can be shown that every ``reasonable'' voter budget profile $(b_i)_{i\in N}$ can be induced by picking appropriate voter payments
in stage (1)\,---\,see \Cref{claim:rebalancing_induces_every_profile} in \Cref{app:omitted_proofs} for details.

\smallskip

We now define four pre-allocation methods for $\alg{MES}(B_{\alg{MES}},P_0)$ with available budget share $\alpha$, each following the two-stage process in \Cref{def:rebalancing_process}: \pay{Null}, \pay{MES-Style}, \pay{Equal-Split}, and \pay{Value-Based}.
Alongside, we present a running example, %
with the results of our methods illustrated in \Cref{fig:pre_allocation_example}.
All four methods can be computed in polynomial time. %

\begin{example}\label{ex:pre-allocation}
    Consider an instance with $B = 32$, $P=\{p_1,p_2,p_3\}$, and approval profile and project costs as specified in \Cref{fig:instance}. 
    We assume that \alg{MES} is given a budget $B_{\alg{MES}} = 32$, and a pre-selected project set $P_0 = \{p_1, p_2\}$, resulting in an \alg{MES} budget share of $\alpha=\frac{B_{\alg{MES}}-c(P_0)}{B}=0.25$.
\end{example}

The outcomes of the pre-allocation methods for \Cref{ex:pre-allocation} are illustrated in \Cref{fig:pre_allocation_example}.

\pgfplotsset{
/pgfplots/line_legend/.style={
legend image code/.code={
\draw[ultra thick,red](-0.1cm,0cm) -- (0.1cm,0cm);%
   }
  }
}
\begin{method}\label{method:null}
    \textbf{\underline{\pay{Null}}}: 
    Set $\pi_i=0$ for all voters, obtaining voter budgets $b_i=\frac{\alpha B}{n}$ using the rebalancing step. This pre-allocation method splits the remaining budget $\alpha B$ equally among all voters. \hfill $\diamond$
\end{method}

    In \Cref{ex:pre-allocation}, each voter $i\in N$ pays $\pi_i=0$ for $P_0$ and gets an \alg{MES} budget of $b_i=2$. 

\begin{method}\label{method:mes-style}
    \textbf{\underline{\pay{MES-Style}}}: 
    In order to determine voter payments, choose any order for $P_0$ and initialize voter budgets
    to $b^0_i=\frac{B_{\alg{MES}}}{n}$. Iteratively fund the projects in $P_0$ as if they were sequentially selected by \alg{MES}. If at any point the voters in $N_p$ cannot fully fund some project $p\in P_0$, the remaining cost is discarded.\footnote{Alternatively, we could fund such projects by overcharging their supporters; this would result in the same voter budgets.}  %
    Our theoretical result (\Cref{prop:mes_style_fails_EJRk}) holds regardless of the order. For \Cref{ex:pre-allocation} and our experiments in \Cref{sec:experiments}, we order $P_0$ by the number of supporters. \hfill $\diamond$
\end{method}

\tikzstyle{payments1}=[ybar, violet, pattern={mylines[size=5.4pt,line width=4.7pt,angle=0]}, pattern color=violet!20]
\tikzstyle{payments2}=[ybar, violet, pattern={mylines[size=5.4pt,line width=4.7pt,angle=45]}, pattern color=violet!20]
\tikzstyle{budgets}=[ybar, blue, fill=blue!20, ]

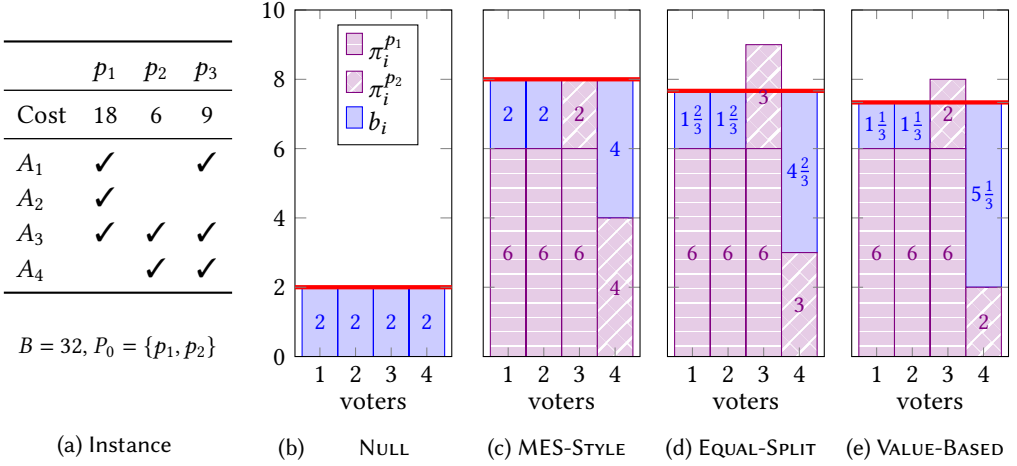
\begin{figure}[t]
\begin{subfigure}[t]{0.25\textwidth}
    \centering
    \scalebox{1}{   
        \renewcommand{\arraystretch}{1.1}
        \begin{tabular}{l  c c c}
            \toprule
             & $p_1$ & $p_2$ & $p_3$ \\
            \midrule
            Cost  & 18 & 6 & 9 \\
            \midrule
            $A_1$ & \x &    & \x  \\
            $A_2$ & \x &    &     \\
            $A_3$ & \x & \x & \x  \\
            $A_4$ &    & \x & \x  \\
            \bottomrule
            & & & \\
            \multicolumn{4}{c}{\small $B = 32$, $P_0 = \{p_1, p_2\}$}\\[0.69cm]
        \end{tabular}
    }
    \caption{Instance}
    \label{fig:instance}
\end{subfigure} \hfill
    \begin{minipage}{0.73\textwidth}
        \hspace{-2mm}\begin{subfigure}[b]{.24\textwidth}
        \begin{tikzpicture}
            \begin{axis}[
                    ybar stacked,
                    yscale=0.8,
                    xtick=data,
                    ymin=0,
                    ymax=10,
                    xlabel=voters,
                    bar width=1,
                    enlarge x limits={abs=.7},
                    nodes near coords,
                    every node near coord/.append style={
                        font=\small,
                        /pgf/number format/.cd, 
                        precision=1,
                        frac,
                        frac shift=1,
                    },
                    width=1.5\textwidth,
                    height=3\textwidth,
                    extra y ticks = 2,
                    extra y tick labels={},
                    extra y tick style={grid=major,major grid style={ultra thick,draw=red}},
                    axis on top,
                    legend style={anchor=east, xshift=-0.4cm, yshift=-0.0cm},
                    legend cell align={left}
            ]
            
            \addplot+ [payments1] coordinates {(1,0) (2,0) (3,0) (4,0)};
            \addplot+ [payments2] coordinates {(1,0) (2,0) (3,0) (4,0)};
            \addplot+ [budgets] coordinates {(1,2) (2,2) (3,2) (4,2)};
            \legend{$\pi_i^{p_1}$,$\pi_i^{p_2}$,$b_i$}
            \end{axis}
        \end{tikzpicture}
        \caption{\ \ \ \ \ \ \ \ \pay{Null}}
        \end{subfigure}\hspace{3.8mm}%
        \begin{subfigure}[b]{.24\textwidth}
        \begin{tikzpicture}
            \begin{axis}[
                    ybar stacked,
                    yscale=0.8,
                    xtick=data,
                    ymin=0,
                    ymax=10,
                    yticklabel=\empty,
                    xlabel=voters,
                    bar width=1,
                    enlarge x limits={abs=.7},
                    nodes near coords,
                    every node near coord/.append style={
                        font=\small,
                        /pgf/number format/.cd, 
                        precision=1,
                        frac,
                        frac shift=1,
                    },
                    width=1.5\textwidth,
                    height=3\textwidth,
                    extra y ticks = 8,
                    extra y tick labels={},
                    extra y tick style={grid=major,major grid style={ultra thick,draw=red}},
                    axis on top
            ]
            \addplot+ [payments1] coordinates {(1,6) (2,6) (3,6) (4,0)};
            \addplot+ [payments2] coordinates {(1,0) (2,0) (3,2) (4,4)};
            \addplot+ [budgets] coordinates {(1,2)  (2,2) (3,0) (4,4)};
            \end{axis}
        \end{tikzpicture}
        \caption{\pay{MES-Style}
        }
        \end{subfigure}%
        \begin{subfigure}[b]{.24\textwidth}
        \begin{tikzpicture}
            \begin{axis}[
                    ybar stacked,
                    yscale=0.8,
                    xtick=data,
                    ymin=0,
                    ymax=10,
                    yticklabel=\empty,
                    xlabel=voters,
                    bar width=1,
                    enlarge x limits={abs=.7},
                    nodes near coords,
                    every node near coord/.append style={
                        font=\small,
                        /pgf/number format/.cd, 
                        precision=1,
                        frac,
                        frac shift=1,
                    },
                    width=1.5\textwidth,
                    height=3\textwidth,
                    extra y ticks = 23/3,
                    extra y tick labels={},
                    extra y tick style={grid=major,major grid style={ultra thick,draw=red}},
                    axis on top
            ]
            \addplot+ [payments1] coordinates {(1,6) (2,6) (3,6) (4,0)};
            \addplot+ [payments2] coordinates {(1,0) (2,0) (3,3) (4,3)};
            \addplot+ [budgets] coordinates {(1,5/3)  (2,5/3) (3,0) (4,14/3)};
            \end{axis}
        \end{tikzpicture}
        \caption{\pay{Equal-Split}}
        \end{subfigure}%
        \begin{subfigure}[b]{.24\textwidth}
        \begin{tikzpicture}
            \begin{axis}[
                    ybar stacked,
                    yscale=0.8,
                    xtick=data,
                    ymin=0,
                    ymax=10,
                    yticklabel=\empty,
                    xlabel=voters,
                    bar width=1,
                    enlarge x limits={abs=.7},
                    nodes near coords,
                    every node near coord/.append style={
                        font=\small,
                        /pgf/number format/.cd, 
                        precision=1,
                        frac,
                        frac shift=1,
                    },
                    width=1.5\textwidth,
                    height=3\textwidth,
                    extra y ticks = 22/3,
                    extra y tick labels={},
                    extra y tick style={grid=major,major grid style={ultra thick,draw=red}},
                    axis on top
            ]
            \addplot+ [payments1] coordinates {(1,6) (2,6) (3,6) (4,0)};
            \addplot+ [payments2] coordinates {(1,0) (2,0) (3,2) (4,2)};
            \addplot+ [budgets] coordinates {(1,4/3)  (2,4/3) (3,0) (4,16/3)};
            \end{axis}
        \end{tikzpicture}
        \caption{\pay{Value-Based}}
        \end{subfigure}
    \end{minipage}
    \caption{PB instance and pre-allocation outcomes for \Cref{ex:pre-allocation}. Voter payments $\pi_i$ are calculated separately by each pre-allocation method (and are split by project in the diagrams, with $\pi_i=\pi_i^{p_1}+\pi_i^{p_2}$), and the voter budgets $b_i$ are obtained from the rebalancing step. The \textcolor{red}{{red line}} represents the minimum voter endowment $\min_{i\in N}(\pi_i+b_i)$. \pay{MES-Style} uses project order $[p_1,p_2]$.
    }
    \label{fig:pre_allocation_example}
    
\end{figure}
    
In \Cref{ex:pre-allocation}, with project order $[p_1,p_2]$, $p_1$ is funded first and the first 3 voters each pay~6 for it, obtaining intermediate budgets of $b'_1=b'_2=b'_3=2$. When funding $p_2$, voter $3$ can no longer pay their fair share, so they pay as much as they can, with the rest covered by $4$. 
We obtain payments $(\pi_i)_{i \in [4]}=(6,6,8,4)$
and budgets of $(b_i)_{i \in [4]}=(2,2,0,4)$.
With the reverse project order we instead obtain $(\pi_i)_{i \in [4]}=(6.5,6.5,8,3)$
and $(b_i)_{i \in [4]}=(1.5,1.5,0,5)$.

\begin{method}\label{method:equal-split}
    \textbf{\underline{\pay{Equal-Split}}}:
    Split the total cost of every project $p\in P_0$ equally among its supporters~$N_p$, setting $\pi_i=\sum_{p \in A_i \cap P_0} \frac{c(p)}{|N_p|}$ and derive the voter budgets using the rebalancing step. Note that $\sum_{i\in N} \pi_i=c(P_0)$, i.e., the voter payments fund the pre-selected projects completely.
    \hfill $\diamond$
\end{method}
    
    In \Cref{ex:pre-allocation}, supporters pay $6$ each for $p_1$ and $3$ each for $p_2$, with total payments $(\pi_i)_{i \in [4]}=(6,6,9,3)$. Voter $3$ has spent more than their fair share. From the rebalancing step, we get 
    $(b_i)_{i \in [4]}=(\frac{5}{3},\frac{5}{3},0,\frac{14}{3})$. \\

    The \pay{MES-Style} and \pay{Equal-Split} methods represent fairly natural ways to divide the cost of $P_0$ among voters. However, we will see in \Cref{sec:proportionality} that they do not achieve proportionality at a level dictated by our theoretical baseline, while the (trivial) \pay{Null} method does. We introduce one more method that works similarly to \pay{Equal-Split}, but achieves better proportionality guarantees. The idea behind this method is to partially fund the projects in $P_0$ from voter payments $\pi_i$ in such a way that voters get good value for money whenever they contribute to a project.

\begin{method}\label{method:value-based}
    \textbf{\underline{\pay{Value-Based}}}:
    We define the utilitarian \emph{value}
    (for money) $v(p)$ of a project $p$ to be $v(p)=|N_p|$. This is the ratio of the utilitarian welfare of a project, $uw(p)=|N_p|c(p)$, to its cost.
    For a set of pre-selected projects $P_0$, we define the \emph{threshold value} $v^*$ as the value of the most valuable unselected project that is affordable under budget $B_{\alg{MES}}$ after every project from $P_0$ with greater value has been selected. Formally,
    $$v^*=\max_{p\in P\setminus P_0} \{v(p) \mid p \text{ satisfies } c(U(p))+c(p)\leq B_{\alg{MES}}\},$$
    where $U(p) = \{p'\in P_0 \mid v(p')\geq v(p)\}$ 
    denotes the upper contour set of $p$ in $P_0$, i.e., the set of pre-selected projects with value at least $v(p)$.
    If no project $p \in P\setminus P_0$ satisfies $c(U(p))+c(p)\leq B_{\alg{MES}}$, we let $v^*=0$.  
    In the case that all projects have distinct values, the threshold value $v^*$ is the value of the first project from $P\setminus P_0$ that $\alg{Greedy}(B_{\alg{MES}},\emptyset)$ would select.

    When funding some project $p^*$ with value $v^*$, voters in $N_{p^*}$ would each pay $\frac{c(p^*)}{v^*}$, or equivalently $\frac{1}{v^*}$ per unit satisfaction they obtain from $p^*$. The idea of the \pay{Value-Based} pre-allocation method is to allow voters to spend \textit{at most} $\frac{1}{v^*}$ per unit satisfaction, defining voter payments as follows:
    $$\pi_i=\sum_{p\in A_i \cap P_0} \frac{c(p)}{\max\{v(p),v^*\}} \, .$$

    Thus, any project $p\in P_0$ with value $v(p)\geq v^*$ is fully funded, identically to the \pay{Equal-Split} method, and any project with $v(p)<v^*$ is partially funded.
    \hfill $\diamond$

\end{method}
    
In \Cref{ex:pre-allocation}, supporters pay $6$ each for $p_1$, like they did under the \pay{Equal-Split} method. The threshold value in this instance is $v^*=3$, corresponding to the value of $p_3$. Thus, the \pay{Value-Based} method allows voters $3$ and $4$ to partially fund $p_2$, each paying $\frac{1}{3}\cdot 6=2$ for it, resulting in total payments of $(\pi_i)_{i \in [4]}=(6,6,8,2)$. Using the rebalancing step, we get $(b_i)_{i \in [4]}=\big(\frac{4}{3},\frac{4}{3},0,\frac{16}{3}\big)$. We briefly discuss a potential extension to the \pay{Value-Based} method in \Cref{app:value_based} which uses a voter-specific notion of threshold value. \\

Arguably, the four methods outlined above choose voter payments $\pi_i$ in a reasonable way, from the perspective of the voters. In particular, they never force a voter to pay for a project they do not approve, and never require a group of supporters to pay for more than the cost of a project.  %

\section{Proportionality for \alg{MES} Variants}\label{sec:proportionality}

In this section, we study the proportionality of \alg{MES} in the mixed framework, for each of the four pre-allocation methods defined in \Cref{sec:adapting_mes}. We consider parameterized versions of the proportionality axiom \textit{EJR+ up to any project} (\Cref{def:EJR+}) and weakenings of it. 
The axioms we consider are parameterized using the \textit{minimum voter budget share}, an important value that can be calculated from the output of a pre-allocation method. Notably, this creates a non-standard approach where the strength of our proportionality guarantees cannot be determined until the mixed rule is partially executed and we observe which projects were selected by earlier rules. Such an approach is necessary in light of the negative result of \Cref{prop:in_general_baseline}. We identify the \pay{Value-Based} pre-allocation method as the sole method that improves on our theoretical baseline from \Cref{sec:mixed}.

We define a parameterized version of EJR+ up to any project from \Cref{def:EJR+}.

\begin{definition}\label{def:EJR+xk}
    Let $k \in \mathbb{N}^+$ and $\alpha \in [0,1]$.
    An outcome $P^*\subseteq P$ satisfies \emph{$\alpha$-budget EJR+ up to any $k$ projects} if for every group of voters $N'\subseteq N$ and every set of projects $P'\subseteq \bigcap_{i \in N'} A_i \setminus P^*$ with $|P'|=k$, there is a voter $i\in N'$ with
    $c(A_i\cap P^*)+c(P')> \alpha \frac{|N'| B}{n}$.
\end{definition}

For $k=1$, this is the direct result of applying \Cref{def:alpha-budget} to \textit{EJR+ up to any project}. This property becomes weaker for larger values of $k$, and for smaller values of $\alpha$. Similar ``up to $k$''-style notions have been defined in the fair division literature \citep{ARS22a}.

The normative goal of the pre-allocation methods we defined in \Cref{sec:adapting_mes} is to select the profile of voter payments $(\pi_i)_{i\in N}$ in such a way that each voter gets good "value for money" whenever they contribute to a project in $P_0$. We can think of $\pi_i+b_i$ as the \textit{total endowment} of the voter, which the rebalancing step tries to make as equal as possible among the set of voters. 
The following definition focuses on a voter with minimal total endowment and compares their total endowment to a voter's fair share of the instance budget, $\frac{B}{n}$. 

\begin{definition}\label{def:minimum_voter_budget_share}
    Consider an MES pre-allocation method $\pay{M}$ that is applied as part of a mixed rule.
    The \emph{minimum voter budget share} $\alpha^{\pay{M}}$ is defined as $\alpha^{\pay{M}}=\min_{i\in N}\{(\pi_i+b_i)\frac{n}{B}\}$.
\end{definition}

The minimum voter budget share represents how much the worst off voter gets to spend on (1) projects in $P_0$ and (2) during the execution of $\alg{MES}^{\pay{M}}$, as a fraction of their fair share of the instance budget. A method \pay{M} with a higher value of $\alpha^\pay{M}$ is not necessarily more proportional as it is possible for \pay{M} to spend voter budgets inefficiently on projects from $P_0$ (see the negative results in \Cref{thm:violations}).

The rebalancing step (see \Cref{def:rebalancing_process}) places some constraints on the possible values of $\alpha^\pay{M}$.

\begin{restatable}{observation}{alphav}\label{claim:alpha_v}
   For any pre-allocation method $\pay{M}$, $\alpha^{\pay{M}}$ lies between \alg{MES}'s available budget share $\alpha$, and the proportion of the instance budget $B$ allocated to \alg{MES}: $\frac{B_{\alg{MES}}-c(P_0)}{B}=\alpha\leq \alpha^{\pay{M}}\leq \frac{B_{\alg{MES}}}{B}$.
\end{restatable}

Before considering proportionality guarantees for each of our pre-allocation methods, we derive the following relationship between their minimum voter budget shares. 

\begin{restatable}{proposition}{alphavrelationship}\label{prop:min_voter_budget_share}
Fix an instance $I$ and a set of pre-selected projects $P_0$. The minimum voter budget shares for $\alg{MES}^{\pay{M}}(B_{\alg{MES}},P_0)$ with $M \in \{\pay{Null}, \pay{MES-Style},$ 
$\pay{Equal-Split}, \pay{Value-Based}\}$ satisfy
$$\alpha=\alpha^{\pay{Null}}\leq  \alpha^{\pay{Value-Based}}\leq  \alpha^{\pay{Equal-Split}}  \leq  \alpha^{\pay{MES-Style}}.$$
\end{restatable}

It is important to note that $\alpha^{\pay{Null}}< \alpha^{\pay{Value-Based}}$ whenever $P_0\neq \emptyset$, and that the gap between those two values can be quite large in practice, while the other gaps are typically smaller.\footnote{For instance, when running $[\alg{Greedy},\alg{MES}]([0.5B, B])$ on the real-world PB instances considered in \Cref{sec:experiments}, our pre-allocation methods obtain the following minimum voter budget shares on average: $\alpha^{\pay{Null}}=0.50$,  $\alpha^{\pay{Value-Based}}=0.92$, $\alpha^{\pay{Equal-Split}}=0.94$, and $\alpha^{\pay{MES-Style}}=0.97$. See \Cref{app:alpha_v_in_practice} for a more detailed overview of these values in practice. \label{fn:alpha-values}}

We will now consider which of our pre-allocation methods \alg{M} guarantee that the outcome of the voting rule $\alg{MES}^{\pay{M}}$ satisfies a proportionality property of the form described in \Cref{def:EJR+xk}.
These guarantees will be parameterized by the minimum voter budget share $\alpha^{\pay{M}}$ corresponding to \alg{M}.

Let $P^*=\alg{MES}^{\pay{M}}(B_{\alg{MES}},P_0)$ be the outcome of $\alg{MES}^{\pay{M}}$ when it is provided with rule budget $B_{\alg{MES}}$ and set of pre-selected projects $P_0$. Any proportionality guarantee that we can prove for $P^*$ directly translates to the same guarantee for the outcome of any mixed rule $\mix$ with $R_k=\alg{MES}^{\pay{M}}$, provided that the input $P_{k-1}$ for $R_k$ equals $P_0$. This is because $P^*$ will be contained in the output of $\mix$.

We state our proportionality guarantees formally for mixed rules $\mix$ containing $\alg{MES}^\pay{M}$. Note that the parameter $\alpha^\pay{M}$ of the resulting guarantees cannot be directly inferred from the inputs to $\mix$, but depends on the partial outcome provided to $\alg{MES}^\pay{M}$ during the execution of $\mix$. That is, we first need to partially run the rule $\mix$ before determining how good a proportionality guarantee on its outcome we can give. While this might sound like a disadvantage, we remind the reader of our empirical observation that the values of $\alpha^\pay{M}$ are often close to $1$ in realistic scenarios (see \Cref{fn:alpha-values}).

We start with a guarantee for $\alg{MES}^\pay{Null}$, which is a straightforward extension of existing results~(Proposition 19 of \citet{BrPe23a}).

\begin{restatable}{theorem}{nullEJR}\label{thm:null_EJR}
    Consider a mixed voting rule $\mix$ such that $\alg{MES}^\pay{Null}\in \mix$ with available budget share $\alpha$.
    Then, the outcome of $\mix$ satisfies $\alpha$-budget EJR+ up to any project.
\end{restatable}

\begin{proof}

    Let $P^*$ be the outcome of $\alg{MES}^\pay{Null}(B_{\alg{MES}},P_0)$. It is sufficient for us to show that $P^*$ satisfies $\alpha$-budget EJR+ up to any project, as $P^*$ will be contained in the outcome of $\mix$. Assume for a contradiction that $P^*$ does not satisfy $\alpha$-budget EJR+ up to any project. Then, there exists $p\notin P^*$ and voter set $N'\subseteq N_p$ with $c(A_i \cap P^*)+c(p)\leq \frac{|N'|}{n}\alpha B$ for all $i \in N'$.

    Since $p\notin P^*$, we know that this project was not affordable when $\alg{MES}^\pay{Null}$ terminated, and thus the remaining \alg{MES} budgets $b_i^r$ for voters in $N'$ satisfy $\sum_{i\in N'} b_i^r<c(p)$. Therefore, we get that for projects from~$P^*$,
    $$ 
    \frac{\text{spending by voters in } N'}{\text{satisfaction of voters in } N'}=\frac{\sum_{i\in N'}(\frac{\alpha B}{n}-b_i^r)}{\sum_{i\in N'} c(A_i\cap P^*)}>\frac{|N'|\frac{\alpha B}{n}-c(p)}{|N'|(|N'|\frac{\alpha B}{n}-c(p))}=\frac{1}{|N'|}.
    $$
    
    Hence, during the execution of $\alg{MES}^\pay{Null}$ at least one voter has to pay more than $\frac{1}{|N'|}$ per unit satisfaction they received, for a project from $P^*$. This means that at least one project $p'$ with $\rho(p')>\frac{1}{|N'|}$ was selected by $\alg{MES}^\pay{Null}$. Just before the first such project was selected, each voter $i\in N'$ must have spent at most $\frac{c(A_i\cap P^*)}{|N'|}\leq \frac{\alpha B}{n}-\frac{c(p)}{|N'|}$ during the execution of $\alg{MES}^\pay{Null}$, and thus $p$ must have been affordable with $\rho(p)\leq \frac{1}{|N'|}$ at that point. But then it should have been selected over $p'$, leading to a contradiction.
\end{proof}

Thus, $\alg{MES}^\pay{Null}$ meets\,---\,but does not exceed\,---\,our theoretical baseline from \Cref{sec:mixed}. (It does so directly, without following the ``trivial adaptation'' approach from \Cref{obs:1}.) 
However, we can improve upon our baseline by employing the \pay{Value-Based} pre-allocation method. 

\begin{restatable}{theorem}{valuebasedEJR}\label{thm:value_based_EJR}
    Consider a mixed voting rule $\mix$ such that $\alg{MES}^\pay{Value-Based}\in \mix$ with available budget share $\alpha$ and minimum voter budget share $\alpha^{\pay{Value-Based}}\geq \alpha$. Then the outcome of $\mix$ satisfies $\alpha^{\pay{Value-Based}}$-budget EJR+ up to any project.
\end{restatable}

\Cref{thm:value_based_EJR} is the main theoretical result of our paper. Its proof heavily relies on the fact that voter spending per unit of utility is bounded in the \pay{Value-Based} pre-allocation.

\begin{proof}

Let $\alpha_v=\alpha^{\pay{Value-Based}}$. We let the outcome of $\alg{MES}^\pay{Value-Based}(B_{\alg{MES}},P_0)$ be $P^*$ and assume for contradiction that there exist $p\notin P^*$ and voter set $N'\subseteq N_p$ with $c(A_i \cap P^*)+c(p)\leq \frac{|N'|}{n}\alpha_v B$ for all $i \in N'$. Let $v^*$ be the threshold value from the definition of the \pay{Value-Based} pre-allocation method. We distinguish two cases.

\paragraph{\textbf{Case 1:}} $|N_p|\leq v^*$. 

    From the definition of minimum voter budget share, we know that for each voter $i\in N$, $\pi_i+b_i\geq \frac{\alpha_v B}{n}$.
    Analogously to the proof of \Cref{thm:null_EJR}, we know that for projects from $P^*$,
    \[ 
    \frac{\text{spending by voters in } N' \ (\pi_i \text{ and } b_i)}{\text{satisfaction of voters in } N'}\geq\frac{\sum_{i\in N'}(\frac{\alpha_v B}{n}-b_i^r)}{\sum_{i\in N'} c(A_i\cap P^*)}> \frac{1}{|N'|} \text.
    \]

    Hence, during either the pre-allocation of $P_0$ or the execution of $\alg{MES}$ at least one voter has to pay more than $\frac{1}{|N'|}$ per unit satisfaction they received, for a project from $P^*$. Further, this must be a project from $P^*\setminus P_0 $ as whenever a voter funds projects from $P_0$ during the pre-allocation, they must spend at most $\frac{1}{v^*}\leq \frac{1}{|N_p|}\leq \frac{1}{|N'|}$ per unit satisfaction, from the definition of the \pay{Value-Based} pre-allocation method. Then, a contradiction can be obtained analogously to the proof of \Cref{thm:null_EJR}.
\paragraph{\textbf{Case 2:}} $|N_p| > v^*$. 

    From our definition of the \pay{Value-Based} method this must mean that there exists $P'_0\subset P_0$ with $c(P'_0)+c(p)>B_{\alg{MES}}$ and $|N_{p'}|\geq |N_p|$ for each $p'\in P'_0$, as $p$ was not selected and  $v(p)=|N_p|>v^*$.
    
    Let the voter payments and budgets produced in the pre-allocation of $\alg{MES}^\pay{Value-Based}(B_{\alg{MES}},P_0)$ be $(\pi_i)_{i\in N}$ and $(b_i)_{i\in N}$ respectively. In order to reach a contradiction, we will consider running $\alg{MES}^{\pay{Value-Based}}(B_{\alg{MES}},P_0')$. We let the voter payments and budgets its pre-allocation produces be $(\pi'_i)_{i\in N}$ and $(b'_i)_{i\in N}$ respectively and let its available and minimum voter budget shares be $\alpha'$ and $\alpha'_v$ respectively. Clearly $\pi'_i\leq \pi_i$ for any voter $i$ as $P'_0 \subseteq P_0$, and every project in $P'_0$ was funded fully by voters. We claim that the following inequalities hold:
    \begin{enumerate}[label={(\arabic*)}]
    \setlength{\itemsep}{2pt}
        \item $\alpha'_v\geq \alpha_v$, \label{eq1}
        \item $\frac{c(A_i \cap P'_0)}{|N'|}+\frac{c(p)}{|N'|}\leq(\pi'_i+b'_i)$ for all $i\in N'$, \label{eq2}
        \item $\pi'_i\leq \frac{c(A_i \cap P_0')}{|N'|}$ for all $i\in N'$, and \label{eq3}
        \item $b'_i<\frac{c(p)}{|N'|}$ for some $i\in N'$. \label{eq4}
    \end{enumerate}

    For \ref{eq1}, observe that one (perhaps not optimal) way to choose voter budgets $(b'_i)_{i\in N}$ would be to give each voter $b'_i=b_i+(\pi_i-\pi'_i)\geq b_i$. This is a feasible allocation as $\sum_{i\in N} b'_i=\sum_{i\in N} (b_i+\pi_i-\pi'_i)\leq \alpha B+c(P_0)- c(P'_0)=\alpha'B$ (we can account for the inequality by allocating any remaining budget arbitrarily). Thus, $\alpha'_v\geq \min_{i\in N}\{(\pi'_i+b'_i)\frac{n}{B}\}\geq \min_{i\in N}\{(\pi_i+b_i)\frac{n}{B}\}=\alpha_v$.

    For \ref{eq2}, recall that $P'_0\subset P^*$. Thus we know that for each voter $i\in N'$ the following holds: $c(A_i \cap P'_0)+c(p)\leq c(A_i \cap P^*)+c(p)\leq \frac{|N'|}{n}\alpha_v B\leq\frac{|N'|}{n}\alpha'_v B$, using \ref{eq1}. Further, from the definition of minimum voter budget share: $\alpha'_v\leq (\pi'_i+b'_i)\frac{n}{B}$. Combining these, we obtain the statement above.

    For \ref{eq3}, note that each project $p'\in P'_0\subseteq P_0$ has $|N_{p'}|\geq |N_p|\geq |N'|$. This means that each voter $i\in N'$ paid at most $\frac{1}{|N'|}$ per unit satisfaction they obtained from projects in $P'_0$, i.e.,  $\frac{\pi'_i}{c(A_i \cap P'_0)}\leq \frac{1}{|N'|}$.

     For \ref{eq4}, observe that $c(p)>B_{\alg{MES}}-c(P'_0)=\alpha'B=\sum_{i\in N} b'_i\geq \sum_{i\in N'} b'_i$. Therefore, there exists a voter $i\in N'$ such that $\frac{c(p)}{|N'|}> b'_i$.

    Combining \ref{eq2}, \ref{eq3}, and \ref{eq4}, we obtain a contradiction.
\end{proof}

We can also show that the \pay{Equal-Split} pre-allocation method performs reasonably well in the special case that the set of pre-selected projects was chosen by \alg{Greedy}.

\begin{restatable}{theorem}{equalsplitEJR}\label{thm:equal_split_EJR}
    Suppose $\mix=[R_k]_{k \in [m]}$, where $R_1$ is \alg{Greedy} and $R_2$ is $\alg{MES}^{\pay{Equal-Split}}$ with minimum voter budget share $\alpha^{\pay{Equal-Split}}$. Then, the outcome of $\mix$ %
    satisfies $\alpha^{\pay{Equal-Split}}$-budget EJR+ up to any two projects.
\end{restatable}

Unlike the \pay{Value-Based} method, \pay{Equal-Split} and \pay{MES-Style} can subject voters to arbitrarily high payments per unit satisfaction. Because of this, the outcome of \alg{MES} with either of these pre-allocation methods may violate arbitrarily weak proportionality notions, and for \pay{MES-Style} this holds even if the set of pre-selected projects was chosen by \alg{Greedy}.

\begin{restatable}{theorem}{violations}\label{thm:violations}
    Let $M\in \{\pay{Equal-Split},\pay{MES-Style}\}$ and $k\in \mathbb{N}$. Consider a mixed voting rule $\mix$ such that $\alg{MES}^\pay{M} \in \mix$ with available budget share $\alpha$. Then, the outcome of $\mix$ does not necessarily satisfy $\alpha$-budget EJR up to $k$ projects. If $k=1$ or $M=\pay{MES-Style}$, this holds even when \alg{MES} is the second rule in $\mix$ and the first rule in $\mix$ is \alg{Greedy} (both tie-breaking in favor of large projects).
\end{restatable}

\section{Experiments}

\label{sec:experiments}

In this section, we empirically evaluate our framework by applying mixed voting rules to a large dataset of real-world PB instances. We initially consider mixed rules obtained by a combination of \alg{Greedy} and \alg{MES} (with \alg{Greedy} completion), where \alg{MES} is implemented with one of the pre-allocation methods introduced in \Cref{sec:adapting_mes}. We then progressively refine our methodology, and consider other voting rules.

We compute the results for our mixed rules while varying the fraction $\alpha_\alg{G}$ of the budget provided to \alg{Greedy} from $0$ to $1$ in steps of $0.1$. 
Formally, let $\mix^{\pay{M}}=[\alg{Greedy}, \alg{MES}^{\pay{M}}, \alg{Greedy}]$. Then, for an instance $I = (B,P,A,c)$, we compute $\mix^{\pay{M}}([\alpha_\alg{G} B, B, B])$ for all $\alpha_\alg{G} \in \{0, 0.1, 0.2, \dots, 1\}$ and all $\pay{M} \in \{\pay{Null}, \pay{MES-Style}, \pay{Equal-Split}, \pay{Value-Based}\}$. The variable $\alpha_\alg{G}$, which we will refer to as the \textit{greedy (budget) share}, is used to interpolate between the two rules, with $\alpha_\alg{G} = 0$ corresponding to \alg{MES} (with \alg{Greedy} completion) and $\alpha_\alg{G} = 1$ corresponding to \alg{Greedy}.

\paragraph{Dataset} The data for our experiments is obtained from \textsc{Pabulib} \citep{FFP+23a}, a library of PB instances. Our dataset comprises all non-trivial real-world instances with at least 20 projects, which amounts to 313 instances.\footnote{\textsc{Pabulib} contains instances that are artificially generated or trivial instances, where all projects can be funded. The rationale for considering only instances with 20+ projects is discussed in \Cref{app:splitting_budget}.}
We break ties lexicographically by project name.

\paragraph{Measures} \label{subsec:measures}

Since we are combining rules which aim to optimize utilitarian welfare (\alg{Greedy}) and provide proportional representation (\alg{MES}), we evaluate the mixed rules on these criteria. For a given instance $I = (B,P,A,c)$ and outcome $P^\ast \subseteq P$, we compute several numerical measures and then average the results over all instances in our dataset. 

To measure welfare, we use the utilitarian ratio (\Cref{def:utilitarian_ratio}). %
For proportionality, we consider two measures.
First, we check whether $P^*$ satisfies EJR+ up to any project (henceforth abbreviated to \textit{EJR+X}).
Then, to get a more granular understanding of how ``close'' an outcome is to violating EJR+X, we compute the maximum 
value of $\alpha$ for which $P^*$ satisfies $\alpha$-budget EJR+X.\footnote{This is technically a supremum; we look for the smallest value of $\alpha$ for which $P^*$ violates $\alpha$-budget EJR+X.}
This yields a quantitative proportionality measure, analogously to one recently suggested by \citet{BBMP25a}. %

\subsection{Mixing \alg{Greedy} and \alg{MES}}
\label{sec:exp1}

We start by discussing our experimental results for mixing \alg{Greedy} with \alg{MES} using the four different pre-allocation methods defined in \Cref{sec:adapting_mes}.
In \Cref{fig:results_payments_EJR_violations} we observe that, for non-\alg{Null} pre-allocation methods, the number of proportionality violations only significantly increases once \alg{Greedy} has a share of at least $80\,\%$ of the budget, suggesting that even a small amount of budget allocated to \alg{MES} is sufficient to achieve proportional outcomes. On the other hand, we see that even \alg{Greedy} itself satisfies EJR+X on most of the instances, indicating that the property is relatively easy to satisfy in practice. \Cref{fig:results_payments_min_beta_EJR}, showing the trade-off between proportional representation and utilitarian welfare, further supports this observation. The average values of $\alpha$ for which $\alpha$-budget EJR+X is satisfied range between $1.5$ and $2.9$, well above the theoretical guarantees. 

For greedy shares $\alpha_\alg{G}\in \{0.1, 0.2, 0.3\}$, most mixed rules perform similarly, with a slightly lower utilitarian welfare, but about the same proportionality as \alg{MES} ($\alpha_\alg{G} = 0$). This decrease in utilitarian welfare can be explained by the \alg{Greedy} rule having to pick suboptimal projects due to its low budget constraint. We discuss this effect in more detail in \Cref{sec:exp_early_stopping}. The mixed rules perform best for greedy shares in the range $0.6$ to $0.9$. 
For greedy shares $\alpha_\alg{G} \ge 0.5$, all other pre-allocation methods consistently outperform \pay{Null} across both metrics. 
It is also noteworthy that even for $\alpha_\alg{G}= 0.8$, the proportionality of the mixed rule is much closer to \alg{MES} than \alg{Greedy}. 
However, this is partially a side effect of splitting the budget into two parts, which on its own can lead to higher proportionality values on our data. For a detailed discussion of this effect, we refer to \Cref{app:splitting_budget}.

All pre-allocation methods aside from \pay{Null} behave similarly and continue to do so in the remaining experiments. Since \pay{Value-Based} provides the strongest theoretical guarantees, we report only the results of $\alg{MES}^\pay{Value-Based}$ for the remainder of the experimental section.

\begin{figure}[t]
    \begin{subfigure}[t]{.48\textwidth}
        \centering
        \includegraphics[scale=0.6]{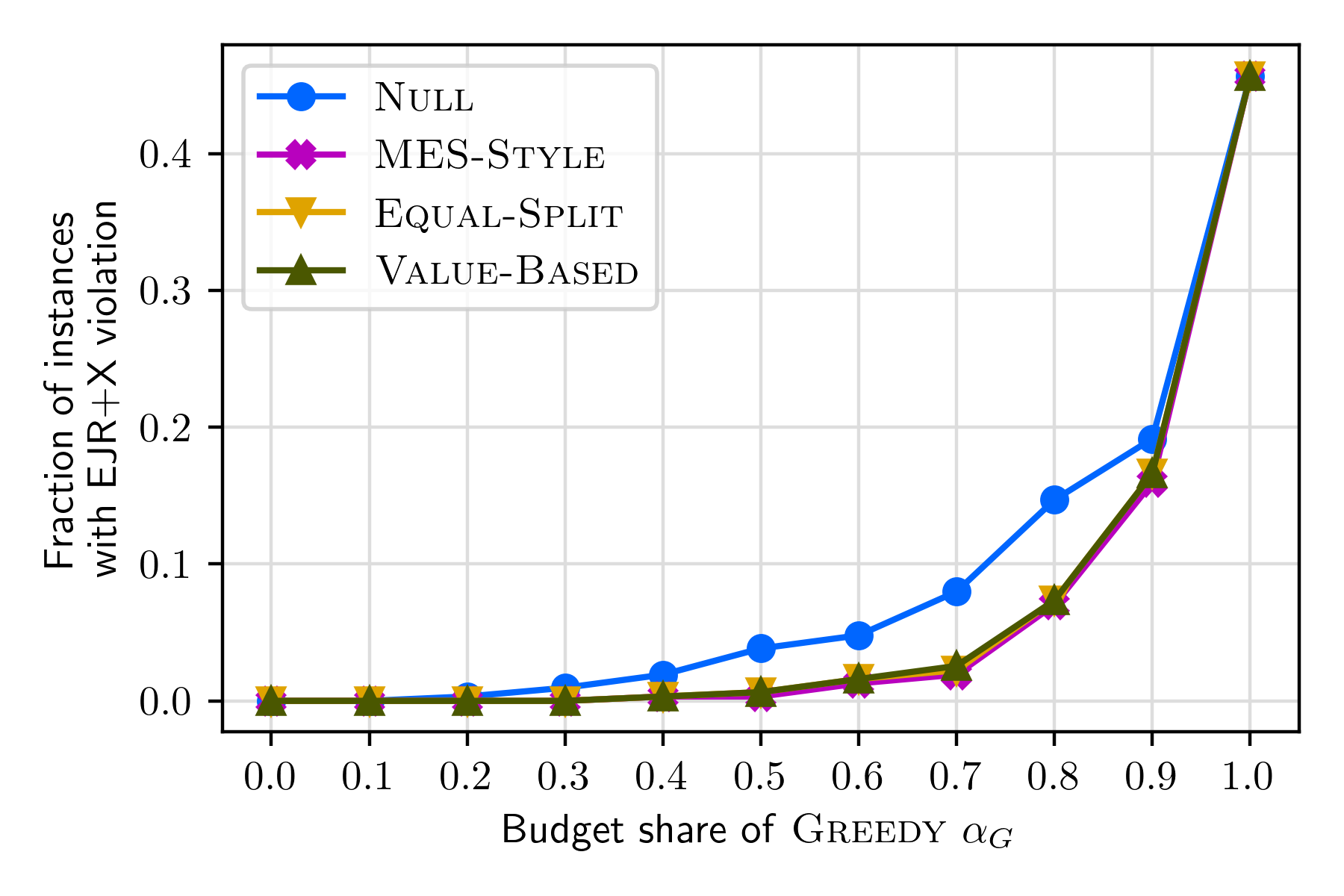}
        \caption{Fraction of instances for which EJR+ up to any project is violated.} %
        \label{fig:results_payments_EJR_violations}
    \end{subfigure}
    \hfill
    \begin{subfigure}[t]{.48\textwidth}
        \centering
        \includegraphics[scale=0.609]{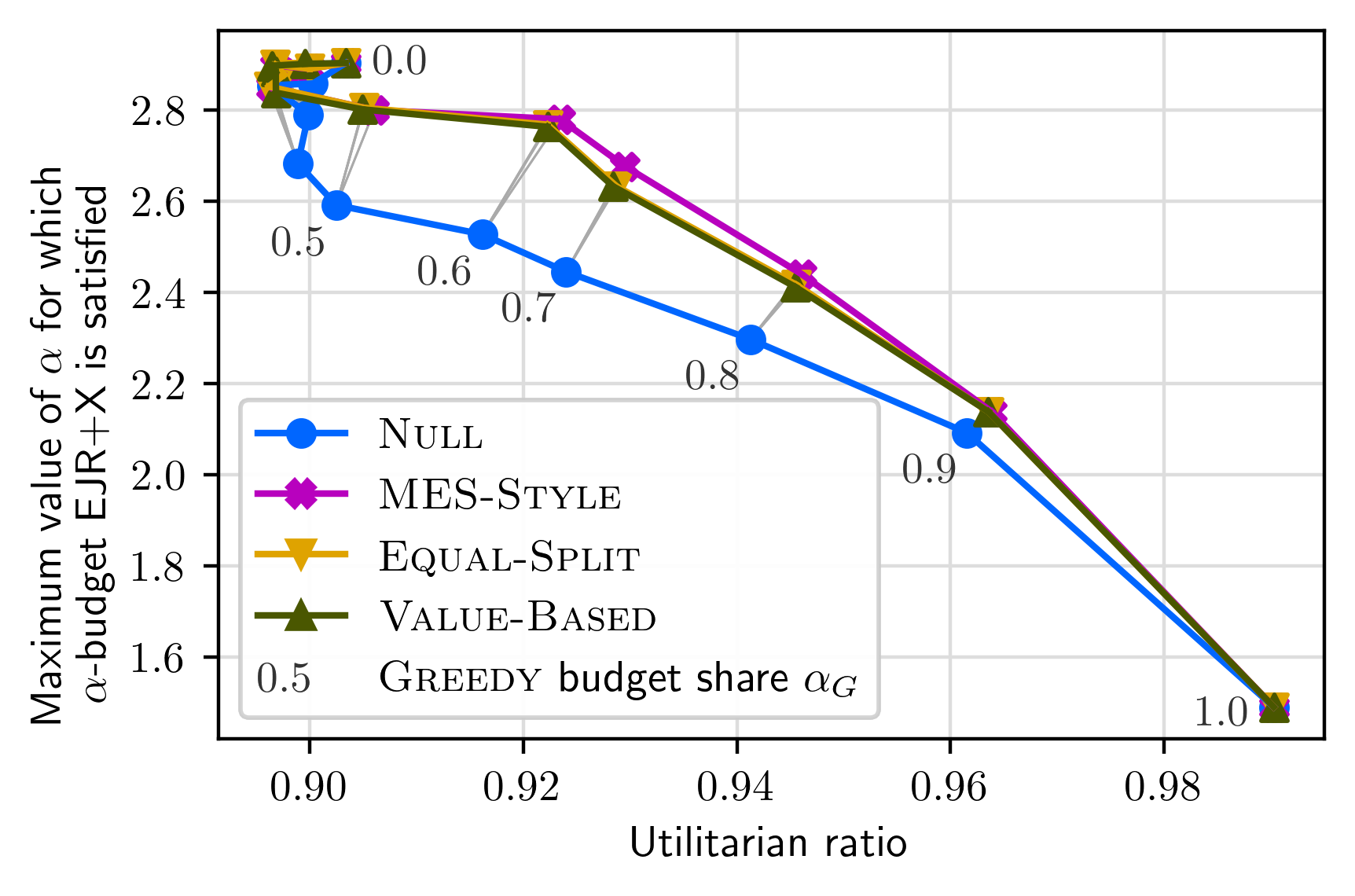}
        \caption{Maximum $\alpha$ for which $\alpha$-budget EJR+ up to any project is satisfied over the utilitarian ratio.}
        \label{fig:results_payments_min_beta_EJR}
    \end{subfigure}
    \caption{Experimental results for mixing \alg{Greedy} with \alg{MES} (without budget increase, with \alg{Greedy} completion) for different pre-allocation methods and greedy budget shares $\alpha_{G} \in \{0, 0.1, 0.2, \dots, 1\}$. Metrics are averaged over all instances with at least $20$ projects.}
    \label{fig:results_payments}
\end{figure}

\subsection{\alg{MES} with Budget Increase}
\label{sec:exp2}

A well-known issue with \alg{MES} is that it can leave a significant portion of the budget unspent, which is why we complete it with \alg{Greedy} to guarantee an exhaustive outcome. On our dataset, when \alg{MES} is allocated the entire budget ($\alpha_\alg{G} = 0$), it spends only about $62\,\%$ of it, on average. Splitting the budget across multiple rules amplifies this problem: an even split ($\alpha_\alg{G} = 0.5$) results in \alg{MES} spending $18\,\%$ of the total budget, and for $\alpha_\alg{G} = 0.9$, \alg{MES} spends less than $1\,\%$ (see \Cref{app:budget_spending} for details). In the latter case, $[\alg{Greedy}, \alg{MES}^{\pay{M}}, \alg{Greedy}]$ almost degenerates into $[\alg{Greedy}, \alg{Greedy}]$.

A popular way to increase the proportion of the budget spent by \alg{MES} is to artificially increase the budget available to \alg{MES}, until its outcome is close to exhaustive  \citep{PPS21a}.
More precisely, we iteratively increase the budget available to MES by some $\beta$, stopping either when an exhaustive outcome is returned, or just before an infeasible outcome is reached. We generalize this to our setting by altering the input of any pre-allocation method \pay{M} in \Cref{def:rebalancing_process} to use an MES budget share of $\alpha^x =\alpha+\frac{\beta x}{B}$, increasing $x\in \mathbb{N}$ until either stopping condition is reached, calling the resulting rule $\mesincrease{M}$. The proportionality guarantees from Theorems~\ref{thm:null_EJR}--\ref{thm:equal_split_EJR} then hold with respect to the minimum voter budget share produced by this method in its final iteration, which can exceed~$1$. For more details on \alg{MES} with Budget Increase and how we adapt it to the mixed setting, see \Cref{app:budget_increase}.

\begin{figure}[t]
    \begin{minipage}[t]{.48\textwidth}
        \centering
        \includegraphics[scale=0.609]{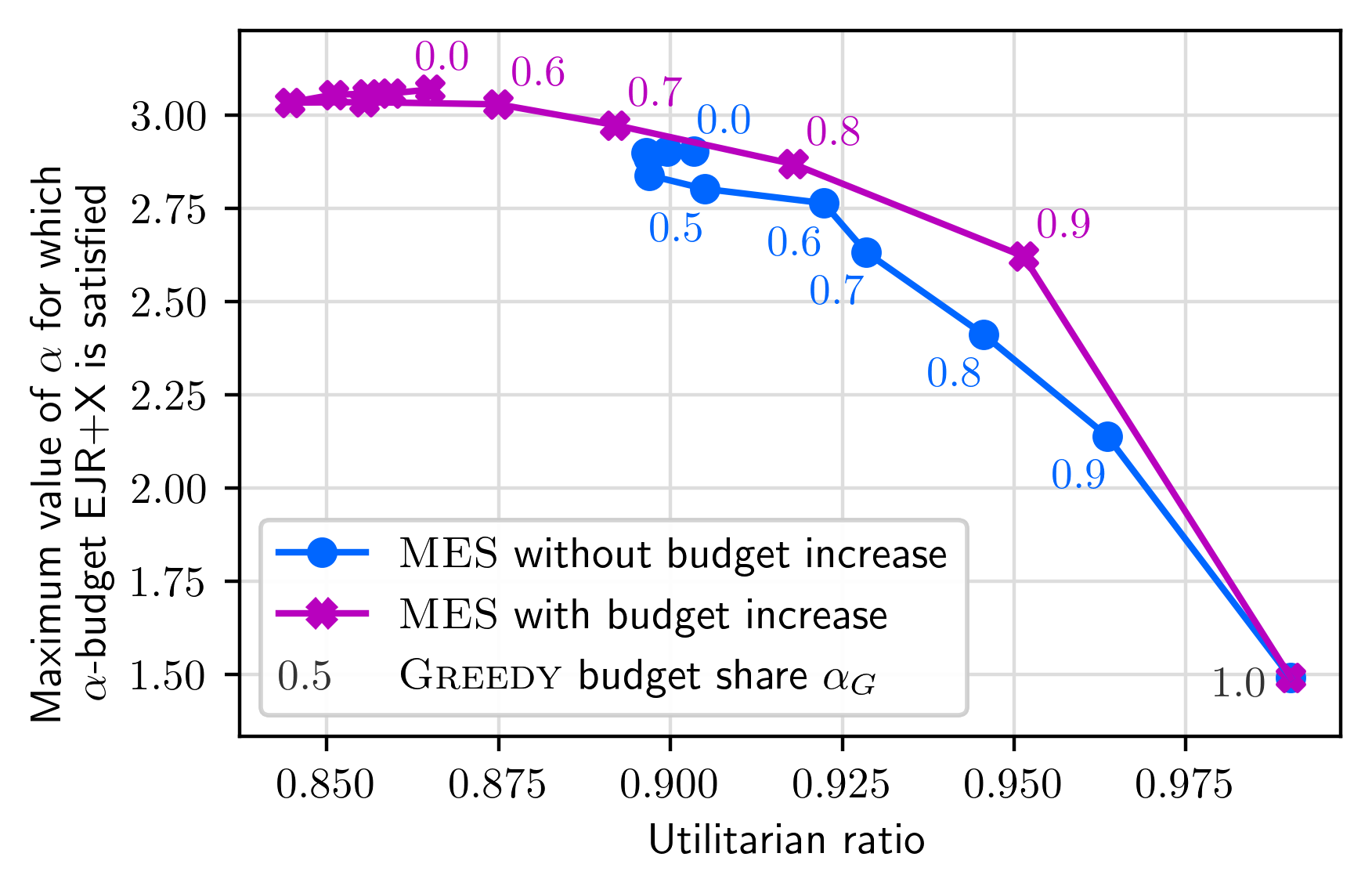}
        \caption{Proportionality and welfare of mixing \alg{Greedy} and \alg{MES} with and without budget increase.}
        \label{fig:results_budget_increment}
    \end{minipage}
    \hfill
    \begin{minipage}[t]{.48\textwidth}
        \centering
        \includegraphics[scale=0.6]{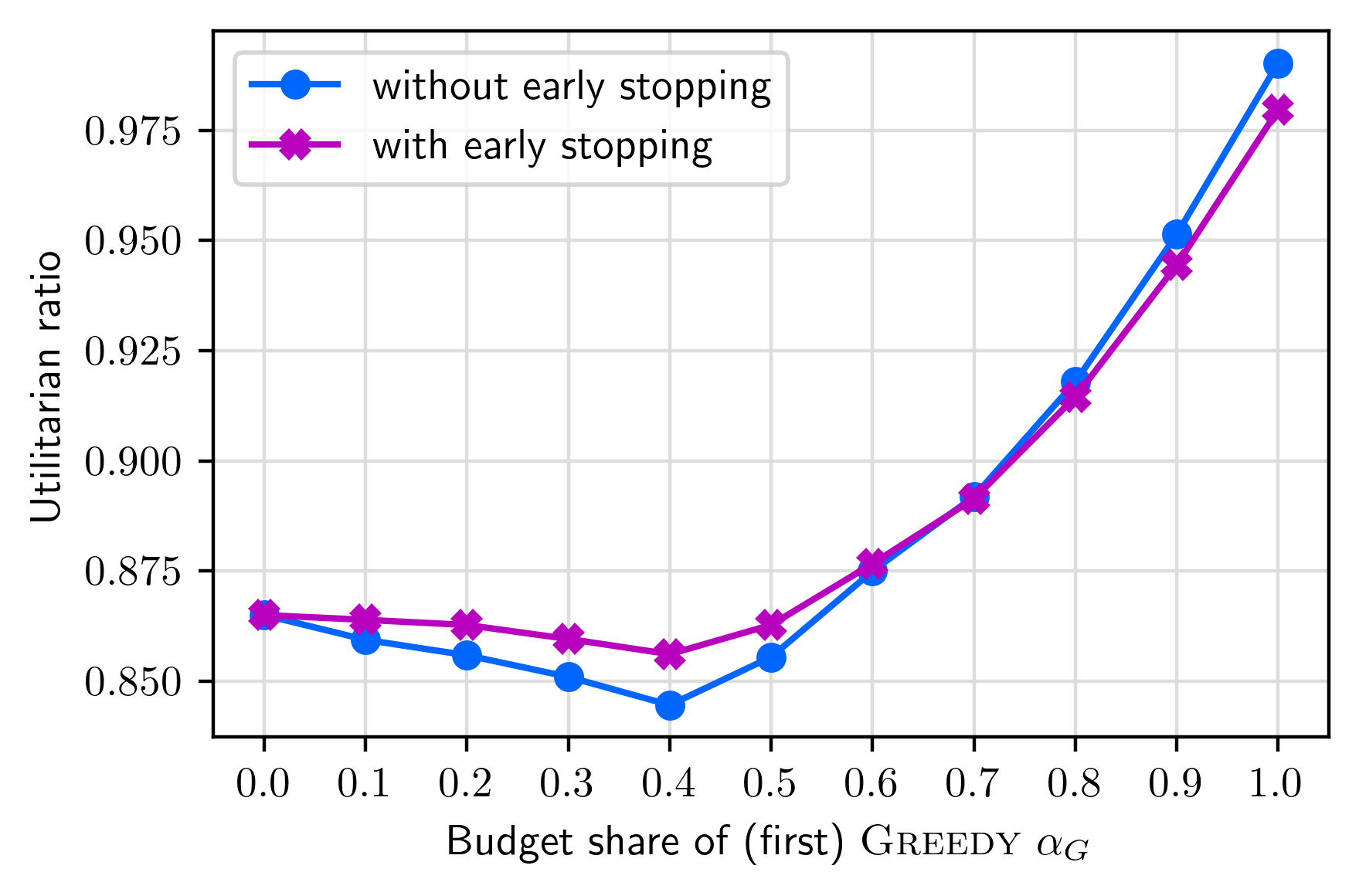}
        \caption{Decrease of welfare when mixing small fractions of \alg{Greedy} with $\alg{MES}$.}
        \label{fig:early_stopping_welfare_dip}
    \end{minipage}
    \label{fig:results_improvements}
\end{figure}

When using this approach on our data, the average fractions of the budget spent by \alg{MES} increase to $95\,\%$ for $\alpha_\alg{G} = 0$, to $46\,\%$ for $\alpha_\alg{G} = 0.5$ and to $8\,\%$ for $\alpha_\alg{G} = 0.9$ (see \Cref{app:budget_spending}).
\Cref{fig:results_budget_increment} compares our experimental results for mixing \alg{Greedy} with \alg{MES} with and without budget increase. When comparing the two methods for a fixed greedy share $\alpha_\alg{G}$, we observe an increase in proportionality and a decrease in welfare. This is unsurprising, since we essentially take some of the budget that the \alg{Greedy} completion spent and allow \alg{MES} to spend it more proportionally. 
However, when considering the entire curve (assuming well-behaved interpolation), the mixed rule with budget increase dominates the one without, in the sense that for any point on the curve without budget increase, we can find a point on the curve with budget increase that improves upon it across both metrics.

Introducing the budget increase heuristic resolves the issue of \alg{MES} spending significantly less than the budget allocated to it, which makes our results for mixing \alg{Greedy} and \alg{MES} more meaningful, in particular for large $\alpha_\alg{G}$. We continue to use this adaptation for the remaining experiments.

\subsection{Stopping \alg{Greedy} Early} \label{sec:exp_early_stopping} 

In \Cref{fig:results_payments_min_beta_EJR,fig:results_budget_increment} we can observe a somewhat surprising drop in utilitarian welfare when mixing small percentages of \alg{Greedy} with $\alg{MES}$. 
We primarily attribute this effect to employing an exhaustive rule as the initial rule in the composition. Many popular rules from the literature iteratively select the next project which maximizes a certain objective among the projects that still fit into the budget constraint. As we approach the budget limit, the pool of affordable projects shrinks, and the chosen projects can become increasingly suboptimal. While this is beneficial in standard PB, where spending the remaining budget on any project is better than leaving it unused, it might not be desirable for the first rule in a mixture, as any unspent budget can still be utilized by later rules. In order to alleviate this problem, we consider using ``\alg{Greedy} with early stopping,'' which we briefly introduced in \Cref{sec:mixed}. This non-exhaustive variant of \alg{Greedy} iteratively picks unselected projects with the maximum number of approvals, until the next project to be selected is not affordable. An interesting theoretical observation is that the pre-allocation methods \pay{Equal-Split} and \pay{Value-Based} coincide when the set of pre-selected projects is chosen by \alg{Greedy} with early stopping.

\Cref{fig:early_stopping_welfare_dip} shows a comparison of the average utilitarian ratio when using \alg{Greedy} with versus without early stopping, mixed with \alg{MES} with budget increase and \alg{Greedy} completion. 
We can see that the decrease of utilitarian welfare for small greedy budget shares is significantly reduced. 
Note that with early stopping, we are no longer interpolating between \alg{MES} and \alg{Greedy}. Even with a greedy share of $\alpha_\alg{G} = 1$, \alg{MES} still runs on the portion of the budget left unspent by the non-exhaustive version of \alg{Greedy}, which explains why the method without early stopping achieves higher welfare for large greedy budget shares. 

For our experiments so far, we have only considered large instances with at least $20$ projects. This is because small instances tend to include projects that cost a large fraction of the total budget, and splitting the budget into two parts effectively prevents them from being selected, leading to suboptimal outcomes. \alg{Greedy} with early stopping can help mitigate this problem; see \Cref{app:splitting_budget}, where we also show empirical results for small instances.

\subsection{Mixing Other Rules} \label{sec:other_rules}

Thus far, we have focused on mixing \alg{Greedy} with \alg{MES}. In the following, we consider mixing \alg{Greedy} with three other popular PB rules. \alg{Greedy} remains the first rule in the mix and its outcome continues to serve as a proxy for a reasonable set of pre-selected projects that doesn't maximize proportionality.
Before discussing our results, we briefly explain each rule and how we adapt it to the mixed setting. For more detailed explanations, we refer to \Cref{app:other_rules_theory}.

\textit{Sequential Phragmén} (\seqPhrag) \citep{BFJL24a,LCG22a} is a sequential voting rule similar to \alg{MES}. \seqPhrag assigns each voter an initial budget of $0$ and then continuously increases all voter budgets, until some group of voters can collectively afford a commonly approved project. That project is selected, and its supporters pay for it from their budgets. This process continues until an exhaustive outcome is reached. Inspired by the \pay{Equal-Split} method,\footnote{The \pay{Value-Based} method and the notion of threshold value are heavily adapted to \alg{MES} and would need to be altered to be meaningful for \seqPhrag.} we adapt \seqPhrag to a set of pre-selected projects by splitting the cost of the pre-selected set $P_0$ equally among supporters, giving each voter $i\in N$ a negative initial budget of $-\sum_{p \in A_i \cap P_0} \frac{c(p)}{|N_p|}$. %

The \textit{Method of Equal Shares with Bounded Overspending} (\alg{BOS}) \citep{GPS+24a} is a recent evolution of \alg{MES} that occasionally allows voters to spend more than their remaining budget. Nevertheless, \alg{BOS} uses the same initial voter budgets as \alg{MES}, and can thus be adapted in the same way. We use the \pay{Value-Based} pre-allocation method in the experiments below.

\textit{Greedy Chamberlin--Courant} (GreedyCC) \citep{TaFa19a} %
greedily selects projects approved by the largest amount of unrepresented voters per unit cost, until an exhaustive outcome is obtained. Like the \alg{Greedy} method, we adapt \alg{GreedyCC} to the mixed rule framework simply by greedily selecting projects from those outside the pre-selected set.

\Cref{fig:results_other_methods} shows the results of mixing \alg{Greedy} with the above rules on our data. As none of the above rules satisfy EJR+X, it is not surprising that all of them perform worse with respect to our proportionality measure than \alg{MES} for $\alpha_\alg{G} = 0$. \alg{BOS} is the only rule that slightly outperforms \alg{MES} across both metrics for some greedy share. This constitutes additional evidence of \alg{BOS} being an effective project selection mechanism.

It is surprising that for \seqPhrag and \alg{GreedyCC}, proportionality increases as we allocate more budget to a purely utilitarian rule (for $0 \leq \alpha_\alg{G}\leq 0.6$). In the case of \alg{GreedyCC}, this behavior can be explained by the fact that this rule does not prioritize selecting projects that any particular voter approves of, beyond the first. Therefore, any group of voters that approves several low-cost projects can cause $\alpha$-budget EJR+X violations for relatively small values of $\alpha$. Thus, allocating more of the budget to \alg{Greedy} might actually give these projects a higher chance of being selected, improving proportionality. In the case of \seqPhrag, we believe this effect could be caused by voters spending their budgets inefficiently, as the rule does not prioritize splitting project costs evenly among supporters.

\begin{wrapstuff}[type=figure,width=0.49\textwidth]
    \centering
    \includegraphics[scale=0.6]{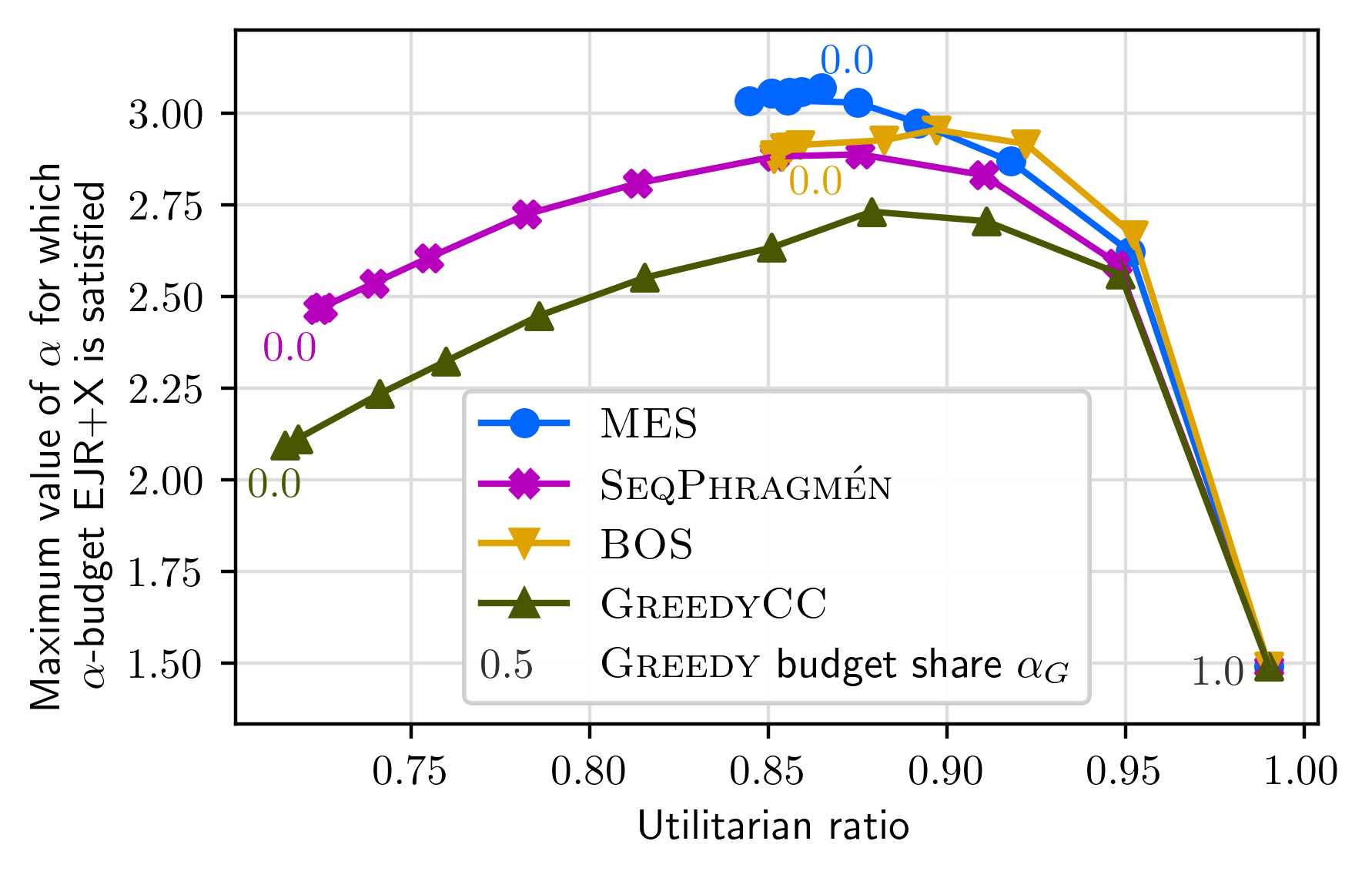}
    \vspace{-15pt}
    \caption{Mixing \alg{Greedy} with other rules.}
    \label{fig:results_other_methods}
\end{wrapstuff}

To some extent, evaluating all mixed rules on our EJR+X-based proportionality measure may be seen as unfair, since this is not necessarily the objective these methods are designed to optimize. For example, the mixed rule containing \alg{GreedyCC} may be more appropriately evaluated using a measure based on \textit{coverage}, such as the proportion of voters approving at least one project in the outcome (see \Cref{app:other_methods_results}). On the other hand, using rules with weaker proportionality guarantees, such as \alg{BOS}, or \alg{MES} with \pay{Equal-Split} or \pay{MES-Style} pre-allocation, can sometimes lead to slight empirical improvements for the two measures we are considering, as shown in \Cref{fig:results_payments,fig:results_other_methods}. Combined with the large gap between parameterized theoretical guarantees and practical outcomes, this suggests that exploring different PB rules and their mixed versions might prove fruitful in improving the efficiency and proportionality of PB outcomes.

\section{Extensions}\label{sec:extensions}

In this section, we consider mixed rules for variations of our approval-based PB setting with cost satisfaction. We first examine a special case, namely approval-based multi-winner voting, and then discuss a generalization of our setting that allows for arbitrary satisfaction functions.

\subsection{Multiwinner Voting} 
An approval-based \emph{multiwinner} instance is a type of PB instance where all projects have the same cost, normalized to $1$, and the instance budget $B$ and rule budgets $B_R$ are all integers.
Multiwinner voting is an important and extensively studied field used to model committee elections (see, e.g., \citet{LaSk22a} for an overview).

For multiwinner instances, it is possible to satisfy stronger versions of the welfare and proportionality axioms introduced in \Cref{sec:preliminaries}. Namely, \alg{Greedy} always selects an efficient outcome, achieving a utilitarian ratio of $1$, and \alg{MES} always selects an outcome satisfying \emph{EJR+}.

\begin{definition}\label{def:EJR+_multiwinner}
    An outcome $P^*\subseteq P$ satisfies \emph{EJR+} if for every group of voters $N'\subseteq N$ and every project $p\in \bigcap_{i \in N'} A_i \setminus P^*$, there is a voter $i\in N'$ with
    $c(A_i\cap P^*)=|A_i\cap P^*| \geq  \lfloor\frac{|N'|B}{n}\rfloor$. 
\end{definition}

Both of our adaptations of \alg{Greedy} from \Cref{sec:mixed}, with and without early stopping, are equivalent in the multiwinner setting, and it is easy to see that they achieve $\alpha$-efficiency when provided with a budget share of $\alpha$, analogously to \Cref{prop:greedy_efficiency}. The \alg{MES} proportionality guarantees from \Cref{thm:null_EJR,thm:value_based_EJR,thm:equal_split_EJR} also extend naturally to this constrained setting.

\begin{restatable}{corollary}{cormultiwinner}\label{cor:multiwinner}
    For multiwinner instances:
    \begin{itemize}
        \item The result of \Cref{thm:null_EJR} holds with respect to $\alpha$-budget EJR+.
        \item The result of \Cref{thm:value_based_EJR} holds with respect to $\alpha^\pay{Value-Based}$-budget EJR+.
        \item The result of \Cref{thm:equal_split_EJR} holds with respect to $\alpha^\pay{Equal-Split}$-budget EJR+.
    \end{itemize}
\end{restatable}

\mut{Maybe talk about which of the results do not hold/are open here? I don't think it needs much detail, just to wrap up the subsection.}

\subsection{General Satisfaction Functions} \label{sec:general_satisfaction}
Throughout this paper, we have assumed that voters derive their satisfaction from the cost of approved projects. However, this is not the only possible way to model satisfaction. Instead, we will assume that our satisfaction functions are \emph{additive} \citep{PPS21a}, with $\mu_i(P^*)=\sum_{p\in P^*} \mu_i(p)$ for any voter $i\in N$ and project set $P^*\subseteq P$. Additionally, in line with \citet{BFL+23a} we will also assume that voters attain no satisfaction from unapproved projects, i.e., for any voter $i\in N$, with $p\not\in A_i$ we have $\mu_i(p)=0$, and that voters get the same satisfaction from a commonly approved project, i.e., for any voters $i,j\in N$, with $p\in A_i\cap A_j$, $\mu_i(p)=\mu_j(p)$. For notational convenience, we use $\mu(p)$ to refer to the satisfaction of any supporter of $p$, defining welfare, $uw(p)=|N_p|\mu(p)$, and value (for money), $v(p)=\frac{uw(p)}{c(p)}=\frac{|N_p|\mu(p)}{c(p)}$, accordingly. Observe that cost satisfaction functions are a special case of this model, but that it also captures many other functions, such as the cardinality satisfaction function, which assigns a satisfaction of $1$ to every approved project.

\mut{We could try to introduce this satisfaction model a bit quicker/directly. As a reader, I was slightly disappointed by it becoming progressively weaker after we started talking about additive utilities. Maybe something like (not polished at all) "A more general model is that of \citet{BFL+23a} in which the satisfaction of a voter from a project is not bound to its cost, but some arbitrary value for each project. More formally, the function $\mu$ assigns every project a utility and any voter $i\in N$, with $p\in A_i$ has $u_i(p)=\mu_j(p)$, while voters with $p\not\in A_i$ have $\mu_i(p)=0$."}

The voting rules we defined in \Cref{sec:preliminaries} extend to this more general setting, with \alg{Greedy} sequentially selecting the affordable project with highest value for money, and \alg{MES} selecting project $p$ with lowest $\rho(p)$ from
$\sum_{i\in N_p} \min(b_i,\rho(p) \mu(p))=c(p)$ and updating the budgets of its supporters to $b_i-\min(b_i, \rho(p)\mu(p))$ \citep{PPS21a}. Similarly, our pre-allocation methods extend with their existing definitions, using the more general versions of \alg{MES} and value defined above. Voter payments for the \pay{Value-Based} method, which we focus on in this section, are generalized to $\pi_i=\sum_{p\in A_i \cap P_0} \frac{\mu_i(p)}{\max\{v(p),v^*\}}$.

In order to argue about the proportionality of these generalized methods, we cannot use the notion of `EJR+ up to any project', as it is only well defined for cost satisfaction. Instead, we will utilize a strengthening of this axiom that was recently introduced by \citet{Skow26a}; we present it in its parameterized version (compare~\Cref{obs:1}).

\mut{I think it would be nice if there was some intuitive explanation of this strong EJR+ here, especially since it's a property that the reader won't find anywhere else. (Can't say I fully understood it myself)}

\begin{definition}
    An outcome $P^*\subseteq P$ satisfies \emph{$\alpha$-budget Strong EJR+ up to any project} if for every project $p\in P\setminus P^*$, and every group of voters $N'\subseteq N$ with $N'\cap N_p\neq \emptyset$ there is a voter $i\in N'$ with

    $$ \sum_{p'\in (P^*\cup \{p\})\cap A_i} \mu(p') \cdot \min \bigg\{\frac{c(p')}{|N'\cap N_{p'}| \cdot  \mu(p')},\frac{c(p)}{|N'\cap N_{p}| \cdot \mu(p)}\bigg\}>\alpha \frac{B}{n}.$$
\end{definition}

For cost satisfaction this condition reduces to 
$ \sum_{p'\in (P^*\cup \{p\})\cap A_i} \min \Big\{\frac{c(p')}{|N'\cap N_{p'}|},\frac{c(p')}{|N'\cap N_{p}|}\Big\}>\alpha \frac{ B}{n}$,
which is a strengthening of `$\alpha$-budget EJR+ up to any project'.
Nevertheless, the outcome of \alg{MES} always satisfies Strong EJR+ up to any project \citep{Skow26a}.
We use this notion to prove a proportionality guarantee for $\alg{MES}^\pay{Value-Based}$ analogous to \Cref{thm:value_based_EJR}, for our general class of satisfaction functions. 

\begin{restatable}{theorem}{valuebasedstrongEJR}\label{thm:value_based_strong_EJR}
    Consider a mixed voting rule $\mix$ such that $\alg{MES}^\pay{Value-Based}\in \mix$ with available budget share~$\alpha$ and minimum voter budget share $\alpha^{\pay{Value-Based}}\geq \alpha$. Then, the outcome of $\mix$ satisfies $\alpha^{\pay{Value-Based}}$-budget Strong EJR+ up to any project.
\end{restatable}

We provide some initial arguments for why \pay{Value-Based} is, in some sense, optimal, with respect to this stronger proportionality notion in \Cref{app:optimality}.

\section{Conclusion}

Taking inspiration from mixed-member electoral systems across the world, we introduced mixed voting rules for participatory budgeting. Using combinations of \alg{Greedy} and \alg{MES} as our primary examples, we established a general framework for analyzing the performance of such mixed rules, from both a theoretical and an empirical perspective.

Our focus on \alg{MES} stems from its strong proportionality guarantees and its potential to rebalance disproportional selections by earlier rules. Just as the proportional component in a mixed-member electoral system (like Scotland's \textit{Additional Member System}) aims to correct the disproportionalities created by district voting, MES can enhance the proportionality of outcomes when used as a later stage of a mixed voting rule. This parallel has the potential to improve the explainability of mixed PB voting rules, particularly when voters are already familiar with mixed-member systems. At the same time, it would be interesting to apply our framework and techniques to existing mixed-member systems, to analyze what proportionality properties they might guarantee.

We proposed several methods to account for an existing set of selected projects when setting initial voter budgets for \alg{MES}.
Since no rule can exceed the baseline proportionality guarantee for all instances (\Cref{prop:in_general_baseline}), we adopted a more granular approach based on key instance properties, parameterizing our guarantees by the minimum voter budget share. This quantity captures how effectively each method accounts for voters' satisfaction for the already selected projects, and allows us to differentiate between alternative methods.
Among these, only the \pay{Value-Based} method exceeded our baseline in terms of proportionality guarantees, by making sure that voters obtain good value for money when paying for the potentially inefficient set of already selected projects. 

Our experiments suggest that mixing \alg{Greedy} and \alg{MES} works best on large instances, when letting \alg{Greedy} spend at least $60\,\%$ of the budget. 
We achieve significantly better results when using the \pay{MES-Style}, \pay{Equal-Split} or \pay{Value-Based} pre-allocation method over the \pay{Null} method, with \pay{MES-Style} consistently performing slightly better than the other two.  
However, as \pay{MES-Style} fails to give \textit{any} theoretical proportionality guarantees (see \Cref{thm:violations}), the \pay{Value-Based} method is preferable, as it gives strong theoretical guarantees with only a slight practical performance trade-off. We identified several key improvements to mixed rules in our experiments, including an adaptation of the \alg{MES} budget increase heuristic to our setting, in order to better balance the budget between the two rules, and stopping the \alg{Greedy} algorithm early, which helped combat the inefficiencies arising from using an exhaustive rule early in the mix. Finally, we conducted some experiments that considered how other voting rules might be adapted to the mixed rule setting. %

\newpage

\begin{acks}
    This work was supported by a \textit{Structural Democracy Fellowship} through the Brooks School of Public Policy at Cornell University. We thank the anonymous reviewers for their valuable feedback. We thank Piotr Skowron, Jannik Peters, Tomasz Wąs, Dominik Peters and Tuva Bardal for helpful discussions.
\end{acks}

\bibliographystyle{ACM-Reference-Format}
\bibliography{abb,algo,additional}

@STRING{proc = {Proceedings of the }}

@STRING{ijcai = { International Joint Conference on Artificial Intelligence (IJCAI)}}

@STRING{aamas = { International Conference on Autonomous Agents and Multiagent Systems (AAMAS)}}

@STRING{aaai = { AAAI Conference on Artificial Intelligence (AAAI)}}

@STRING{acmec = { ACM Conference on Economics and Computation (ACM-EC)}}

@STRING{neurips = { Conference on Neural Information Processing Systems (NeurIPS)}}

@STRING{springer = {Springer-Verlag}}

@STRING{acm = {ACM Press}}

@article{Waru21a,
	author = {W. Suksompong},
	journal = {SIGecom Exchanges},
	number = {1},
	pages = {46–61},
	title = {Constraints in fair division},
	volume = {19},
	year = {2021}}

@inproceedings{FZC17a,
  author    = {R. Freeman and S. M. Zahedi and V. Conitzer},
  title     = {Fair and Efficient Social Choice in Dynamic Settings},
  booktitle = proc # {26th} # ijcai,
  pages     = {4580--4587},
  year      = {2017},
}

@book{ShWa03a,
	editor = {M. S. Shugart and M. P. Wattenberg},
	publisher = {Oxford University Press},
	title = {Mixed-member electoral systems: The best of both worlds?},
	year = {2003}}

@inproceedings{ARS22a,
  title={An EF2X Allocation Protocol for Restricted Additive Valuations},
  author={Akrami, H. and Rezvan, R. and Seddighin, M.},
  booktitle = proc # {31st} # ijcai,
  pages={17--23},
  year={2022}
}

@unpublished{Skow26a,
	author = {P. Skowron},
	note = {Unpublished manuscript. Personal Communication},
	title = {Robust Proportionality Axioms for Additive Utilities},
	year = {2026}}

@inproceedings{BBP26a,
	author = {A. Baychkov and M. Brill and J. Peters},
	booktitle = proc # {40th} # aaai,
	note = {Forthcoming.},
	title = {Utilitarian Guarantees for the Method of Equal Shares},
	year = {2026}}

@inproceedings{GPS+24a,
	author = {G. Papasotiropoulos and S. Z. Pishbin and O. Skibski and P. Skowron and T. W{\k a}s},
	booktitle = proc # {26th} # acmec,
	institution = {arXiv:2409.15005v1 [cs.GT]},
	pages = {841--868},
	title = {Method of Equal Shares with Bounded Overspending},
	year = {2025}}

@inproceedings{BBMP25a,
	author = {T. Bardal and M. Brill and D. McCune and J. Peters},
	booktitle = proc # {39th} # aaai,
	pages = {13581--13588},
	title = {Proportional Representation in Practice: Quantifying Proportionality in Ordinal Elections},
	year = {2025}}

@article{EFI+24a,
	author = {E. Elkind and P. Faliszewski and A. Igarashi and P. Manurangsi and U. Schmidt-Kraepelin and W. Suksompong},
	journal = {ACM Transactions on Economics and Computation},
	number = {3},
	pages = {11:1--11:27},
	title = {The Price of Justified Representation},
	volume = {12},
	year = {2024}}

@inproceedings{TaFa19a,
	author = {N. Talmon and P. Faliszewski},
	booktitle = proc # {33rd} # aaai,
	pages = {2181--2188},
	title = {A framework for approval-based budgeting methods},
	year = {2019}}

@incollection{AzSh21a,
	author = {H. Aziz and N. Shah},
	booktitle = {Pathways Between Social Science and Computational Social Science},
	editor = {T. Rudas and G. P{\'e}li},
	pages = {215--236},
	publisher = {Springer},
	series = {Computational Social Sciences},
	title = {Participatory Budgeting: Models and Approaches},
	year = {2021}}

@inproceedings{BrPe24a,
	author = {M. Brill and J. Peters},
	booktitle = proc # {38th} # aaai,
	pages = {9528--9536},
	publisher = {AAAI Press},
	title = {Completing Priceable Committees: Utilitarian and Representation Guarantees for Proportional Multiwinner Voting},
	year = {2024}}

@article{BFJL24a,
	author = {M. Brill and R. Freeman and S. Janson and M. Lackner},
	journal = {Mathematical Programming},
	number = {1--2},
	pages = {47--76},
	title = {Phragm\'{e}n's Voting Methods and Justified Representation},
	volume = {203},
	year = {2024}}

@book{WMT21a,
	author = {B. Wampler and S. McNulty and M. Touchton},
	publisher = {Oxford University Press},
	title = {Participatory Budgeting in Global Perspective},
	year = {2021}}

@article{IsBr25a,
	author = {J. Israel and M. Brill},
	journal = {Social Choice and Welfare},
	number = {1--2},
	pages = {221--261},
	title = {Dynamic Proportional Rankings},
	volume = {64},
	year = {2025}}

@inproceedings{FFP+23a,
	author = {P. Faliszewski and J. Fils and D. Peters and G. Pierczy{\'n}ski and P. Skowron and D. Stolicki and S. Szufa and N. Talmon},
	booktitle = proc # {32nd} # ijcai,
	pages = {2667--2674},
	title = {Participatory Budgeting: Data, Tools, and Analysis},
	year = {2023}}

@inproceedings{BrPe23a,
	author = {M. Brill and J. Peters},
	booktitle = proc # {24th} # acmec,
	note = {Full version arXiv:2302.01989 [cs.GT]},
	pages = {301},
	publisher = acm,
	title = {Robust and Verifiable Proportionality Axioms for Multiwinner Voting},
	year = {2023}}

@techreport{ReMa23a,
	author = {S. Rey and F. Schmidt and J. Maly},
	institution = {arXiv:2303.00621 [cs.GT]},
	title = {The (Computational) Social Choice Take on Indivisible Participatory Budgeting},
	year = {2025}}

@inproceedings{BFL+23a,
	author = {M. Brill and S. Forster and M. Lackner and J. Maly and J. Peters},
	booktitle = proc # {37th} # aaai,
	pages = {5524--5531},
	publisher = {AAAI Press},
	title = {Proportionality in Approval-Based Participatory Budgeting},
	year = {2023}}

@book{LaSk22a,
	author = {M. Lackner and P. Skowron},
	publisher = {Springer},
	title = {Multi-Winner Voting with Approval Preferences},
	year = {2023}}

@inproceedings{LCG22a,
	author = {M. Los and Z. Christoff and D. Grossi},
	booktitle = proc # {31st} # ijcai,
	pages = {398--404},
	title = {Proportional Budget Allocations: A Systematization},
	year = {2022}}

@inproceedings{PeSk20b,
	author = {D. Peters and P. Skowron},
	booktitle = proc # {21st} # acmec,
	note = {Full version arXiv:1911.11747 [cs.GT]},
	pages = {793--794},
	publisher = acm,
	title = {Proportionality and the Limits of Welfarism},
	year = {2020}}

@inproceedings{Lack20a,
	author = {M. Lackner},
	booktitle = {Proceedings of the 34th {AAAI} Conference on Artificial Intelligence ({AAAI})},
	pages = {2103--2110},
	publisher = {{AAAI} Press},
	title = {Perpetual Voting: Fairness in Long-Term Decision Making},
	year = {2020}}

@inproceedings{PPS21a,
	author = {D. Peters and G. Pierczy{\'n}ski and P. Skowron},
	booktitle = proc # {34th} # neurips,
	pages = {12726--12737},
	title = {Proportional Participatory Budgeting with Additive Utilities},
	year = {2021}}

@inproceedings{KKE+19a,
	author = {M. Kocot and A. Kolonko and E. Elkind and P. Faliszewski and N. Talmon},
	booktitle = {Proceedings of the 28th International Joint Conference on Artificial Intelligence (IJCAI)},
	pages = {385 -- 391},
	title = {Multigoal Committee Selection},
	year = {2019}}

@inproceedings{ALT18a,
	author = {H. Aziz and B. E. Lee and N. Talmon},
	booktitle = proc # {17th} # aamas,
	pages = {23--31},
	publisher = {IFAAMAS},
	title = {Proportionally Representative Participatory Budgeting: Axioms and Algorithms},
	year = {2018}}

@incollection{FSST17a,
	author = {P. Faliszewski and P. Skowron and A. Slinko and N. Talmon},
	booktitle = {Trends in Computational Social Choice},
	chapter = 2,
	editor = {U. Endriss},
	publisher = {AI Access},
	title = {Multiwinner Voting: A New Challenge for Social Choice Theory},
	year = {2017}}

@article{ABC+16a,
	author = {H. Aziz and M. Brill and V. Conitzer and E. Elkind and R. Freeman and T. Walsh},
	journal = {Social Choice and Welfare},
	number = {2},
	pages = {461--485},
	title = {Justified Representation in Approval-Based Committee Voting},
	volume = {48},
	year = {2017}}

@book{KPP04a,
	author = {H. Kellerer and U. Pferschy and D. Pisinger},
	publisher = {Springer},
	title = {Knapsack Problems},
	year = {2004}}

\newpage
\appendix

\section*{Appendix}

\section{Omitted proofs}\label{app:omitted_proofs}

\obsone*

\begin{proof}
    Fix $\alpha>0$ and let $P^*=R(\alpha B,\emptyset)$. Define $R'$ such that whenever it is provided with available budget share $\alpha$ and set of pre-selected projects $P_0$ it produces outcome $R'(\alpha B+c(P_0),P_0)=P_0\cup P^*$. This outcome is feasible for $R'$ as $c(P_0\cup P^*)\leq c(P_0)+c(P^*)\leq c(P_0)+\alpha B$. Further, if $P^*$ satisfies $\alpha$-budget $X$ then so does $P_0\cup P^*$, completing the proof. Note that when $P_0=\emptyset$, $R'$ is identical to~$R$.
\end{proof}

\ingeneralbaseline*

\begin{wrapstuff}[type=table,width=.37\textwidth]
    \centering
    \scalebox{0.85}{
    \begin{tabular}{c c c c c c c }
        \toprule
         & $p_0$ & $p_1$ & \dots & $p_m$ \\
        \midrule
        Cost  & $(1-\alpha)B$ & $\alpha'B/m$ & \dots & $\alpha'B/m$ \\
        \midrule
        $A_1$ & &  \x   & \dots   &  \x \\
        \bottomrule
    \end{tabular}
    }
    \caption{Instance for \Cref{prop:in_general_baseline} with 
    $P_0=\{p_0\}$ and $m$ is chosen such that $\frac{m-1}{m}\alpha'> \alpha $.}
    \label{tab:empirical_baseline_example_restated}
\end{wrapstuff}

\begin{proof}
    Consider the instance in \Cref{tab:empirical_baseline_example_restated}, and let the outcome of some rule \alg{R}, with rule budget $B_R=B$ and pre-selected set of projects $P_0=\{p_0\}$, be $P^*=\alg{R}(B,P_0)$. Due to feasibility, $P^*$ can contain at most $m-2$ projects from $\{p_1,\dots,p_m\}$. The welfare-maximizing feasible outcome for the instance $(\alpha'B,P,A,c)$ is $\{p_1,\dots,p_m\}$, and thus $P^*$ does not have an $\alpha'$-budget utilitarian ratio of $1$ up to one project. Similarly, for some $p_i\notin P^*$, we have $c(A_1\cap P^*)+c(p_i)\leq \frac{m-1}{m}\alpha'B<\alpha' B$, so $p_i$ constitutes an $\alpha'$-budget EJR+ up to any project violation for voter group $N'=\{1\}$. Note that this argument holds even if we restrict ourselves to projects approved by at least one voter, by having $N_{p_0}=\{1\}$ and $N_{p_j}=\{2,\dots,n\}$ for all $1\leq j\leq m$, choosing $n$ suitably large.
\end{proof}

\greedyefficiency*

\begin{proof}
    Fix an instance $I=(B,P,A,c)$, set of pre-selected projects $P_0$ and greedy budget share $\alpha$.
    Let the outcome of \alg{Greedy} with early stopping be $P_G\supseteq P_0$. Observe that the outcome of \alg{Greedy} without early stopping contains $P_G$, and thus has weakly higher utilitarian welfare.
    Let the outcome of \alg{Greedy} with early stopping for the modified instance $I'=(\alpha B,P,A,c)$ be $P'_G$, and let the first project it could not afford be $p'$. Assuming consistent tie-breaking, it is easy to see that $P'_G\subseteq P_G$, as the projects from $P'_G\setminus P_0$ are the first projects that \alg{Greedy} with early stopping would add to $P_0$. Then, compare $P'_G$ to the welfare-maximizing outcome for $I'=(\alpha B,P,A,c)$. $P'_G$ is efficient up to one project for $I'=(\alpha B,P,A,c)$ using results from Knapsack literature (see e.g. Corollary 2.2.2 in \citet{KPP04a} and Proposition 4 of \citet{BBP26a}). Thus, according to \Cref{def:alpha-budget}, $P_G$ is $\alpha$-efficient up to one project. %
\end{proof}

\subsection{Properties of the Minimum Voter Budget Share}\label{app:alphav}

\alphav*

\begin{proof}\label{proof:alpha_v}\leavevmode
    Pre-allocation methods defined using \Cref{def:rebalancing_process} always perform weakly better (in terms of minimum voter budget share) than just giving each voter their fair share of \alg{MES}'s available budget $\frac{\alpha B}{n}$, which means that $\alpha^\pay{M}\geq \min_{i\in N}\{(\pi_i+\frac{\alpha B}{n})\frac{n}{B}\}\geq \frac{\alpha B}{n}\frac{n}{B}=\alpha$.
    
    Also, $\alpha^\pay{M}=\min \{(\pi_i+b_i)\frac{n}{B}\}\leq \sum_{i\in N}\{(\pi_i+b_i)\frac{n}{B}\}\frac{1}{n}\leq \frac{c(P_0)+\alpha B}{B}= \frac{B_{\alg{MES}}}{B}$.
        
\end{proof}

\begin{claim}\label{claim:rebalancing_induces_every_profile}
    Let $B$ be the instance budget, $B_{\alg{MES}}$ the \alg{MES} budget, and $\alpha$ the \alg{MES} budget share.
    Consider a target \alg{MES} budget profile $(b_i)_{i\in N}$ and a target minimum voter budget share $\alpha^\pay{M}$ (\Cref{def:minimum_voter_budget_share}), satisfying the following constraints:
    \begin{itemize}
        \item $0\leq b_i\leq \frac{\alpha^\pay{M} B}{n}$ (if $b_i>\frac{\alpha^\pay{M} B}{n}$, then we could decrease $b_i$ and raise the budgets of other voters to increase the minimum voter budget share and thus have a fairer initial budget allocation)
        \item $\sum_{i\in N} b_i=\alpha B$
        \item $\alpha\leq \alpha^\pay{M}\leq \frac{B_{\alg{MES}}}{B}$
    \end{itemize}

    Then, there exists a choice of $(\pi_i)_{i\in N}$ (for some pre-allocation method \pay{M}) with $\pi_i\geq 0$ and $\sum_{i\in N} \pi_i\leq c(P_0)$ such that the rebalancing step from \Cref{def:rebalancing_process} outputs this profile.
\end{claim}

\begin{proof}
    Choose $\pi_i=\frac{\alpha^\pay{M} B}{n}-b_i$. 
    
    Then, $\pi_i\geq 0$ as $b_i\leq \frac{\alpha^\pay{M} B}{n}$ and $\sum_{i\in N} \pi_i=\alpha^\pay{M} B-\sum_{i\in N}b_i\leq \alpha^\pay{M} B-\alpha B\leq B_{\alg{MES}}-(B_{\alg{MES}}-c(P_0))=c(P_0)$. Further, $\min_{i\in N} \{\pi_i+b_i\}\frac{n}{B}=\frac{\alpha^\pay{M} B}{n}\frac{n}{B}=\alpha^\pay{M}$ as required.
\end{proof}

\alphavrelationship*

\begin{proof}

Trivially $\alpha=\alpha^\pay{Null}$, as $\pi_i+b_i=0+\frac{\alpha B}{n}=\alpha$ for all voters, from the definition of the \pay{Null} method.

Fix an instance $I$, MES budget $B_{\alg{MES}}$ and set of pre-selected projects $P_0$. Suppose two different pre-allocation methods $\pay{M}$ and $\pay{M'}$ produce profiles $\{\pi_i^{\pay{M}}\}_{i\in N}$ and $\{\pi_i^{\pay{M'}}\}_{i\in N}$, which results in minimum voter budget shares of $\alpha^{\pay{M}}$ and $\alpha^{\pay{M'}}$ respectively computed during the rebalancing step. Using \Cref{def:rebalancing_process} we can find that if for all $i \in N$, $\pi_i^{\pay{M}}\leq \pi_i^{\pay{M'}}$, then $\alpha^{\pay{M}}\leq \alpha^{\pay{M'}}$. Using this, we can show the following:

\begin{itemize}

\item For all $i\in N$ we have $0=\pi_i^{\pay{Null}}\leq \pi_i^{\pay{Value-Based}}$ so $\alpha^\pay{Null}\leq \alpha^\pay{Value-Based}$.
\item For all $i\in N$ we again have $\pi_i^{\pay{Value-Based}}\leq \pi_i^{\pay{Equal-Split}}$ as voters pay weakly less for each project they support in the \pay{Value-Based} method. Thus, $\alpha^\pay{Value-Based}\leq \alpha^\pay{Equal-Split}$.
\item 
Define a new pre-allocation method, $\pay{Equal-Split}'$, which selects voter payments  $\pi_i^{\pay{Equal-Split}'}=\min\{\pi_i^{\pay{Equal-Split}}, \frac{B_{MES}}{n}\}$. $\alpha^{\pay{Equal-Split}'}=\alpha^{\pay{Equal-Split}}$ as any voter with $\pi_i>\frac{B_{MES}}{n}$ gets $b_i=0$ in both cases and does not affect the objective of the rebalancing optimization. 

Now suppose there exists $i\in N$ for which $\pi_i^{\pay{MES-Style}}<\pi_i^{\pay{Equal-Split}'}$. In particular, this means that $\pi_i^{\pay{MES-Style}}<\frac{B_{MES}}{n}$ and thus $i$ has never paid less than their fair share $\frac{c(p)}{|N_p|}$ for any project $p\in A_i\cap P_0$. However, under \pay{Equal-Split}, voter $i$ has always paid exactly $\frac{c(p)}{|N_p|}$ for every project $p\in A_i\cap P_0$. This leads to a contradiction, which means that $\pi_i^{\pay{MES-Style}}\geq \pi_i^{\pay{Equal-Split}'}$ for all $i\in N$, and therefore $\alpha^{\pay{Equal-Split}}=\alpha^{\pay{Equal-Split}'}\leq \alpha^{\pay{MES-Style}}$.
\end{itemize}
\end{proof}

\subsection{Proportionality Theorems}\label{app:proportionality}

\equalsplitEJR*

\begin{proof} 
    For this proof, we make some tweaks to the proof of the first case of \Cref{thm:value_based_EJR} (the second case is not needed). Let $\alpha_v=\alpha^{\pay{Equal-Split}}$. We again let the outcome of $\alg{MES}^{\pay{Equal-Split}}$ be $P^*$, but now assume that there exist $p_1,p_2\notin P^*$ and voter set $N'\subseteq N_{p_1}\cap N_{p_2}$ with $c(A_i \cap P^*)+c(p_1)+c(p_2)\leq \frac{|N'|}{n}\alpha_v B$ for all $i \in N'$. Without loss of generality, assume $c(p_1)\leq c(p_2)$ and let $p=p_1$. Thus, we know that $c(A_i \cap P^*)+2c(p)\leq \frac{|N'|}{n}\alpha_v B$ for all $i \in N'$.

    Let the outcome of \alg{Greedy} be $P_0$ (as it is the set of pre-selected projects for $\alg{MES}^{\pay{Equal-Split}}$). \alg{Greedy} did not select $p$, which means that there exists $P'_0\subseteq P_0$ with $|N_{p'}|\geq |N_p|$ for all $p'\in P'_0$ and $c(P'_0)+c(p)>B_{1}\geq c(P_0)$.

    We now consider the spending of voters from $N'$ for projects in $P'_0\cup (P^*\setminus P_0)\subseteq P^*$ (omitting $P_0\setminus P'_0$ as that spending may have been inefficient), and the satisfaction they achieve from those projects:
    $$\frac{\text{spending by voters in } N'}{\text{satisfaction of voters in } N'}\geq\frac{(\sum_{i\in N'}\frac{\alpha_v B}{n}-b^r_i)-c(p)}{\sum_{i\in N'} c(A_i\cap P^*)}>\frac{|N'|\frac{\alpha_v B}{n}-2c(p)}{|N'|(|N'|\frac{\alpha_v B}{n}-2c(p))}=\frac{1}{|N'|}
    $$

    Hence, during either the pre-allocation of $P'_0\subseteq P_0$ or the execution of $\alg{MES}$ (which additionally selected $P_R\setminus P_0$) at least one voter had to pay more than $\frac{1}{|N'|}$ per unit satisfaction they received, for a project from $P_R$. Further, this must be a project from $P_R\setminus P_0$ as whenever a voter funded projects from $P'_0$ during the pre-allocation, they must have spent at most $\frac{1}{|N_p|}\leq \frac{1}{|N'|}$ per unit satisfaction. The rest of the argument follows from the proof of \Cref{thm:null_EJR}.
\end{proof}

\violations*

\begin{proof}
    This is a direct consequence of \Cref{prop:equal_split_not_EJRk,prop:mes_style_fails_EJRk,prop:equal_split_not_EJR} in \Cref{app:negative_results} below, which also prove some slightly stronger claims.
\end{proof}

\cormultiwinner*

\begin{proof}
    The first two claims follow from the fact that EJR+ is equivalent to EJR+ up to one project in the multiwinner setting. For the last claim, observe that when \alg{Greedy} is given some budget $B_G$ it will select the $B_G$ most popular projects. Thus the \pay{Equal-Split} and \pay{Value-Based} methods are equivalent when $P_0$ was selected by \alg{Greedy}, as the value of every project in $P_0$ is weakly higher than the threshold value.
\end{proof}

\valuebasedstrongEJR*

\begin{proof}

    We build on the proof of \Cref{thm:value_based_EJR}. We again let $\alpha_v=\alpha^{\pay{Value-Based}}$, and let the outcome of $\alg{MES}^\pay{Value-Based}(B_{\alg{MES}},P_0)$ be $P^*$. Assume for contradiction that this outcome violates Strong EJR+ up to any project, witnessed by project $p\notin P^*$ and voter set $N'\subseteq N$ with: 
    
    $$\sum_{p'\in (P^*\cup \{p\})\cap A_i} \mu(p')\min \bigg\{\frac{c(p')}{|N'\cap N_{p'}|\mu(p')},\frac{c(p)}{|N'\cap N_{p}|\mu(p)}\bigg\}\leq\alpha_v \frac{B}{n}, \ \forall i \in N'.$$ 
    
    Let $v^*$ be the threshold value, from our definition of the \pay{Value-Based} pre-allocation method.

    \paragraph{Case 1:} $v(p)\leq v^*$

    We begin by showing that each voter $i\in N'$ pays at most $\mu(p')\min\Big\{\frac{c(p')}{|N'\cap N_{p'}|\mu(p')},\frac{c(p)}{|N'\cap N_{p}|\mu(p)}\Big\}$ for each project $p'\in P^*\cap A_i$.
    
    From the definition of the \pay{Value-Based} method, we know that during the pre-allocation step any voter $i$ pays exactly $\frac{\mu(p')}{\max\{v(p'),v^*\}}$ for each project $p'\in P_0\cap A_i$ they approve of, which is bounded from above as follows. 
    \begin{align*}
        \frac{\mu(p')}{\max\{v(p'),v^*\}}
        &\leq \frac{\mu(p')}{\max\{v(p'),v(p)\}} \\
        &=\frac{\mu(p')}{\max\{|N_{p'}|\mu(p')/c(p'),| N_p|\mu(p)/c(p)\}} \\
        &=\mu(p')\min\Big\{\frac{c(p')}{|N_{p'}|\mu(p')},\frac{c(p)}{|N_{p}|\mu(p)}\Big\} \\
        &\leq \mu(p')\min\Big\{\frac{c(p')}{|N'\cap N_{p'}|\mu(p')},\frac{c(p)}{|N'\cap N_{p}|\mu(p)}\Big\}
    \end{align*}

    Similarly, following \citet{Skow26a}, we can also show this bound for any project $p'\in (P^*\setminus P_0)\cap A_i$, i.e., any project approved by voter $i\in N'$ that was selected during the execution of \alg{MES}. Consider the first time a voter $i\in N'$ would pay more than this value for some project $p^*\in P^*$ and let the set of projects selected just before it be $P^{**}\subseteq P^*$. Let $\hat{p}=p$ if 
    $$\frac{c(p^*)}{|N'\cap N_{p^*}|\mu(p^*)}> \frac{c(p)}{|N'\cap N_{p}|\mu(p)},$$ 
    and let $\hat{p}=p^*$ otherwise. From our proportionality violation assumption, observing that $\hat{p}\in (P^*\cup \{p\})\cap A_i$ for each voter $i\in N'\cap N_{\hat{p}}$, and the fact that $\pi_i+b_i\geq \alpha_v\frac{ B}{n}$ for all voters $i\in N$ we know that at this point each voter $i\in N'\cap N_{\hat{p}}$ has at a remaining budget of at least $$\mu_i(\hat{p})\frac{c(\hat{p})}{|N'\cap N_{\hat{p}}|\mu_i(\hat{p})}=\frac{c(\hat{p})}{|N'\cap N_{\hat{p}}|}.$$

    Thus, at this point in the execution of \alg{MES} the voters in $N'\cap N_{\hat{p}}$ have enough money to purchase $\hat{p}$ with equal payments, and thus with a strictly lower $\rho$ than $p^*$ was purchased with, which is a contradiction.

    Now that we have shown this bound for all projects in $P^*$, consider the remaining funds of voters in $N'\cap N_p$ after every project in $P^*$ has been selected. From the definition of the \pay{Value-Based} method, we know that for each voter $i\in N$, $\pi_i+b_i\geq \alpha_v \frac{B}{n}$. We also know that they have spent at most:
    
    \begin{align*}
        \sum_{p'\in P^*\cap A_i} \mu(p')\min \bigg\{\frac{c(p')}{|N'\cap N_{p'}|\mu(p')},\frac{c(p)}{|N'\cap N_{p}|\mu(p)}\bigg\} 
        &\leq \alpha_v \frac{B}{n}-\frac{c(p)}{|N'\cap N_{p}|} \\
        &\leq \pi_i+b_i -\frac{c(p)}{|N'\cap N_{p}|}
    \end{align*}

    Thus, the voters in $N'\cap N_p$ had at least $|N'\cap N_{p}|\frac{c(p)}{|N'\cap N_{p}|}=c(p)$ in unspent funds when \alg{MES} terminated, and would have been able to buy $p$, leading to a contradiction.
    
    \paragraph{Case 2:} $v(p) > v^*$

    We define $P'_0, \pi'_i, b'_i,$ and $\alpha'_v$ in the same way as the proof of \Cref{thm:value_based_EJR}. We claim that the following inequalities hold:
    \begin{enumerate}[label={(\arabic*)}]
        \item $\alpha'_v\geq \alpha_v$, \label{strong_eq1}
        \item $\sum_{p'\in (P_0'\cup \{p\})\cap A_i} \mu(p')\min \big\{\frac{c(p')}{|N'\cap N_{p'}|\mu(p')},\frac{c(p)}{|N'\cap N_{p}|\mu(p)}\big\}\leq(\pi'_i+b'_i)$ for every $i\in N'$, \label{strong_eq2}
        \item $\pi'_i\leq \sum_{p'\in P'_0\cap A_i} \mu(p')\min \big\{\frac{c(p')}{|N'\cap N_{p'}|\mu(p')},\frac{c(p)}{|N'\cap N_{p}|\mu(p)}\big\}$ for every $i\in N'$, and \label{strong_eq3}
        \item $b'_i<\frac{c(p)}{|N'\cap N_{p}|}$ for some $i\in N'\cap N_p$. \label{strong_eq4}
    \end{enumerate}

    The argument for \ref{strong_eq1} is the same as in the proof of \Cref{thm:value_based_EJR}.

    For \ref{strong_eq2}, recall that $P'_0\subset P^*$. Thus we know that for each voter $i\in N'$ the following holds using (1) and our assumption that Strong EJR+ is violated.
    \begin{align*}
        &\sum_{p'\in (P_0'\cup \{p\})\cap A_i} \mu(p')\min \big\{\frac{c(p')}{|N'\cap N_{p'}|\mu(p')},\frac{c(p)}{|N'\cap N_{p}|\mu(p)}\big\} \\
        \leq \ &\sum_{p'\in (P^*\cup \{p\})\cap A_i} \mu(p')\min \big\{\frac{c(p')}{|N'\cap N_{p'}|\mu(p')},\frac{c(p)}{|N'\cap N_{p}|\mu(p)}\big\}\\
        \leq \ & \alpha_v\frac{B}{n} \\
        \leq \ & \alpha'_v\frac{B}{n} 
    \end{align*}
    
    Further, from the definition of minimum voter budget share: $\alpha'_v\leq (\pi'_i+b'_i)\frac{n}{B}$. Combining these, we obtain the statement above.

    For \ref{strong_eq3}, note that each project $p'\in P'_0\subseteq P_0$ has $v(p')\geq v(p)>v^*$, and the cost of those projects was split equally amongst their supporters, giving us the following. 

    \begin{align*}
    \pi'_i
    &= \sum_{p'\in P'_0\cap A_i} \frac{c(p')}{|N_{p'}|}\\
    &= \sum_{p'\in P'_0\cap A_i} \frac{\mu(p')}{\max\{v(p'),v(p)\}} \\
    &\leq \sum_{p'\in P'_0\cap A_i} \mu(p')\min \big\{\frac{c(p')}{|N'\cap N_{p'}|\mu(p')},\frac{c(p)}{|N'\cap N_{p}|\mu(p)}\big\} \\
    \end{align*}

    For \ref{strong_eq4}, observe that $c(p)>B_{\alg{MES}}-c(P'_0)=\alpha'B=\sum_{i\in N} b'_i\geq \sum_{i\in N'\cap N_p} b'_i$. Therefore, there exists $i\in N'\cap N_p$ such that $\frac{c(p)}{|N'\cap N_p|}> b'_i$.

    Finally, note that the RHS of \ref{strong_eq3} and \ref{strong_eq4} sums to the LHS of \ref{strong_eq2}. Thus, combining \ref{strong_eq2}, \ref{strong_eq3}, and \ref{strong_eq4}, we again obtain a contradiction.
\end{proof}

\section{Proportionality Violations}\label{app:negative_results}

In order to formulate the negative results in this section, we will use the proportionality notion of Extended Justified Representation (EJR), and its variants, from PB literature.

\begin{definition} \label{def:cohesiveness}
    Let $T\subseteq P$ and $N'\subseteq N$. We say that voter group $N'$ is \emph{$T$-cohesive} if and only if $T\subseteq \bigcap_{i\in N'} A_i$ and $c(T)\leq \frac{|N'|}{n}B$. We say that $N'$ is $\alpha$-budget $T$-cohesive if and only if $T\subseteq \bigcap_{i\in N'} A_i$ and $c(T)\leq \frac{|N'|}{n}\alpha B$
\end{definition}

\begin{definition}[Extended to PB by \citet{PPS21a}]
    We say that an outcome $P^*\subseteq P$ satisfies \emph{Extended Justified Representation} (EJR) if, for every $T$-cohesive group $N'$, either $T\subseteq P^*$ or there exists a voter $i\in N'$ such that $c(A_i \cap P^*)\geq c(T)$. Following \Cref{def:alpha-budget}, the outcome satisfies $\alpha$-budget EJR if this is instead true for every $\alpha$-budget $T$-cohesive group.
\end{definition}

A feasible outcome satisfying EJR always exists, but cannot be computed in polynomial time, unless P=NP \citep{PPS21a}, which motivated the following relaxation:

\begin{definition}\label{def:EJR1}
    We say that an outcome $P^*\subseteq P$ satisfies Extended Justified Representation up to one project (EJR1) if, for every $T$-cohesive group $N'$, either $T\subseteq P^*$ or there exists a voter $i\in N'$ and a project $p\in A_i \cap (P\setminus P^*)$ such that $c(A_i \cap P^*)+c(p)>c(T)$. The outcome satisfies $\alpha$-budget EJR1 if this is instead true for every $\alpha$-budget $T$-cohesive group.
\end{definition}

The outcome of $\alg{MES}(B_{\alg{MES}},\emptyset)$ always satisfies EJR1, and is computable in polynomial time. EJR+ up to any project implies EJR1 \citep{BrPe23a}.\footnote{We are skipping over some intermediate notions between the two, such as EJR up to any project (EJRx), which are not needed for our results\,---\,see, e.g., \citet{ReMa23a} for an overview.}

We generalize the definition of EJR1 as follows.

\begin{definition}\label{def:EJRk}
    We say that an outcome $P^*\subseteq P$ satisfies \emph{EJR up to $k$ projects (EJRk)} if, for every $T$-cohesive group $N'$, either $|T\cap  P^*|> |T|-k$\footnote{"$|T\cap  P^*|> |T|-k$" represents a generalization of "$T\subseteq P^*$" from the definition of EJR1.} or there exists a voter $i\in N'$ and a set of projects $P'\subseteq A_i \cap (P\setminus P^*)$ with $|P'|\leq k$ such that $c(A_i\cap P^*)+c(P')>c(T)$. $P^*$ satisfies \emph{$\alpha$-budget EJRk} if the above instead holds for every $\alpha$-budget $T$-cohesive group $N'$.
\end{definition}

The idea of EJRk is that either (i) $N'$ is at most $k-1$ projects away from getting $T$, the set of projects they are cohesive over, or (ii) we can give some voter in $N'$ an additional $k$ projects (possibly sourced from outside of $T$) to make them strictly better off than they would be from getting $T$. EJRk reduces to EJR1 when $k=1$. An "up to $k$ projects" style notion has not been considered for the PB setting, and we define it here analogously to the definition of envy-freeness up to $k$ goods in fair division literature (see, e.g., \citet{Waru21a}). 

We can show that EJRk is a weaker axiom than $EJR$ up to any $k$ projects.

\begin{proposition}\label{prop:EJR_relationships}
    Fix $P^*\subseteq P$ and $k\in \mathbb{N}^+$. If $P^*$ satisfies EJR+ up to any $k$ projects, then $P^*$ satisfies \emph{EJRk}.
\end{proposition}

\begin{proof}
    Suppose $P^*$ satisfies EJR+ up to any $k$ projects for some $k\geq 1$.
    Let $T\subseteq P$ and consider a $T$-cohesive group $N'\subseteq N$ with $|T\cap P^*|\leq |T|-k$. Then, consider a $k$-size subset of $T\setminus P^*$ and call it $P'$. We know from the definition of EJR+ up to any $k$ projects that $c(A_i\cap P^*)+c(P')> \frac{|N'|B}{n}$ and we are done.
\end{proof}

The relationships between the proportionality notions restated and introduced in this paper are summarized in \Cref{fig:proportionality_diagram}. Our primary motivation for introducing weaker proportionality axioms is to show that some of our pre-allocation methods can perform arbitrarily badly, from a proportionality perspective. 

\tikzstyle{standard} = [rectangle, minimum width=2cm, minimum height=1cm, text centered, draw=black]
\tikzstyle{standard1} = [rectangle, style=double, minimum width=2cm, minimum height=1cm, text centered, draw=black]
\tikzstyle{etc} = [rectangle, minimum width=2cm, minimum height=1cm, text centered]
\tikzstyle{arrow} = [thick,->,>=stealth]

\tikzstyle{standard2} = [rectangle, style=double, minimum width=3.5cm, minimum height=1cm, text centered, draw=blue, text=blue]
\tikzstyle{arrow2} = [thick,->,>=stealth]

\tikzstyle{arrow12} = [thick,->,>=stealth]
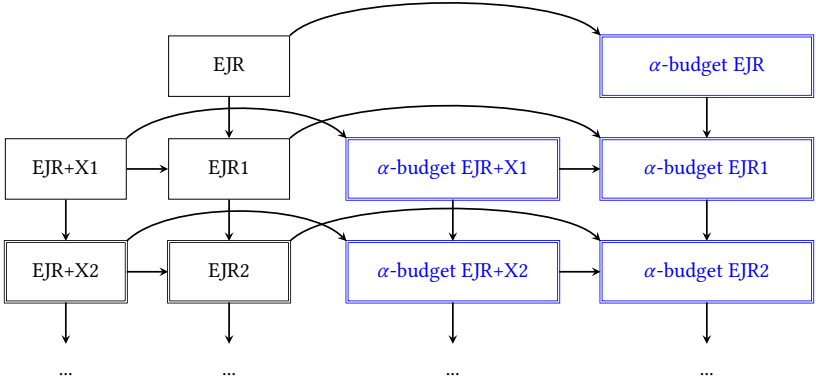
\begin{figure}
    \centering
    \scalebox{0.8}{
    \begin{tikzpicture}[node distance=1.7cm]
        \node (EJR+x) [standard] {EJR+X1};
        \node (EJR+x2) [standard1, below of=EJR+x] {EJR+X2};
        \node (etc1) [etc, below of=EJR+x2] {...};
        \node (EJR1) [standard, right of=EJR+x, xshift=1cm] {EJR1};
        \node (EJR2) [standard1, below of=EJR1] {EJR2};
        \node (etc2) [etc, below of=EJR2] {...};
        \node (EJR) [standard, above of=EJR1] {EJR};
     
        \draw [arrow] (EJR) -- (EJR1);
        \draw [arrow] (EJR+x) -- (EJR+x2);
        \draw [arrow] (EJR+x2) -- (etc1);
        \draw [arrow] (EJR1) -- (EJR2);
        \draw [arrow] (EJR2) -- (etc2);
        \draw [arrow] (EJR+x) -- (EJR1);
        \draw [arrow] (EJR+x2) -- (EJR2);
        
        \node (aEJR+x) [standard2, right of=EJR1, xshift=2cm] {$\alpha$-budget EJR+X1};
        \node (aEJR+x2) [standard2, below of=aEJR+x] {$\alpha$-budget EJR+X2};
        \node (aetc1) [etc, below of=aEJR+x2] {...};
        \node (aEJR1) [standard2, right of=aEJR+x, xshift=2.5cm] {$\alpha$-budget EJR1};
        \node (aEJR2) [standard2, below of=aEJR1] {$\alpha$-budget EJR2};
        \node (aetc2) [etc, below of=aEJR2] {...};
        \node (aEJR) [standard2, above of=aEJR1] {$\alpha$-budget EJR};

        \draw [arrow2] (aEJR) -- (aEJR1);
        \draw [arrow2] (aEJR+x) -- (aEJR+x2);
        \draw [arrow2] (aEJR+x2) -- (aetc1);
        \draw [arrow2] (aEJR1) -- (aEJR2);
        \draw [arrow2] (aEJR2) -- (aetc2);
        \draw [arrow2] (aEJR+x) -- (aEJR1);
        \draw [arrow2] (aEJR+x2) -- (aEJR2);

        \draw [arrow12] (EJR2.north east) to [out=45,in=135,looseness=0.5] (aEJR2.north west);
        \draw [arrow12] (EJR+x.north east) to [out=45,in=135,looseness=0.65] (aEJR+x.north west);
        \draw [arrow12] (EJR1.north east) to [out=45,in=135,looseness=0.5] (aEJR1.north west);
        \draw [arrow12] (EJR+x2.north east) to [out=45,in=135,looseness=0.65] (aEJR+x2.north west);
        \draw [arrow12] (EJR.north east) to [out=45,in=135,looseness=0.5] (aEJR.north west);        
    \end{tikzpicture}
    }
    \caption{Relationships between PB proportionality notions, for $\alpha\leq 1$. We refer to "EJR+ up to any $k$ projects" by the abbreviation \textit{EJR+Xk} (so that EJR+X1 corresponds to EJR+X in \Cref{sec:experiments}). Nodes with a double border correspond to axioms that have been proposed in this paper.}
    \label{fig:proportionality_diagram}
\end{figure}

We restate stronger versions of the results that make up \Cref{thm:violations} using the weaker EJR-based proportionality notions defined above. Let $P^*=\alg{MES}^{\pay{M}}(B_{\alg{MES}},P_0)$ be the outcome of \alg{MES} with available budget share $\alpha=\frac{B_{\alg{MES}}-c(P_0)}{B}$. $P^*$ does not necessarily satisfy the proportionality measure corresponding to that pre-allocation method in \Cref{tab:negative_results}, with respect to the available budget share $\alpha$. Note that each of the counterexamples we construct in this section produce an exhaustive outcome, and thus these violations are not a consequence of \alg{MES} not spending enough of its available budget.

\begin{table}[h!]
\small
\centering
\renewcommand{\arraystretch}{1.5}
\begin{tabular}{ l p{7cm} p{3.9cm} }
\toprule
 \textbf{Method \pay{M}} & \textbf{Guarantees} & \textbf{Violated Properties} \\
 \midrule
 \pay{Null} & $\alpha^{\pay{Null}}$-budget EJR+ up to any project (\Cref{thm:null_EJR}) & $\alpha$-budget EJR (\Cref{prop:null_value_based_not_EJR})   \\ \hdashline
 \pay{Value-Based} & $\alpha^{\pay{Value-Based}}$-budget EJR+ up to any project (\Cref{thm:value_based_EJR}) & $\alpha$-budget EJR (\Cref{prop:null_value_based_not_EJR})   \\ \hdashline
 \pay{Equal-Split} & $\alpha^{\pay{Equal-Split}}$-budget EJR+ up to any 2 projects \newline (\Cref{thm:equal_split_EJR})$^\dagger$ & $\alpha$-budget EJR1 (\Cref{prop:equal_split_not_EJR})  \newline   $\alpha$-budget EJRk (\Cref{prop:equal_split_not_EJRk})$^\ddag$ \\ \hdashline
 \pay{MES-Style} & -- & $\alpha$-budget EJRk (\Cref{prop:mes_style_fails_EJRk}) \\ 
 \bottomrule
\end{tabular}
\caption{Proportionality guarantees and violations for the outcome $P^*=\alg{MES}^{\pay{M}}(B_{\alg{MES}},P_0)$. The marker "$^\dagger$" indicates that the corresponding results hold for the special case in which $P_0$ is selected by \alg{Greedy}. Results in the "Violated Properties" column hold even if $P_0$ selected by \alg{Greedy}, except the result indicated by "$^\ddag$" which only holds without this restriction.
}
\label{tab:negative_results}
\end{table}

\begin{proposition}\label{prop:null_value_based_not_EJR}
    Suppose $R\in \mix$ is MES with any pre-allocation method and available budget share $\alpha$. Then the outcome of $\mix$ does not necessarily satisfy $\alpha$-budget EJR.
\end{proposition}

This follows directly from the fact that the outcome of MES doesn't satisfy EJR \citep{PPS21a}, by considering $R(\alpha B,\emptyset)$.

\begin{proposition}\label{prop:equal_split_not_EJRk}
    Consider a mixed voting rule $\mix$ such that $\alg{MES}^\pay{Equal-Split} \in \mix$ with available budget share $\alpha$. Then, for any two arbitrary positive integers $l, k\in \mathbb{N}^+$, the outcome of $\mix$ does not necessarily satisfy $\frac{\alpha}{l}$-budget EJR up to $k$ projects. 
\end{proposition}

\begin{proof}
    Let $l, k \in \mathbb{N}^+$ be two positive integers. We construct a PB instance $I$, which shows that the outcome of $\mix$ does not satisfy $\frac{\alpha}{l}$-budget EJR up to $k$ projects.
    
    \begin{table}
        \centering
        \begin{tabular}{l c c c c c c c c c}
            \toprule
            Project & $p_1$ & $p_2$ & \dots & $p_{\frac{n}{2}}$ & $p_1'$ & $p_2'$ & \dots & $p_{2k}'$ & $p''$  \\
            \midrule
            Cost  & 1 & 1 & \dots & 1 & $\sfrac{n}{8kl}$ & $\sfrac{n}{8kl}$ & \dots & $\sfrac{n}{8kl}$ & $\sfrac{n}{2}$ \\
            \midrule
            Number of approvals & $1$ & $1$ & $\dots$ & $1$ & $\sfrac{n}{2}$ & $\sfrac{n}{2}$ & $\dots$ & $\sfrac{n}{2}$ & $\sfrac{n}{2}$\\
            \midrule
            $A_1$ & \x &  &  &  & \x & \x & $\cdots$ & \x &    \\
            $A_2$ &  & \x &  &  & \x & \x & $\cdots$ & \x &    \\
            \phantom{x}$\vdots$ &  &  & $\ddots$ & & $\vdots$ & $\vdots$ & $\ddots$ & $\vdots$ \\
            $A_{\frac{n}{2}}$ &  &  &  & \x & \x & \x & $\cdots$ & \x &    \\
            $A_{\frac{n}{2}+1}$ &  &  &  &  &  &  &  &  & \x \\
            $A_{\frac{n}{2}+2}$ &  &  &  &  &  &  &  &  & \x \\
            \phantom{x}$\vdots$ &  &  &  &  &  &  &  &  & $\vdots$ \\
            $A_{n}$ &  &  &  &  &  &  &  &  & \x \\
            \bottomrule
        \end{tabular}
                \caption{Example instance with a budget of $n$ for the proof of \Cref{prop:equal_split_not_EJRk}.} 
        \label{tab:equal_split_not_EJRk}
    \end{table}
    
        Consider the PB instance $I = (B, P, A, c)$ with an even number of voters $n = 8l$ and instance budget $B=n$. Let $P$ contain the following projects:
        \begin{itemize}
            \item $p_i$ for $1\leq i \leq \frac{n}{2}$, with $c(p_i)=1$ and $N_{p_i}=\{i\}$
            \item $p'_j$ for $1\leq j\leq 2k$ with $c(p'_j)=\frac{n}{8k l }$ and $N_{p'_j}=\{1,\dots,\frac{n}{2}\}$
            \item $p''$ with $c(p'')=\frac{n}{2}$ and $N_{p''}=\{\frac{n}{2}+1,\dots,n\}$
        \end{itemize}
        Note, that the approval sets in $A$ are implicitly defined through the approving voter sets $N_p$ for all projects $p \in P$.
        We let $P_0=\{p_1,\dots,p_{\frac{n}{2}}\}$ and consider $\alg{MES}^{\pay{Equal-Split}}(B,P_0)$ which has an available budget share $\alpha=0.5$. The \pay{Equal-Split} method outputs the following voter payments and budgets:
        \begin{itemize}
            \item $\pi_i=1$ and $b_i=0$ for $1\leq i\leq \frac{n}{2}$ and
            \item $\pi_i=0$ and $b_i=1$ for $\frac{n}{2}< i\leq n$.
        \end{itemize}
    
        $\alg{MES}^{\pay{Equal-Split}}$ selects $p''$ and terminates with outcome $P^*=\{p_1,\dots,p_{\frac{n}{2}},p''\}$ with $c(P^*)=n$. 
        
        Now, we choose $T=\{p'_1,\dots p'_{2k}\}$ and $N'=\{1,\dots,\frac{n}{2}\}$ with $c(T)=\frac{n}{4l} = \frac{\frac{0.5}{l} \frac{n}{2} n}{n} = \frac{\frac{\alpha}{l} |N'|B}{n}$. Thus, $N'$ is $\frac{\alpha}{l}$-budget $T$-cohesive (compare \Cref{def:cohesiveness}). However, we have $|T \cap P^\ast| = 0 < k = |T| - k$ and for any voter $i \in N'$ and any set of projects $P' \subseteq A_i \cap (P \setminus
        P^\ast) = T$ with $|P'| \leq k$, we have
        $$
            c(A_i \cap P^*) + c(P') \leq c(p_i) + k\frac{n}{8kl} = 1 + \frac{n}{8l} = 2 = \frac{n}{4l} = c(T).
        $$
        Thus, the outcome of $\mix$ violates $\frac{\alpha}{l}$-budget EJR up to $k$ projects (compare \Cref{def:EJRk}).
\end{proof}

\begin{proposition}\label{prop:mes_style_fails_EJRk}
    Consider a mixed voting rule $\mix$ such that $\alg{MES}^\pay{MES-Style} \in \mix$ with available budget share $\alpha$. Then, for any arbitrary positive integer $k\in \mathbb{N}^+$, the outcome of $\mix$ does not necessarily satisfy $\alpha$-budget EJR up to $k$ projects. This holds even when \alg{MES} is the second rule in $\mix$ and the first rule in $\mix$ is \alg{Greedy} (with tie-breaking in favor of large projects).
\end{proposition}

\begin{proof}
    Let $k\in \mathbb{N}^+$ be an arbitrary positive integer. We construct a PB instance $I$, which shows that the outcome of $\mix$ does not satisfy $\alpha$-budget EJR up to $k$ projects.
    
    Let $I = (B, P, A, c)$ be a PB instance with $n = 100$ voters and a budget $B = 100$. Let parameter $\lambda \in \mathbb{N}$ be defined as $\lambda = \lceil \frac{k}{4} \rceil$.
    Let $P$ contain the following subsets of projects:
    \begin{itemize}
        \item \textbf{Type 1}: For $1\leq j \leq 50$ define $p_j$ with cost $c(p_j)=1$ and 40 approving voters $N_{p_j}=\{j\}\cup \{62,\dots,100\}$.
        \item \textbf{Type 2}: For every unique 10-person subset of voters $\bar{N} \subset \{1,\dots,50\}$ define $5 \lambda$ projects $p^{\bar{N}}_j$ with $1\leq j \leq 5 \lambda$, with cost $c(p^{\bar{N}}_j)=\frac{1}{\lambda} > 0$, approved by the $10$ voters in $\bar{N}$, i.e., $N_{p^{\bar{N}}_j} = \bar{N}$.
        \item \textbf{Type 3}: For every unique 40-person subset of voters $\hat{N} \subset \{1,\dots,50\}$ define $40 \lambda$ projects $p^{\hat{N}}_j$ with $1 \leq j \leq 40 \lambda$ with cost $c(p^{\hat{N}}_j) = \frac{39}{40 \lambda}$, approved by the $40$ voters in $\hat{N}$, i.e., $N_{p^{\hat{N}}_j} = \hat{N}$.
        \item \textbf{Type 4}: Define one project $p'$ with cost $c(p') = 11$, approved by $11$ voters $N_{p'} = \{51,\dots,61\}$.
    \end{itemize}
    Note, that the approval sets in $A$ are implicitly defined through the approving voter sets $N_p$ for all projects $p \in P$.

    \begin{table}[h]
        \centering
        \begin{tabular}{l c c c c c c c c}
            \toprule
            Project & $p_1$ & $p_2$ & \dots & $p_{50}$ & Type 2 & Type 3 & $p'$\\
            \midrule
            Cost & $1$ & $1$ & \dots & $1$ & $1/\lambda$ & $\sfrac{39}{40 \lambda}$ & $11$ \\
            \midrule
            Quantity &  &  &  &  & $\binom{40}{10} \cdot 5 \lambda$ & $\binom{50}{40} \cdot 40 \lambda$ & \\
            \midrule
            Number of approvals & $40$ & $40$ & $\dots$ & $40$ & $10$ & $40$ & $11$\\
            \midrule
            $A_1$     & \x &   &          &   & (\x) & (\x) \\
            $A_2$     &   & \x &          &   & (\x) & (\x) \\
            \phantom{x}$\vdots$  &   &   & $\ddots$ &   & $\vdots$ & $\vdots$  \\
            $A_{50}$  &   &   &          & \x & (\x) & (\x) \\
            $A_{51}$  &   &   &   &   &   &   & \x \\
            $A_{52}$  &   &   &   &   &   &   & \x \\
            \phantom{x}$\vdots$  &   &   &   &   &   &   & $\vdots$ \\
            $A_{61}$  &   &   &   &   &   &   & \x \\
            $A_{62}$  & \x & \x & $\cdots$ & \x \\
            $A_{63}$  & \x & \x & $\cdots$ & \x \\
            \phantom{x}$\vdots$ & $\vdots$ & $\vdots$ & $\ddots$ & $\vdots$ \\
            $A_{100}$  & \x & \x & $\cdots$ & \x \\
            \bottomrule
        \end{tabular}
        \caption{Example instance with a budget of $b = 100$ for the proof of \Cref{prop:mes_style_fails_EJRk}. Each project of type 2 is approved by $10$ voters in $\{1,\dots,40\}$ and each project of type 3 is approved by $40$ voters in $\{1,\dots,50\}$.}
        \label{tab:mes_style_fails_EJR}
    \end{table}

    Let $P_0=\{p_1,\dots,p_{50}\}$, which means that \alg{MES}'s available budget share is $\alpha=0.5$. Note, that $P_0$ is exactly the set of projects that would be chosen by $\alg{Greedy}(50,\emptyset)$ in this instance (with tie-breaking in favor of large projects). Consider the \pay{MES-Style} pre-allocation method, with some arbitrary order of $P_0$. The method funds the first 40 projects in this ordering equally, with supporters paying $\frac{1}{40}$ per project. At this point, all voters in $\{62,\dots,100\}$ run out of money and the remaining projects are funded solely by their supporters from $i\in \{1,\dots, 50\}$. Call the set of these $10$ voters $N^*$.

    The \pay{MES-Style} pre-allocation method selects voter payments and induces voter budgets as follows:
    \begin{itemize}
        \item $\pi_i=1$; $b_i=0$ for $i \in N^* \cup \{62,\dots,100\}$,
        \item $\pi_i=\frac{1}{40}, b_i=\frac{39}{40}$ for $i\in \{1,\dots,50\}\setminus N^*$, and
        \item $\pi_i=0$, $b_i=1$ for $i \in \{51,\dots,61\}$.
    \end{itemize}

    This yields a minimum voter budget share of $\alpha_v=1$. \alg{MES} then selects the $40 \lambda$ type 3 projects supported by $\{1,\dots,50\}\setminus N^*$, and then the Type 4 project $p'$, yielding an output $P^*$ with $c(P^*)=100$.

    Each voter $i\in N^*$ obtains a (cost) satisfaction of $\mu_i(P^*) = c(A_i \cap P^*) = c(p_i) = 1$. However, there exist $5 \lambda$ commonly approved unchosen projects $T=\{p^{N^*}_j\}_{1\leq j\leq 5 \lambda}\subseteq \bigcap _{i\in N^*}A_i $ with $c(T)=5 \lambda\frac{1}{\lambda} = 5 = \frac{\alpha |N^*|B}{n}$. Thus, $N^*$ is $\alpha$-budget $T$-cohesive (compare \Cref{def:cohesiveness}).
    
    However, we have $|T \cap P^\ast| = 0 < 5 \frac{k}{4} - k \le 5 \lambda - k = |T| - k$ and for any voter $i \in N^*$ and any set of projects $P' \subseteq A_i \cap (P \setminus
    P^\ast)$ with $|P'| \leq k$, we have
    $$
        c(A_i \cap P^*) + c(P') \le 1 + k \cdot \frac{1}{\lambda} \le 1 + k \cdot \frac{4}{k} = 5 = c(T).
    $$
    Thus, the outcome of $\mix$ violates $\alpha$-budget EJR up to $k$ projects (compare \Cref{def:EJRk}).
\end{proof}

This example can be tweaked to show that the $\frac{\alpha}{l}$-budget EJR up to $k$ projects is not satisfied, for any arbitrarily high $l\in \mathbb{N^+}$ by choosing an appropriately large number of voters, similarly to the proof of \Cref{prop:equal_split_not_EJRk}.

\Cref{prop:mes_style_fails_EJRk} demonstrates that in order to achieve any kind of proportionality guarantees, it might be necessary to give a budget of less than their fair share of the $\alg{MES}$ budget, $\frac{B_{\alg{MES}}}{n}$, to initially empty-handed voters (those who approve no projects in $P_0$), even if the cost of $P_0$ doesn't exceed the endowments of non-empty-handed voters. This might seem unfair, but it is in this case preferable to charging the cost of an entire project to one voter.

\begin{proposition}\label{prop:equal_split_not_EJR}
    Consider a mixed voting rule $\mix = [R_j]_{1\leq j\leq m}$ with $R_1 = \alg{Greedy}$ and $R_2 = \alg{MES}^\pay{Equal-Split}$ with available budget share $\alpha$. Then the outcome of $\mix$ does not necessarily satisfy $\alpha$-budget EJR up to one project.
\end{proposition}

\begin{proof}

    We construct a PB instance $I$, which shows that the outcome of $\mix$ does not satisfy $\alpha$-budget EJR1.
    
    Let $I = (B, P, A, c)$ be a PB instance with $n=100$ voters and budget $B=100$. 
    Let $P$ contain the following projects:
    
    \begin{itemize}
        \item $p$ with cost $c(p)=48.6$ and 11 approving voters $N_{p}=\{90,\dots, 100\}$
        \item $p_j$ for $1\leq j \leq 10$ with cost $c(p_j)=0.14$ and 1 approving voter $N_{p_j}=\{j\}$.
        \item $p'$ with cost $c(p')=2.85$, and 10 approving voters $N_{p'}=\{1,\dots,10\}$
        \item $p''_j$ for $1\leq j \leq 3$ with cost $c(p''_j)=1.5$, and 9 approving voters $N_{p''_j}=\{1,\dots,9\}$
        \item $\hat{p}_j$ for $1\leq j \leq 19$ with cost $c(\hat{p}_j)=2.31$ and 4 approving voters $N_{\hat{p}_j}=\{4j+7,\dots, 4j+10\}$
        \item $\hat{p}_{20}$ with cost $c(\hat{p}_{20})=1.88$ and 4 approving voters $N_{\hat{p}_{20}}=\{10,87,88,89\}$

    \end{itemize}
    
    Note, that the approval sets in $A$ are implicitly defined through the supporter sets $N_p$ for all projects $p \in P$.

    Consider the output of $\alg{MES}^\pay{Equal-Split}(100,\alg{Greedy}(50,\emptyset))$.
    \alg{Greedy} first chooses $p$, and then chooses $\{p_1,...,p_{10}\}$ in some order, as it can no longer afford any other projects in $P$. $c(\{p,p_1,\dots,p_{10}\})=50$, which provides $\alg{MES}$ with an available budget share of $\alpha=0.5$.
    
    The \pay{Equal-Split} pre-allocation method assigns budgets as follows:
    
    \begin{itemize}
        \item $\pi_i=0.14$ and $0.4375\leq b_i\leq 0.4376$ for $1\leq i \leq 10$, 
        \item $\pi_i=0$ and $0.5775\leq b_i\leq 0.5776$ for $11 \le i \le 89$, and
        \item $\pi_i=\frac{48.6}{11}$ and $b_i=0$ for $90\leq i \leq 100$.
    \end{itemize}
    
    \alg{MES} then selects $p'$, which reduces the budgets of each voter $i\in \{1,\dots,10\}$ to $b_i\approx 0.1525$. Importantly, this means that the voters in $\{1,\dots,9\}$ can no longer afford any project from $\{p''_1,p''_2,p''_3\}$. \alg{MES} then selects $\{\hat{p}_1,\dots, \hat{p}_{20}\}$ and terminates with an exhaustive outcome $P^*=\{p,p_1,\dots,p_{10},p',\hat{p}_1,\dots,\hat{p}_{20}\}$, having spent $48.62$ units of its available budget.

    Consider $T=\{p''_1,p''_2,p''_3\}$. Voter group $N'=\{1,\dots,9\}$ is $0.5$-budget $T$-cohesive as $0.5 \frac{|N'|}{n}B=4.5=c(T)$, and each voter $i\in N'$ has $c(A_i\cap P^*)=c(p')+c(p_i)=2.99$, and thus for any project $p\in P\setminus P^*=\{p''_1,p''_2,p''_3\}$, $c(A_i\cap P^*)+c(p)= 2.99+1.5<c(T)$.
    
    Thus $P^*$ does not satisfy $\alpha$-budget EJR up to one project (compare \Cref{def:EJR1}).
\end{proof}

\section{Voter-Specific Threshold Values for the Value-Based Pre-Allocation Method} \label{app:value_based}

We can weakly improve upon the minimum voter budget share of the \pay{Value-Based} method by defining a set of per-voter threshold values instead of a universal one, without sacrificing the guarantees in \Cref{thm:value_based_EJR,thm:value_based_strong_EJR}. In practice this improvement is minor, as the resulting minimum voter budget share will be bounded from above by $\alpha^{\pay{Equal-Split}}$ --- see \Cref{fn:alpha-values} for details on the empirical magnitude of this bound. We can define the per-voter threshold values as follows.

    $$v^*_i=\max_{p\in A_i\cap (P\setminus P_0)} \{v(p) \mid p \text{ satisfies } c(U(p))+c(p)\leq B_{\alg{MES}}\}$$

Note that $v^*_i\leq v^*$ for every voter. We can then tweak the proof of \Cref{thm:value_based_strong_EJR} to use $v(p)\leq \min_{i\in N'}\{v_i^*\}$ for Case 1, and $v(p)> \min_{i\in N'}\{v_i^*\}$ for Case 2, obtaining the same \emph{Strong EJR+ up to any project} proportionality guarantee, but with respect to a weakly higher minimum voter budget share. Note that this does not counteract our tightness results in \Cref{prop:tightness}, as using voter-specific threshold values represent a higher level of information about the instance than we assumed.

\section{MES with Budget Increase} \label{app:budget_increase}

In \Cref{sec:experiments} we briefly outlined a generalization of MES with budget increase to our setting, using what we will now call the \emph{ex-ante} variant. Here we introduce it more formally, consider alternative ways in which we could have generalized this, and show that the ex-ante variant is the only variant that produced a minimum voter budget share with respect to which our proportionality guarantees for the \pay{Value-Based} method still hold.

Letting the budget increment be $\beta$,\footnote{Equivalently we could let the budget increment be $\beta/n$ per voter.} we can calculate the voter budget profile after $x$ budget increase steps, $(b_i^x)_{i\in N}$, by altering the input for our pre-allocation method in \Cref{def:rebalancing_process} to an MES budget share of $\alpha^x =\alpha+\frac{x\beta}{B}$. We write $\alg{MES}^{\pay{M}(x)}$ for \alg{MES} with pre-allocation method \pay{M} and $x$ budget increase steps and $\alpha^{\pay{M}+}$ for the minimum voter budget share produced as a result of this method. \alg{MES} with pre-allocation method \pay{M} and budget increase, $\alg{MES}^{\pay{M}+}$, can then be formulated as follows: find the smallest $x^*\in \naturals$ for which either the outcome of $\alg{MES}^{\pay{M}(x^*)}$ is exhaustive, or the outcome of $\alg{MES}^{\pay{M}(x^*+1)}$ is infeasible. Then, the outcome of $\alg{MES}^{\pay{M}+}$ is the outcome of $\alg{MES}^{\pay{M}(x^*)}$. 

\begin{restatable}{corollary}{budgetincrease}\label{prop:budget_increase_proportionality}
    The proportionality guarantees from \Cref{thm:null_EJR,thm:value_based_EJR,thm:equal_split_EJR} hold for $\alg{MES}^{\pay{M}(x)}$ with respect to $\alpha^{\pay{M}+}$, and $\alpha+\frac{x\beta}{B}\leq \alpha^{\pay{M}+}\leq \frac{B_{MES}+x\beta}{B}$.
\end{restatable}

\begin{proof}
    This follows directly from the proofs of \Cref{thm:null_EJR,thm:value_based_EJR,thm:equal_split_EJR} by altering the \alg{MES} budget share to $\alpha^x =\alpha+\frac{x\beta}{B}$, and correspondingly altering the \alg{MES} budget to $B_{MES}+x\beta$.
\end{proof}

Therefore, the proportionality guarantees from \Cref{thm:null_EJR,thm:value_based_EJR,thm:equal_split_EJR} also hold for $\alg{MES}^{\pay{M}+}$ with respect to $\alpha^{\pay{M}+}=\alpha^{\pay{M}(x^*)}$, with $x^*$ as defined above.

We now consider two other ways to generalize \alg{MES} with budget increase to the mixed voting rule setting. The \emph{intermediate} and \emph{ex-post} variants corresponding to introducing the budget increase at a later stage in the pre-allocation method process described in \Cref{def:rebalancing_process}.

\begin{itemize}
    \item \underline{Ex-ante}: Run the pre-allocation method with an MES budget share of $\alpha^x =\alpha+\frac{x\beta}{B}$, outputting voter budgets $(b_i^x)_{i\in N}$. This is the variant we use in our definition above.
    \item \underline{Intermediate}: Calculate $(\pi_i)_{i\in N}$ as in stage (1) of \Cref{def:rebalancing_process} and then set $B^x=\alpha B+x\beta$. Calculate $(b_i^x)_{i\in N}$ as in stage (2), but replace the constraint $\sum_{i\in N} b_i=\alpha B$ with $\sum_{i\in N} b_i^x=B^x$. 
    \item \underline{Ex-post}: Calculate $(b_i)_{i\in N}$ as in \Cref{def:rebalancing_process}. Then, set $b_i^x=b_i+\frac{x\beta}{n}$ for each voter $i$.
\end{itemize}

We find that the ex-ante variant is the only one that allows the \pay{Value-Based} method to retain its proportionality guarantee.

\begin{proposition}
    The proportionality guarantee from \Cref{thm:value_based_EJR} does not hold for $\alg{MES}^{\pay{Value-Based}(x)}$ with respect to $\alpha^{\pay{Value-Based}(x)}$, when using the intermediate or ex-post variants.
\end{proposition}

\begin{proof}
    Consider a PB instance $I$ with $n=100$ voters, $B=10000$ and project set $P$ containing the following projects.

    \begin{itemize}
        \item $p_0'$ with cost $c(p_0')=5000$ and $N_{p_0'}=N$.
        \item $p_1'$ with cost $c(p_1')=5100$ and $N_{p_1'}=N$.
        \item $p_j$ for $1\leq j\leq 100$ with cost $c(p_j)=50$ and 1 approving voter $N_{p_j}=\{j\}$.
    \end{itemize}

    Let $P_0=\{p_0',p_1,\dots,p_{100}\}$, and consider a budget increase step $\beta=1000$.\footnote{A smaller budget increase step would also work for our argument.}
        
    $\alpha^{\pay{Value-Based}(0)}_{Ex-post}=\alpha^{\pay{Value-Based}(0)}_{Intermediate}=\alpha^{\pay{Value-Based}(0)}_{Ex-ante}=1$, as the threshold value is $v^*=0$ for all three variants.
    
    $\alpha^{\pay{Value-Based}(1)}_{Ex-post}=\alpha^{\pay{Value-Based}(1)}_{Intermediate}=1.1$, as the threshold value does not change from $v^*=0$. However, $\alpha^{\pay{Value-Based}(1)}_{Ex-ante}=0.605$, as $c(p'_0)+c(p'_1)\leq 11000$, so the threshold value increases to $v^*=100$.

    $P_0=\{p_0',p_1,\dots,p_{100}\}$ is still exhaustive for a budget of $11000$, and thus the outcome of $\alpha^{\pay{Value-Based}(1)}$ is $P^*=P_0$. $P^*$ satisfies $0.605$-EJR+ up to any project. However, $p_1'$ certifies an $1.1$-EJR+ up to any project violation for any voter in voter group $N$, as $c(A_i\cap P^*)+c(p_1')=10150<11000=1.1B$. Thus, the proportionality guarantee from \Cref{thm:value_based_EJR} does not hold for $\alg{MES}^{\pay{Value-Based}(x)}$ with respect to $\alpha^{\pay{Value-Based}(x)}$, when using the intermediate or ex-post variants.
\end{proof}

\section{Adapting Other Rules} \label{app:other_rules_theory}

In this section we discuss how to adapt rules other than \alg{Greedy} and \alg{MES} to our mixed rules framework. These adaptations allow us to compare different voting rules empirically, and are presented without theoretical guarantees; a more in depth analysis of this is left for future work.

\paragraph{Sequential Phragmén}
A well-known rule aiming at proportional representation is \textit{Sequential Phragmén} (\seqPhrag) \citep{BFJL24a,LCG22a}. \seqPhrag initially assigns each voter a budget of $b_i = 0$ and a load $\ell_i = 0$. It then continuously increases all voter budgets, until there is a project that can be afforded by its supporting voters, i.e., finds the minimum value $t$, such that for some project $p$ we have $\sum_{i \in N_p} b_i - \ell_i = c(p)$. Project $p$ is then added to the set of selected projects and the loads of all paying voters $i \in N_p$ get updated to $\ell_i = b_i$. This process continues until an exhaustive outcome is reached or none of the unselected projects is supported by any voter.

To account for a set of pre-selected projects, we can compute the initial loads of the voters using the first step of the pre-allocation methods defined in \Cref{sec:adapting_mes}. The rebalancing step is not necessary, since \seqPhrag does not require each voter's load to be bounded by $\frac{B}{n}$.
Since, during the calculation of \seqPhrag there can now be voters with loads higher than $t$, we adapt the affordability condition slightly in order to not count these voters negatively into the budget of the supporters of a project. Namely, in each iteration \seqPhrag needs to find the minimum value for $t$, such that for some project $p$ we have $\sum_{i \in N_p} \max(0, b_i - \ell_i) = c(p)$. Intuitively, this means that the budget of a voter $i$ that already spent $\ell_i$ on pre-selected projects will not increase until time $t = \ell_i$. For any time $t < \ell_i$, this voter can not contribute to buying any project.

\paragraph{Method of Equal Shares with Bounded Overspending}
In a recent paper, \citet{GPS+24a} propose a new rule that is based on \alg{MES}, but produces an almost exhaustive outcome using a rounding procedure that allows voters to spend more than their remaining budget. They call this rule \textit{Method of Equal Shares with Bounded Overspending} (\alg{BOS}). Similarly to \alg{MES}, every voter starts with a budget of $\frac{B}{n}$. An unselected project $p$ is called $(\alpha, \rho)$-affordable if $\alpha \cdot c(p) = \sum_{i\in N_p} \min(b_i, \alpha \cdot \rho \cdot c(p))$. In each iteration \alg{BOS} adds the $(\alpha, \rho)$-affordable project that minimizes $\rho/\alpha$ among all projects that fit into the remaining budget and the voter budgets of all supporting voters $i \in N_p$ are updated to $b_i = \max(0, b_i - \rho \cdot c(p))$. Note, that unlike for \alg{MES}, the remaining budget does not have to match the sum of voter budgets. An intuitive interpretation of this process is that voters are ``offered'' to pay only a fraction $\alpha$ of any project, if this allows them to spread the cost more evenly. When the project is selected, each voter is still charged their full share of this projects cost. If a voter is charged more than they have, their budget is set to zero and the remaining cost is discarded. \alg{BOS} terminates when there are no more unselected projects that fit into the remaining budget, or there are no voters with positive budget who support an unselected project.
For a more detailed explanation of \alg{BOS}, we refer to \citet{GPS+24a}.

We can adapt \alg{BOS} to account for a set of pre-selected projects in the same way as \alg{MES}; by computing initial voter budgets using one of the pre-allocation methods discussed in \Cref{sec:adapting_mes}.

\paragraph{Greedy Chamberlin-Courant}
The \textit{Greedy Chamberlin-Courant} (\alg{GreedyCC}) greedily minimizes the fraction of voters with zero satisfaction, by iteratively selecting an affordable project with the highest number of unsatisfied supporters per unit cost. More formally, given partial outcome $P'\subset P$, \alg{GreedyCC} selects project $p\in P\setminus P'$ maximizing $\lvert \{i\in N_p\mid A_i \cap P' = \emptyset\}\rvert /c(p)$. The rule stops if there are no more affordable projects or all voters have positive satisfaction. Notably, it is the only rule we consider whose objective is neither efficiency nor proportionality. 
Similar to \alg{Greedy}, \alg{GreedyCC} also does not need any adjustments to work with a set of pre-selected projects.\\

We show empirical results for mixing these rules with \alg{Greedy} in \Cref{sec:other_rules} and \Cref{app:other_methods_results}. 

\section{Is \pay{Value-Based} optimal?}\label{app:optimality}
We can use this new proportionality notion to argue that \pay{Value-Based} is, in some sense, optimal. However, we cannot show the optimality of a pre-allocation method for a \emph{specific instance}. In fact, given a particular instance and a set of pre-selected projects, the best way to maximize the $\alpha$ for which `$\alpha$-budget Strong-EJR+ up to one project' is satisfied would be to solve the optimization problem directly. Likewise, we know from \Cref{prop:in_general_baseline} that for \emph{all instances} we cannot improve upon the empirical baseline. Our strategy will instead be to consider the information that a pre-allocation method uses to come up with an initial \alg{MES} budget profile, and argue that, as long as that information is kept constant, we cannot devise a pre-allocation method that would guarantee a higher level of proportionality. In other words, we derive optimality for \emph{classes of instances}. We return to the constrained setting of PB with cost satisfaction for this proof.

\mut{I would personally not state this proposition here and instead just give an explanation of what it says, i.e., just explain in what sense Value-Based is optimal. Depending on the length of it, I could even see this going into the conclusion.}

\begin{restatable}{proposition}{tightness}\label{prop:tightness}    
    Fix a set of voters $N$, a project $p$ and its supporter set $N_p$, threshold value $v^*\leq n-|N_p|$, and \alg{MES} budget $B_{\alg{MES}}\geq c(p)$. 
    Consider the set of instances $\mathcal{I}$ with budget $B$, project set $P\supset P_0=\{p\}$, voter set $N$, with approval profile $A$ such that $N_p=\{i\in N: p\in A_i\}$ and threshold value of $v^*$ under \alg{MES} budget $B_{\alg{MES}}$. Given these assumptions, let the minimum voter budget share produced by the \pay{Value-Based} method be $\alpha_v$. Note that the value of $\alpha_v$ is the same for every instance in $\mathcal{I}$, as our assumptions contain all of the information required to compute the outcome of the \pay{Value-Based} method.
    
    Consider a pre-allocation method \pay{M} that only uses the information about the instance $I$ and pre-selected set $P_0$ that we fix. If \pay{M} guarantees that the outcome of $\alg{MES}^\pay{M}$ achieves $\alpha_v$-Strong EJR+ for every instance $I\in \mathcal{I}$, then must produce the same voter budget profile as the \pay{Value-Based} method.
\end{restatable}

\begin{proof}

Consider an arbitrary instance $I\in \mathcal{I}$ from our constrained class. Let the voter payment profiles produced by \pay{M} and \pay{Value-Based} be $\{\pi_i\}_{i\in N}$ and $\{\pi_i^V\}_{i\in N}$, respectively. Similarly, let the voter budget profiles be $\{b_i\}_{i\in N}$ and $\{b_i^V\}_{i\in N}$.

First note that it must be the case that $\sum_{i\in N} b_i\leq \sum_{i\in N}b_i^V=B_{\alg{MES}}-c(p)$. Otherwise, it would be possible for \alg{MES} to run out of budget and prevent some voter $i'\notin N_p$ from acquiring an individually-approved project.

\medskip

We consider two cases: either $|N_p|\geq  v^*$, or $|N_p|< v^*$.

\paragraph{\textbf{Case 1.}} First suppose $|N_p|\geq v^*$. We can always find some voter $i'\in N$ with $b_{i'}+\varepsilon\leq b_{i'}^V$ for some $\varepsilon>0$ (else $b_i=b_i^V$ for all $i\in N$, which contradicts our assumption that \pay{M} produces a different budget profile). We then define two projects:

\begin{itemize}
    \item $p^*$, with $c(p^*)=b_{i'}+\varepsilon$ and $N_{p^*}=\{i'\}$
    \item $\hat{p}$, with $c(\hat{p})=\varepsilon$ and $|N_{\hat{p}}|=v^*$ with $N_{\hat{p}}\cap N_p=\emptyset$ (recalling our assumption that $v^*\leq n-|N_p|$).
\end{itemize}

Consider an instance from $I\in \mathcal{I}$ with project set $P=\{p,p^*,\hat{p}\}$. The output of $\alg{MES}^\pay{M}(I, B, P_0)$ is a subset of $\{p,\hat{p}\}$, as $i'$ cannot afford $p^*$.

Suppose $i'\notin N_p$, which means that $\pi^V_{i'}=0$, and thus $b_{i'}^V=\frac{\alpha_v B}{n}$. Then $p^*$ with $N'=\{i'\}$ constitutes an EJR+ violation (and therefore a Strong EJR+ violation), as $c(P^*\cap A_{i'})+c(p^*)\leq c(\hat{p})+c(p^*)<b_{i'}^V=\frac{\alpha_v B}{n}$. If instead $i'\in N_p$, we know that $\pi^V_{i'}=\frac{c(p)}{|N_p|}$ from the definition of the \pay{Value-Based} method, as $|N_p|\geq v^*$. Observe that $\pi_{i'}^V+b_{i'}^V=\frac{\alpha_v B}{n}$ as $b_{i'}^V>0$, from the definition of a pre-allocation method. Then $p^*$ with voter group $N'=N_p$ constitutes a Strong EJR+ violation, as we can show the following for any voter $i\in N_p$.

\begin{align*}
    \sum_{p'\in P\cap A_i} c(p') \cdot \min \bigg\{\frac{c(p')}{|N'\cap N_{p'}| \cdot  c(p')},\frac{c(p^*)}{|N'\cap N_{p^*}| \cdot c(p^*)}\bigg\}
    &=\sum_{p'\in P\cap A_i} \frac{c(p')}{\max\{|N_{p}\cap N_{p'}|,|N_{p}\cap N_{p^*}|\}} \\
    &=\sum_{p'\in P\cap A_i} \frac{c(p')}{\max\{|N_{p}\cap N_{p'}|,1\}} \\
    &\leq \frac{c(p)}{|N_p|}+\frac{b_{i'}+\varepsilon}{1} \\
    &\leq \frac{c(p)}{|N_p|}+b_{i'}^V \\
    &=\pi_{i'}^V+b_{i'}^V \\
    &=\frac{\alpha_v B}{n}
\end{align*}

\medskip

\paragraph{\textbf{Case 2.}}  Now suppose $|N_p|< v^*$. First observe that $b_i^V=\frac{\alpha_v B}{n}$ for $i\notin N_p$ and $b_i^V=\frac{\alpha_v B}{n}-\frac{c(p)}{v^*}$ for $i\in N_p$. We can again find some voter $i'\in N$ with $b_{i'}<b_{i'}^V$. The $i'\notin N_p$ case is identical to the above, so we will assume that for each $i\in N\setminus N_p$, $b_i\geq b_i^V$. 

We define two projects:

\begin{itemize}
     \item $p^*$ with $N_{p^*}$ and $c(p^*)$ defined as follows. Since some voter $i'\in N_p$ has $b_{i'}<b_{i'}^V$, we can select a set of voters $N_{p^*}$, such that $|N_{p^*}|=v^*$, $|N_{p^*}\cap  N_p|=|N_p|-1$, and $\sum_{i\in N_{p^*}} b_i <\sum_{i\in N_{p^*}} b_i^V$. Then, define 
     \begin{align*}
         c(p^*)&=\sum_{i\in N_{p^*}} b_i^V-(|N_{p^*}|-|N_p|+1)\frac{c(p)}{v^*}\\
         &=(|N_p|-1)\bigg(\frac{\alpha_v B}{n}-\frac{c(p)}{v^*}\bigg)+(|N_{p^*}|-|N_p|+1)\bigg(\frac{\alpha_v B}{n}-\frac{c(p)}{v^*}\bigg) \\
         &=v^*\bigg(\frac{\alpha_v B}{n}-\frac{c(p)}{v^*}\bigg).
     \end{align*} 
     \item $\hat{p}$, with $c(\hat{p})=c(p)$ and $N_{\hat{p}}$ selected in such a way that $|N_{\hat{p}}|=v^*$ such that $N_{\hat{p}}\cap N_p=\emptyset$ and $N_{p^*}\subseteq N_{\hat{p}}\cup N_p$ (i.e. to make sure each voter in $N_{p^*}$ is either in $N_p$ or $N_{\hat{p}}$, but not both). We can choose a suitable set due to our assumption that $v^*\leq n-|N_p|$.
\end{itemize}

Let our instance $I\in \mathcal{I}$ have project set $P=\{p,\hat{p},p^*\}$, and consider the output of $\alg{MES}^\pay{M}(I, B, P_0)$. The supporters of $p^*$ cannot initially afford it \emph{with equal payments}, recalling $b_{i'}<b_{i'}^V=\frac{\alpha_v B}{n}-\frac{c(p)}{v^*}$, so $\hat{p}$, which can be afforded with equal payments, is selected over $p^*$. When $\hat{p}$ is selected, all voters in $N_{\hat{p}}$ pay $\frac{c(p)}{v^*}$ each for it. Thus, the voters in $N_{p^*}$ have total remaining budget of $(\sum_{i\in N_{p^*}} b_i)-(|N_{p^*}|-|N_p|+1)\frac{c(p)}{v^*}$, and can no longer afford $p^*$.

Then $p^*$ with voter group $N'=N_{p^*}$ constitutes a Strong EJR+ up to any project violation as the following is true for every voter $i\in N_{p^*}$, recalling that each voter approves only one of $p$ and $\hat{p}$.

\begin{align*}
    \sum_{p'\in P\cap A_i} \frac{c(p')}{\max\{|N'\cap N_{p'}|,|N'\cap N_{p^*}|\}}
    &=\sum_{p'\in P\cap A_i} \frac{c(p')}{|N_{p^*}|} \\
    &= \frac{c(p)+c(p^*)}{|N_{p^*}|} \\
    &\leq \frac{c(p)+
    v^*(\frac{\alpha_v B}{n}-\frac{c(p)}{v^*})}{|N_{p^*}|} \\
    &=\frac{v^*\frac{\alpha_v B}{n}}{|N_{p^*}|} \\
    &=\frac{\alpha_v B}{n}
\end{align*}

\end{proof}

We leave an exploration of the optimality of the \pay{Value-Based} pre-allocation method for other satisfaction functions, and non-singleton sets of pre-selected projects, for future work.

\section{Further Experimental Results} \label{app:experimental}

In this section, we discuss some more detailed results from the experiments we conducted. We denote by $\mix^\pay{M}$ the mixed rule $[\alg{Greedy}, \alg{MES}^\pay{M}, \alg{Greedy}]$, mixing \alg{Greedy} with a budget share of $\alpha_\alg{G}$ and \alg{MES} with \alg{Greedy} completion. Similarly, $\mix^{\pay{M}+}$ denotes the rule mixing \alg{Greedy} and \alg{MES} with budget increase and \alg{Greedy} completion.

\subsection{Minimum Voter Budget Shares in Practice} \label{app:alpha_v_in_practice}

In \Cref{def:minimum_voter_budget_share} we introduce the concept of the \textit{minimum voter budget share} $\alpha^M$, based on which we show different proportionality guarantees for $\alg{MES}^M$ in \Cref{sec:proportionality}. In particular, we show that the outcome of $\alg{MES}^M(B', P_0)$ satisfies $\alpha^M$-budget EJR up to any project for $M \in \{\pay{Null, Value-Based}\}$ and $\alpha^M$-budget EJR up to any two projects for $M = \pay{Equal-Split}$ if $P_0$ was selected by \alg{Greedy}.
The relation between $\alpha^M$ for different pre-allocation methods $M$ is stated in \Cref{prop:min_voter_budget_share}. In particular, we know that $\alpha^\pay{Null} \le \alpha^\pay{Value-Based} \le \alpha^\pay{Equal-Split}$, where $\alpha$ is the available budget share of $\alg{MES}^M$. To get a feeling how far apart these values are in practice, we computed them on our data for the pre-allocation of $P_0=\alg{Greedy}(\alpha_\alg{G} B,\emptyset)$ for $\alpha_\alg{G} \in \{0, 0.1, 0.2, \dots, 1\}$.

\begin{figure}
    \begin{minipage}[t]{.48\textwidth}
        \centering
        \includegraphics[scale=0.60]{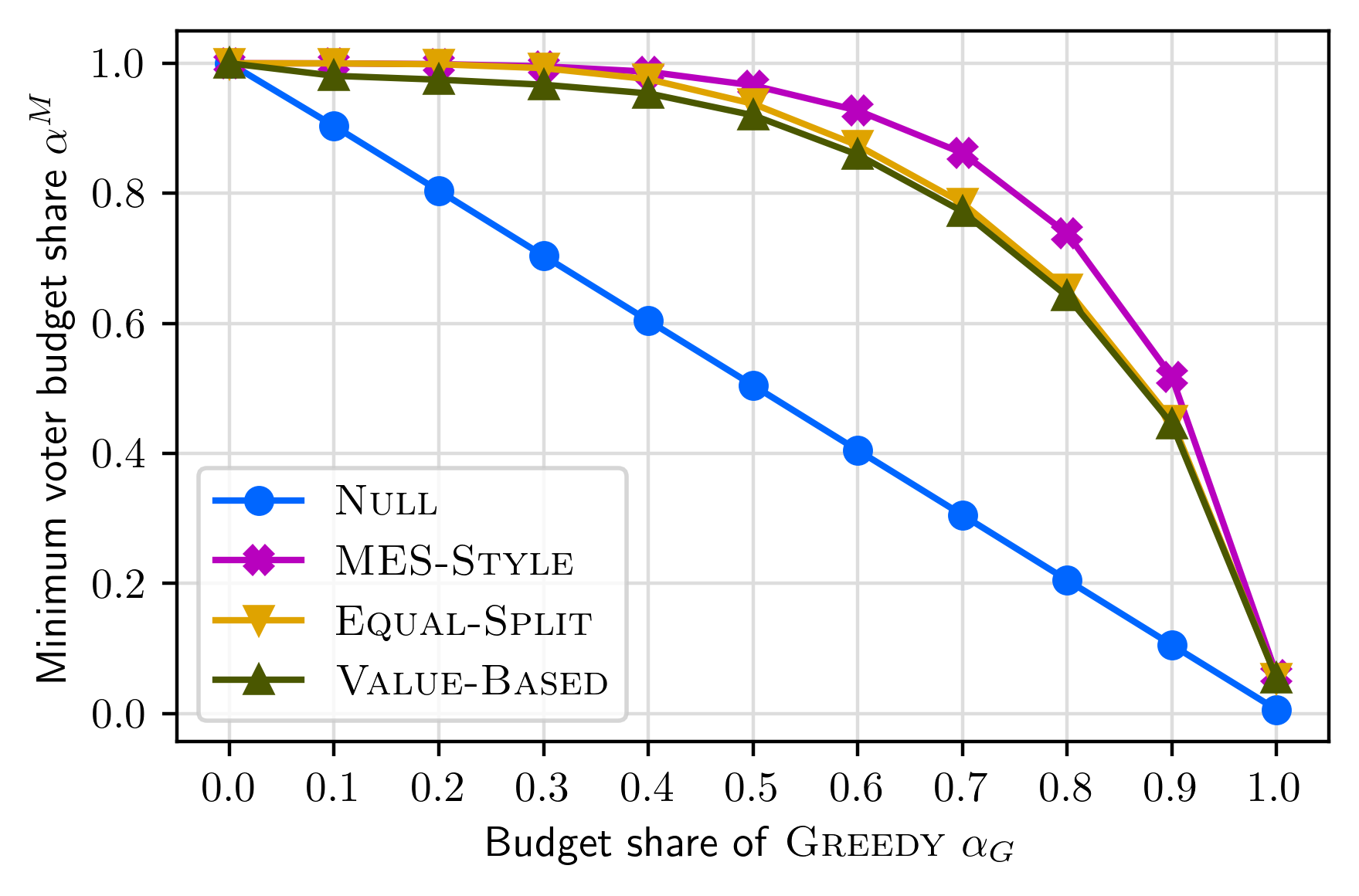}
        \caption{Minimum voter budget share $\alpha^{\pay{M}}$ in practice for different pre-allocation methods and \alg{Greedy} budget shares $\alpha_\alg{G}$, averaged over all instances with at least 20 projects.}
        \label{fig:results_alpha_v_no_budget_increase}
    \end{minipage}\hfill
    \begin{minipage}[t]{.48\textwidth}
        \centering
        \includegraphics[scale=0.60]{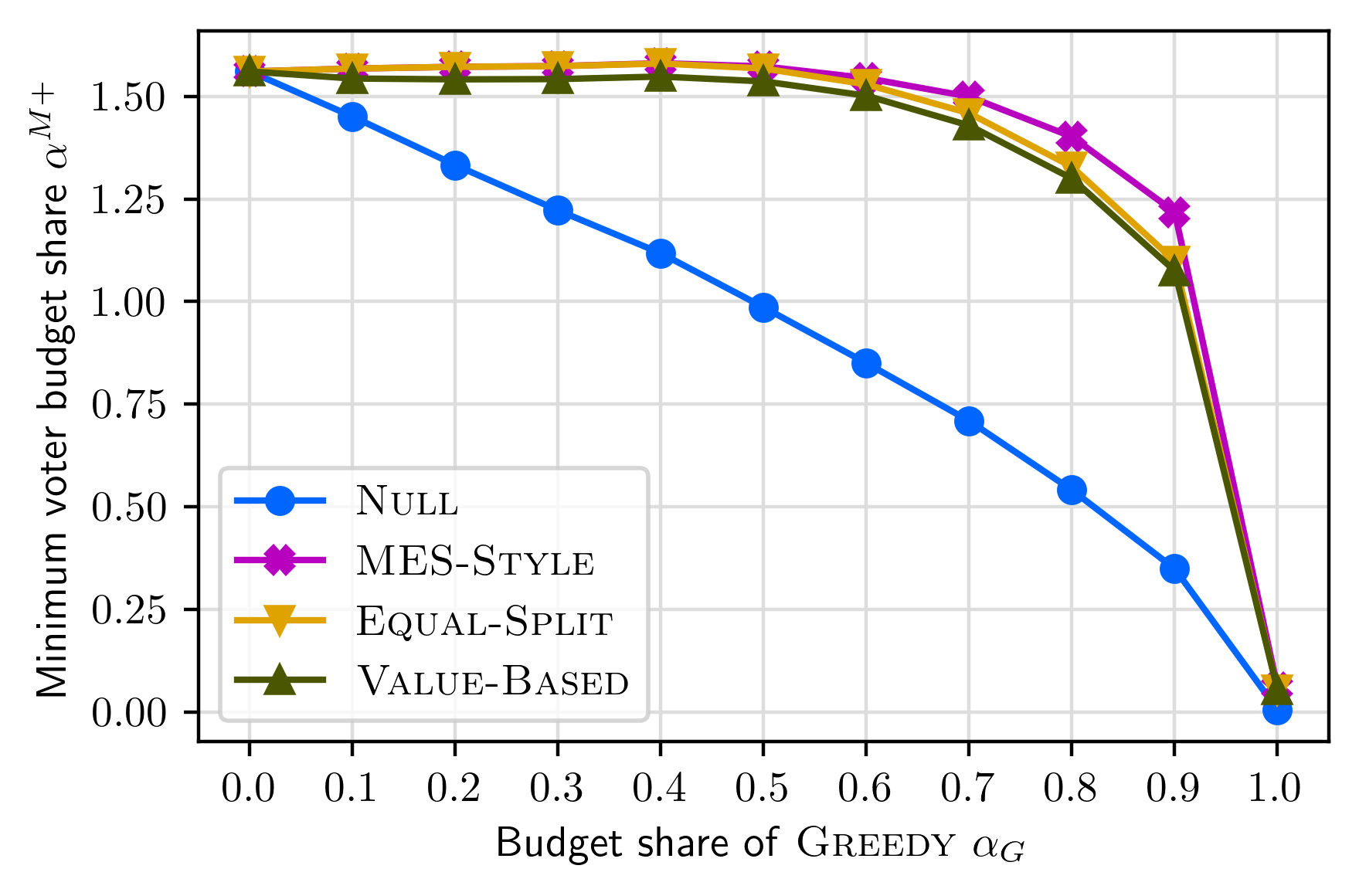}
        \caption{Minimum voter budget shares $\alpha^{\pay{M}+}$ with respect to the increased budget used by $\mesincrease{M}$ for different pre-allocation methods and \alg{Greedy} budget shares $\alpha_\alg{G}$, averaged over all instances with at least 20 projects.}
        \label{fig:results_alpha_v}
    \end{minipage}
\end{figure}

\Cref{fig:results_alpha_v_no_budget_increase} shows the average minimum voter budget share $\alpha^M$ for each pre-allocation method $M$, and for different greedy shares of the budget, over all instances with at least $20$ projects. We observe that $\alpha^\pay{Null}$ is significantly lower than the minimum voter budget share produced by the other three methods, suggesting that $\mix^\pay{Value-Based}$ provides a much better proportionality guarantee than $\mix^\pay{Null}$ in practice. The proportionality guarantees of $\mix^\pay{Value-Based}$ and $\mix^\pay{Equal-Split}$ are incomparable, as \pay{Equal-Split} provides a weaker guarantee, but for a potentially larger fraction of the budget $\alpha^\pay{Equal-Split}$. However, \Cref{fig:results_alpha_v_no_budget_increase} shows that, in practice, $\alpha^\pay{Equal-Split}$ is not much larger than $\alpha^\pay{Value-Based}$.

In \Cref{app:budget_increase} we show that our proportionality guarantees also hold for \alg{MES} with budget increase, with respect to the minimum voter budget share $\alpha^{\pay{M}+}$. \Cref{fig:results_alpha_v} shows the values for $\alpha^{\pay{M}+}$ in practice. The general trend is the same as in \Cref{fig:results_alpha_v_no_budget_increase} with the minimum voter budget shares of $\pay{Equal-Split}$, $\pay{Value-Based}$ and $\pay{MES-Style}$ being close together and significantly higher than for $\pay{Null}$. Note that all values of $\alpha^{\pay{M}+}$ are larger than $\alpha^{\pay{M}}$ by a factor of at least $1.5$.

\subsection{Budget Spending of \alg{MES} in a Mixed Rule}\label{app:budget_spending}

\begin{figure}
    \begin{subfigure}[t]{.48\textwidth}
        \centering
        \includegraphics[scale=0.60]{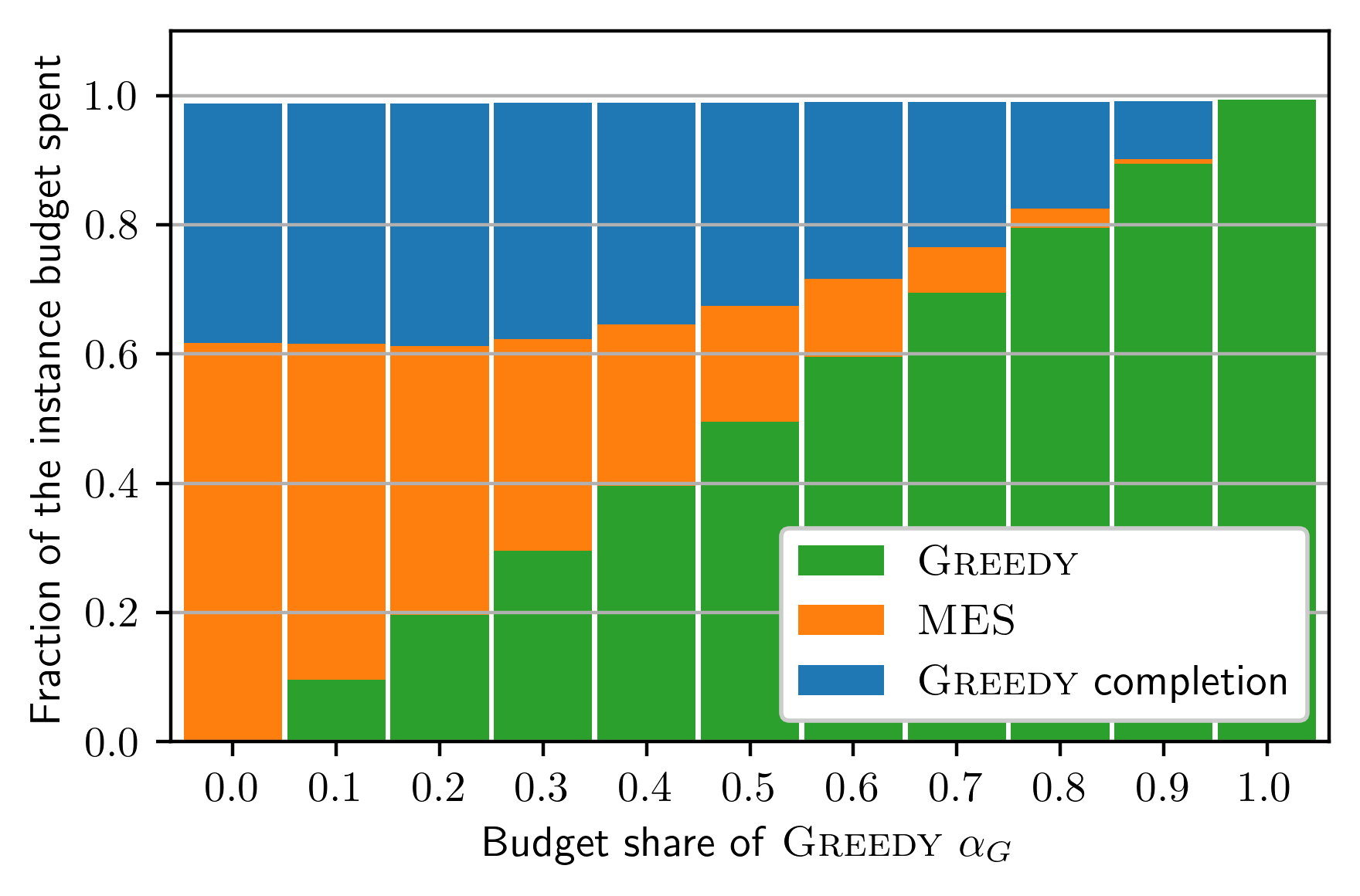}
        \caption{$\mix^{\pay{Value-Based}}$ (without budget increase)}
        \label{fig:budget_spending_no_budget_increase}
    \end{subfigure}\hfill
    \begin{subfigure}[t]{.48\textwidth}
        \centering
        \includegraphics[scale=0.60]{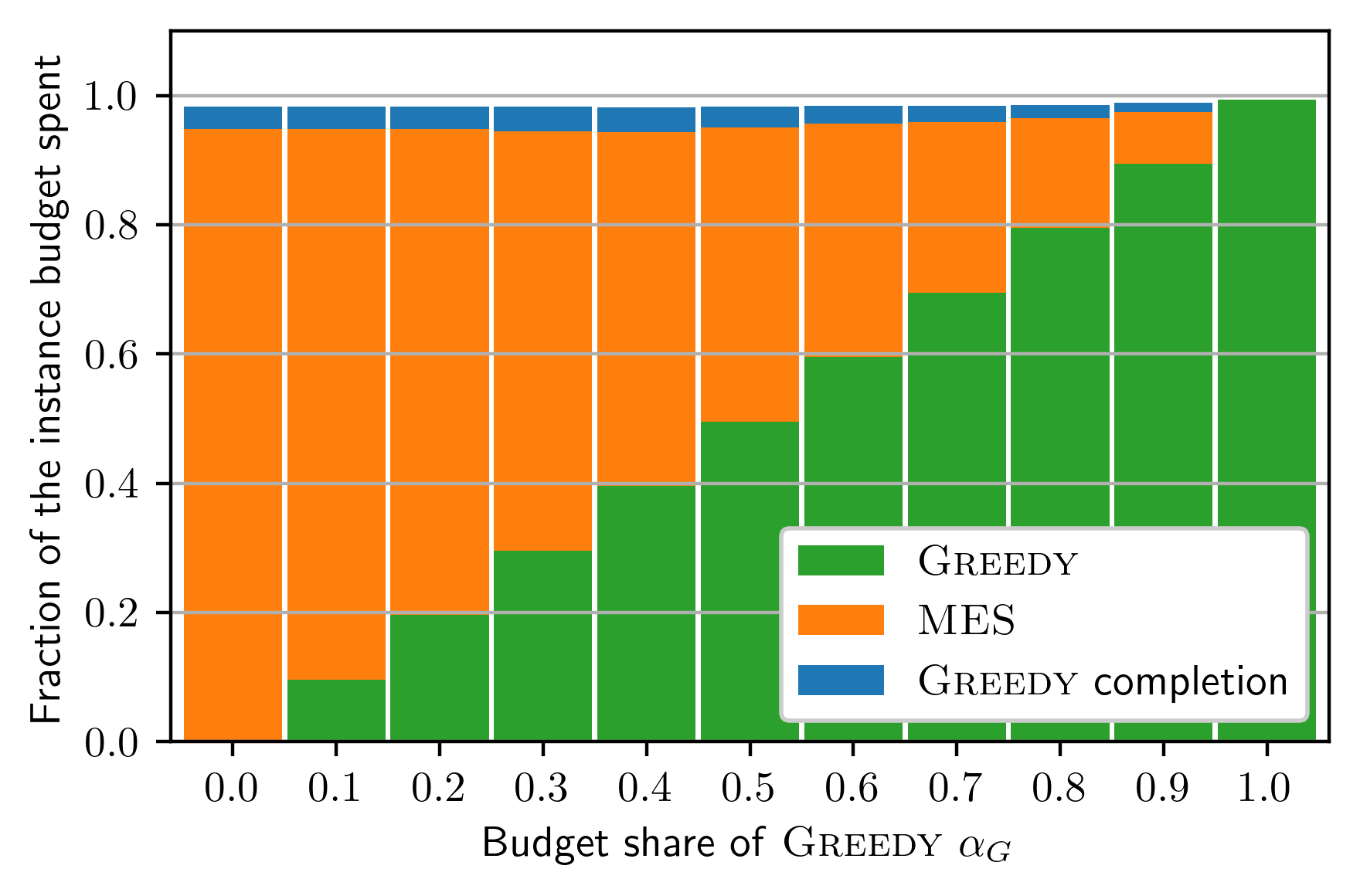}
        \caption{$\mix^{\pay{Value-Based}+}$ (with budget increase)}
        \label{fig:budget_spending}
    \end{subfigure}
    \caption{Budget spending of each individual rule in the mixed rules $\mix^\pay{Value-Based}$ (left) and $\mix^{\pay{Value-Based}+}$ (right), averaged over all instances with at least 20 projects.}
\end{figure}

A well known problem of the Method of Equal Shares is that it is not exhaustive, i.e., after its execution, there can be unchosen affordable projects. This occurs because \alg{MES} is constrained to only select projects that can be funded by their supporters. In practice this issue is often solved by using \alg{MES} with budget increase and/or a completion method, like \alg{Greedy}. 
In this section, we explain the empirical issues that arise when using \alg{MES} without budget increase in our mixed framework and show how introducing the budget increase can fix them. For a theoretical analysis, see \Cref{app:budget_increase}.

\Cref{fig:budget_spending_no_budget_increase} shows how much of the budget is spent by each separate rule in $\mix^{\pay{Value-Based}}$, averaged over all instances with at least $20$ projects. We can see, that while (the first) $\alg{Greedy}$ spends its available budget $\alpha_\alg{G}B$ almost entirely, this is not true for \alg{MES}. %
In fact, for large values of $\alpha_\alg{G}$, the spending of \alg{MES} is close to zero. In these cases, $\mix^M$ can almost be interpreted as mixing \alg{Greedy} with \alg{Greedy}, instead of \alg{Greedy} with \alg{MES}. We discuss mixing \alg{Greedy} with itself in the subsection below. In \Cref{fig:budget_spending} we show the spending of the mixed rule $\mix^{\pay{M}+}$ containing \alg{MES} with budget increase (in steps of 10). Notably, \alg{MES} now utilizes nearly its entire budget, resolving this issue.

\subsection{Further Empirical Evaluation}\label{app:more_experiments}

In this section, we evaluate our results against two empirical baselines. 
When interpolating between \alg{Greedy} and \alg{MES}, it is desirable to perform at least as well as the former with respect to representation and at least as well as the latter with respect to utilitarian welfare. In other words, we do not want the outcome of $\mix^{\pay{M}+}$ to be ``Pareto-dominated'' by either constituent rule. We refer to this requirement as the \textit{weak empirical baseline}, visualized in \Cref{fig:empirical_baselines}. 
All rules $\mix^{\pay{M}+}$ are dominated by \alg{MES} for $\alpha_\alg{G}$ up to $0.4$, because of the initial decrease of utilitarian welfare. As already stated in \Cref{sec:experiments}, we attribute this to \alg{Greedy} not being able to pick the most popular projects, because of its low budget constraint (see \Cref{app:splitting_budget} for a way to avoid this). For $\alpha_\alg{G}\ge 0.6$, the mixed rule $\mix^{\pay{M}+}$ exceeds the weak empirical baseline with any of the four pre-allocation methods.

We can also compare our mixed rules to a hypothetical \textit{randomized combination} of \alg{Greedy} and \alg{MES}, where we run \alg{Greedy} with some probability $p \in [0,1]$ and \alg{MES} with probability $1-p$. Ideally, we would want the outcome of our mixed rule $\mix^{\pay{M}+}$ to not be Pareto-dominated by any of these randomized rules. We consider this to be the \textit{strong empirical baseline}, also visualized in \Cref{fig:empirical_baselines}.
We can see that $\mix^{\pay{M}+}$ meets the strong empirical baseline for any payments method if $\alpha_\alg{G}\ge 0.6$.

\begin{figure}
    \centering
    \begin{minipage}[t]{.48\textwidth}
        \centering
        \includegraphics[scale=0.62]{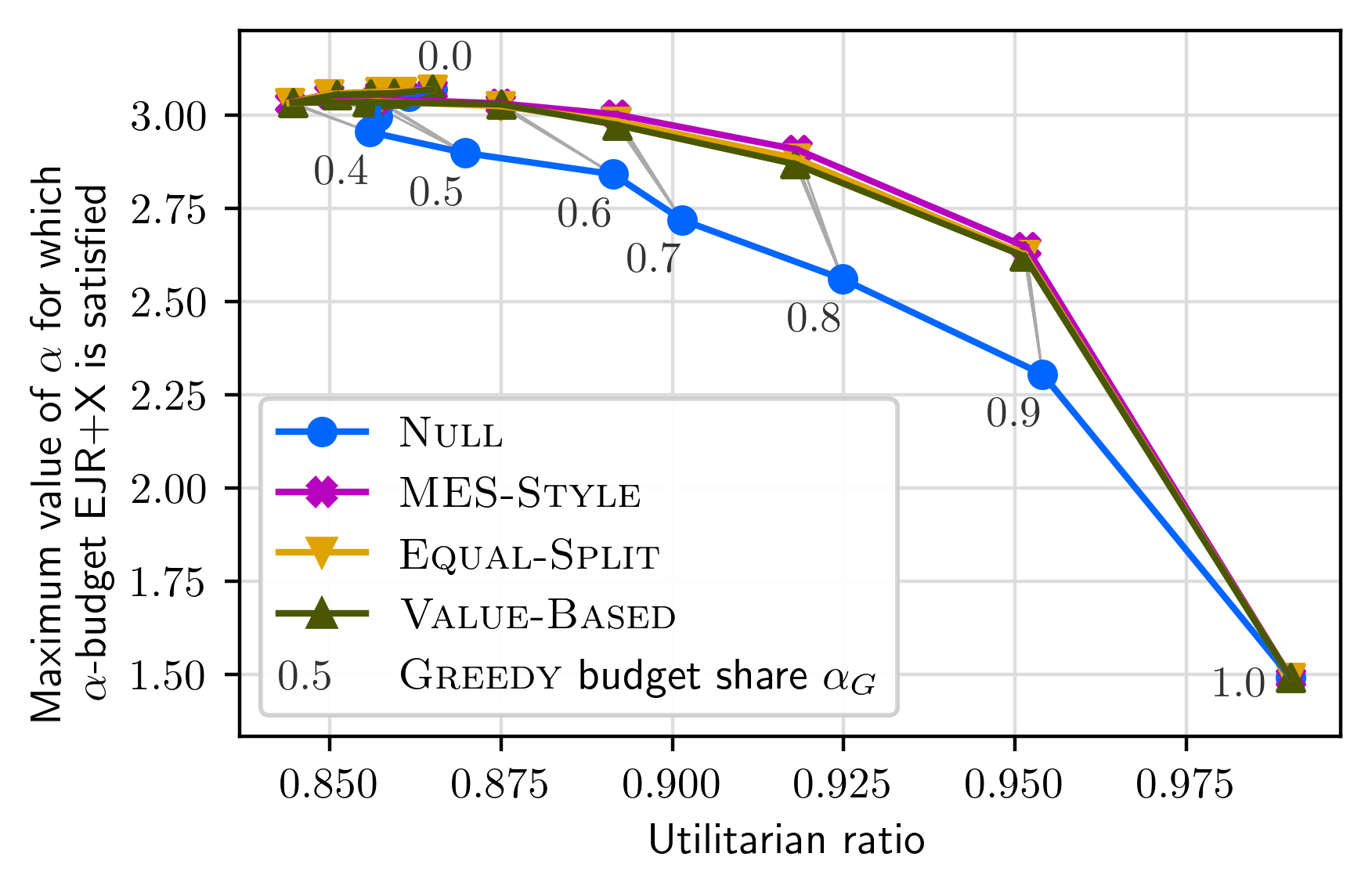}
        \caption{Results for $\mix^{M+}([\alpha_\alg{G}B,B,B])$ for $M \in \{\pay{Null}, \pay{MES-Style}, \pay{Equal-Split}, \pay{Value-Based}\}$ and $\alpha_{G} \in \{0, 0.1, 0.2, \dots, 1.0\}$. Metrics are averaged over all instances with at least $20$ projects.}
        \label{fig:results_min_beta_EJR_restated}
    \end{minipage}\hfill
    \begin{minipage}[t]{.48\textwidth}
    \begin{tikzpicture}
            \def\xGreedy{8}
            \def\yGreedy{1}
            \def\xMES{1}
            \def\yMES{8}
            \def\xmax{11}
            \def\ymax{10}
            \small
            \begin{axis}[
                xscale = 0.87,
                yscale = 0.65,
                xmin = 0, xmax = \xmax+1,
                ymin = 0, ymax = \ymax+1,
                axis lines = left,
                xtick = \empty, ytick = \empty,
                clip = false,
                xlabel=Utilitarian welfare,
                ylabel=Proportionality,
            ]
        
                \addplot[color = black, thick, ->, dashed] coordinates {(\xMES, \yGreedy) (\xMES, \ymax)};
                \addplot[color = black, thick, ->, dashed] coordinates {(\xMES, \yGreedy) (\xmax, \yGreedy)};
                \addplot[color = black, thick, dashed] coordinates {(\xMES, \yMES) (\xGreedy, \yGreedy)};
                \filldraw [black] (\xGreedy, \yGreedy) ellipse (2.5pt and 3.25pt);
                \node (Greedy) [above right] at (\xGreedy, \yGreedy) {Greedy};
                \filldraw [black] (\xMES, \yMES) ellipse (2.5pt and 3.25pt);
                \node (MES) [above right] at (\xMES, \yMES) {MES};
                
                \node (Min) at (\xMES, \yGreedy) {};
                \node (Max) at (\xGreedy, \yMES) {};

                \fill[teal, opacity = 0.1] (\xMES, \yMES) -- (\xMES, \yGreedy) -- (\xGreedy, \yGreedy);
                \fill[violet, opacity = 0.1] (\xMES, \yMES) -- (\xGreedy, \yGreedy) -- (\xmax, \yGreedy) -- (\xmax, \ymax) -- (\xMES, \ymax);
        
                \node[text width=2cm,align=center, text=teal!100] at (barycentric cs:Greedy=1,MES=1,Min=2) {Weak  baseline};
                \node[text width=2cm,align=center, text=violet!100] at (barycentric cs:Greedy=1,MES=1,Max=2) {Strong  baseline};
            \end{axis}
        \end{tikzpicture}
        \caption{Empirical baselines for evaluating the performance of mixed rules.}
        \label{fig:empirical_baselines}
    \end{minipage}
\end{figure}

\subsection{Effects of Splitting the Budget}\label{app:splitting_budget}
In this section, we investigate the issues that arise from splitting the budget between two rules, particularly when the first rule in the mix is exhaustive. Using \alg{Greedy} as a running example, we demonstrate these effects empirically and examine how the resulting issues can be resolved.

\paragraph{How much does splitting the budget matter?} One might wonder whether reallocating a small amount of budget from one rule to another has a noticeable effect on the final outcome. In particular, \Cref{fig:results_min_beta_EJR_restated} shows that for values of $\alpha_G$ up to $0.4$, the two metrics we consider barely change. Furthermore, at the start of their execution, \alg{Greedy} and \alg{MES} select the same projects, before voters start to exhaust their personal budgets in the latter. 

\begin{wrapstuff}[type=figure,width=0.49\textwidth]
        \centering
        \includegraphics[scale=0.5]{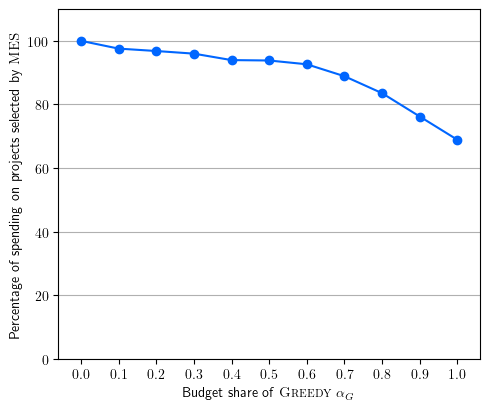}
        \caption{Overlap between the outcomes of \alg{MES} (with budget increase and greedy completion) and the mixed rule $\mix^{\pay{Value-Based}+}([\alpha_\alg{G}B,B,B])$, averaged over all instances with at least $20$ projects.}
        \label{fig:percentage_of_spending_on_projects_selected_by_MES}
\end{wrapstuff}
Consider \Cref{fig:percentage_of_spending_on_projects_selected_by_MES}, which shows the difference in outcomes between `pure' \alg{MES} (with budget increase and greedy completion) and the mixed rule combining \alg{Greedy} with \alg{MES} using the \pay{Value-Based} method. The curve plots the cost of projects that are commonly selected by both rules, for a given value of $\alpha_G$, normalized by the total spending of pure MES. By construction, this ratio is 1 for $\alpha_G=0$ when the two rules coincide. Pure \alg{MES} and pure \alg{Greedy} do share a substantial fraction of their selections (the curve shows, at $\alpha_G=1$, that almost 70\% of \alg{MES}'s spending goes to projects also selected by \alg{Greedy}).  However, similarity to the MES outcome gradually decreases as $\alpha_G$ increases, even for low values of $\alpha_G$, indicating that the interpolation does affect individual project selection even when aggregate metrics are stable.

\paragraph{Splitting the Budget Between Exhaustive Rules}
Splitting the budget between multiple rules can sometimes have undesirable consequences\,---\,for example, a project costing more than $50\,\%$ of the budget is unlikely to ever be selected by a mixed rule that splits the budget into two equal sized parts. This is especially problematic when the first component of a mixed rule is exhaustive, leaving little more than $50\,\%$ of the budget to the second rule.
The problem is less likely to occur in larger instances, since the average project cost tends to be lower for instances with a large number of projects. For example, the cost of the most expensive project in instances with less than $10$ projects is $71\,\%$ of the budget limit on average, while this value is only $44\,\%$ over instances with at least $20$ projects.
We can empirically observe the effect of splitting the budget, as we explain below.

\paragraph{Mixing \alg{Greedy} with Itself} By comparing the outcomes of $\mix^\alg{G} = [\alg{Greedy}, \alg{Greedy}]$ to the outcome of \alg{Greedy}, we can isolate the consequences of splitting the budget.
\Cref{fig:greedy_greedy_util} shows the utilitarian welfare for $\mix^\alg{G}([\alpha_\alg{G} B, B])$, when varying $\alpha_\alg{G}$ for different instance sizes (measured by the number of projects). Note that for values $0$ and $1$ of $\alpha_\alg{G}$ the mixed rule $\mix^\alg{G}$ is just the \alg{Greedy} rule. The performance of $\mix^\alg{G}$ drops significantly when splitting the budget for instances with fewer than 20 projects, whereas the effect is much less pronounced for larger instances. For instances with at least $50$ projects even splitting the budget in half seems to have almost no effect. This indicates that using mixed rules with an exhaustive first component does not work well in practice on small instances.

Another interesting side effect can be observed in \Cref{fig:greedy_greedy_prop}, which shows the performance of $\mix^\alg{G}$ in terms of proportionality, compared to that of $\mix^{\pay{Value-Based}}$ and $\mix^{\pay{Value-Based}+}$. It seems surprising that the proportionality increases when splitting the budget of the \alg{Greedy} rule into two parts. Merely forcing \alg{Greedy} to pick cheaper projects (without checking \textit{who} approves these projects) seems to increase proportionality, which could potentially be caused by correlations between the votes, or the fact that when larger projects are chosen, more projects remain unchosen, which could make an EJR+X violation more likely.

\begin{figure}
    \centering
    \begin{subfigure}[t]{.48\textwidth}
        \centering
        \includegraphics[scale=0.62]{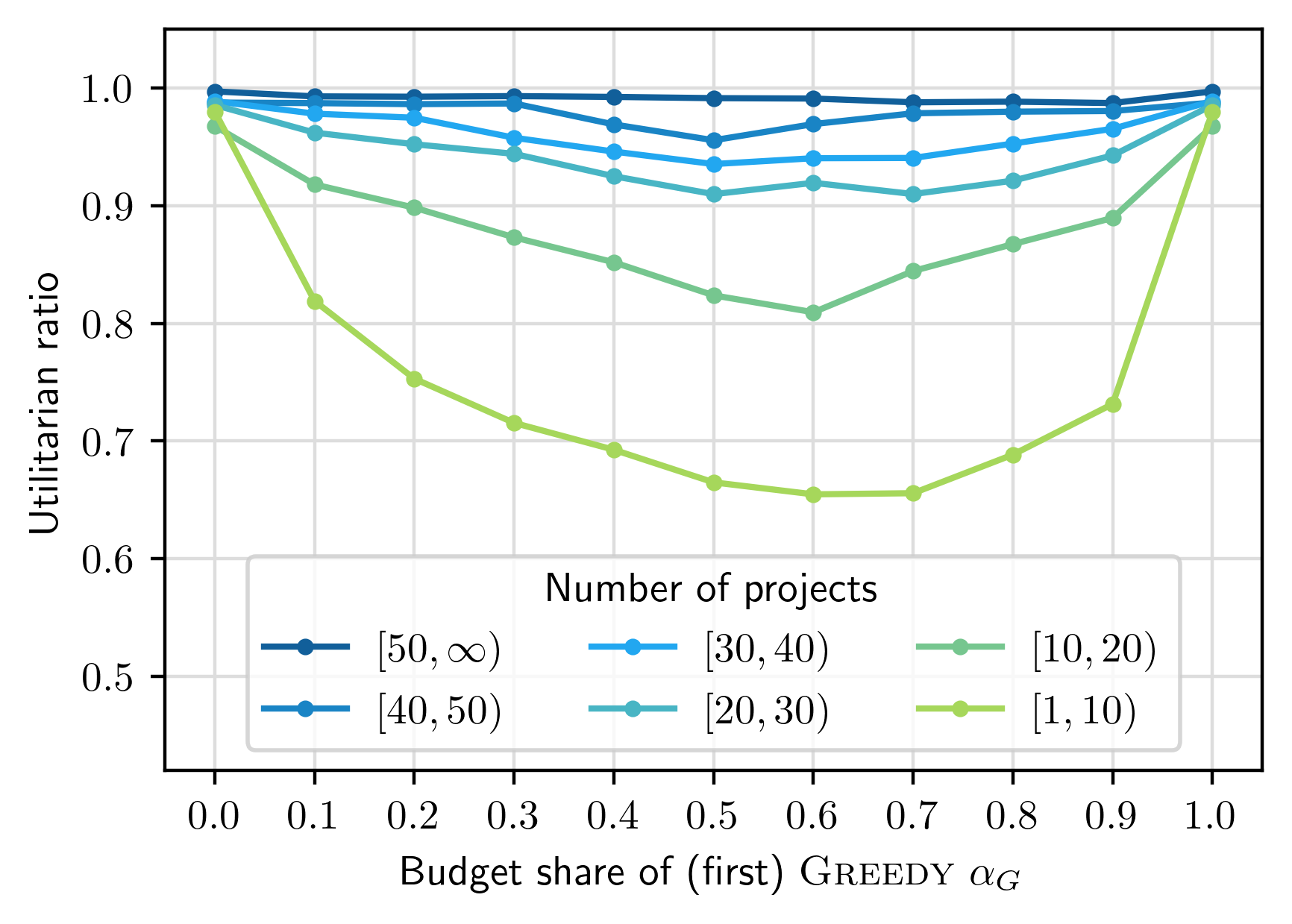}
        \caption{Loss of utilitarian welfare due to splitting the budget, averaged over instances of different sizes.}
        \label{fig:greedy_greedy_util}
    \end{subfigure}\hfill
    \begin{subfigure}[t]{.48\textwidth}
        \centering
        \includegraphics[scale=0.62]{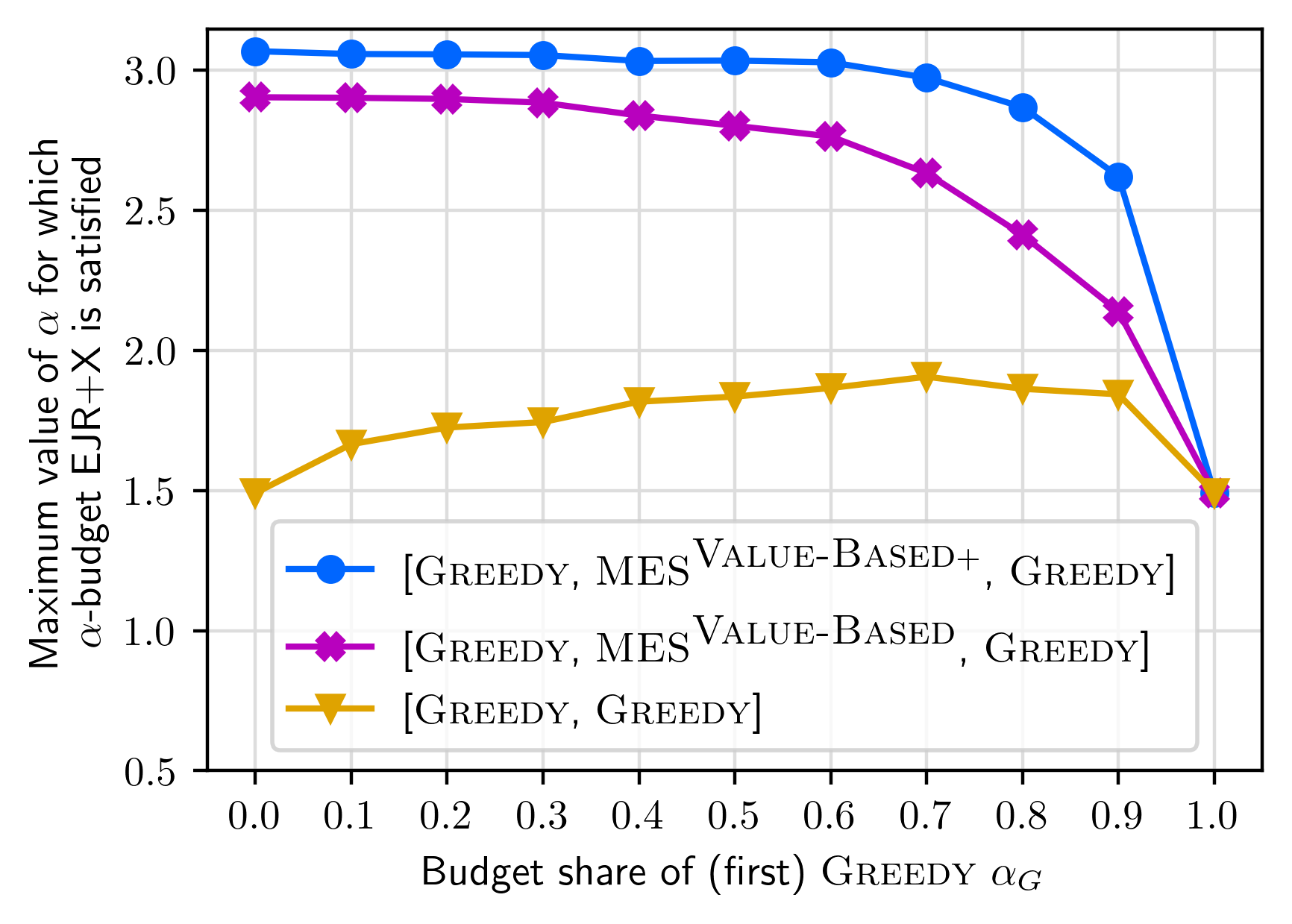}
        \caption{Increase of proportionality due to splitting the budget, averaged over all instances with at least 20 projects. The values for the mixed rule $\mix^\pay{Value-Based}+$ are displayed for comparison.}
        \label{fig:greedy_greedy_prop}
    \end{subfigure}
    \caption{Performance of mixing \alg{Greedy} with \alg{Greedy} for different (first) \alg{Greedy} budget shares $\alpha_\alg{G}$ from $0$ to $1$.}
    \label{fig:results_greedy_greedy}
\end{figure}

\paragraph{\alg{Greedy} with Early Stopping}
\begin{figure}
    \centering
    \begin{subfigure}[t]{.48\textwidth}
        \centering
        \includegraphics[scale=0.62]{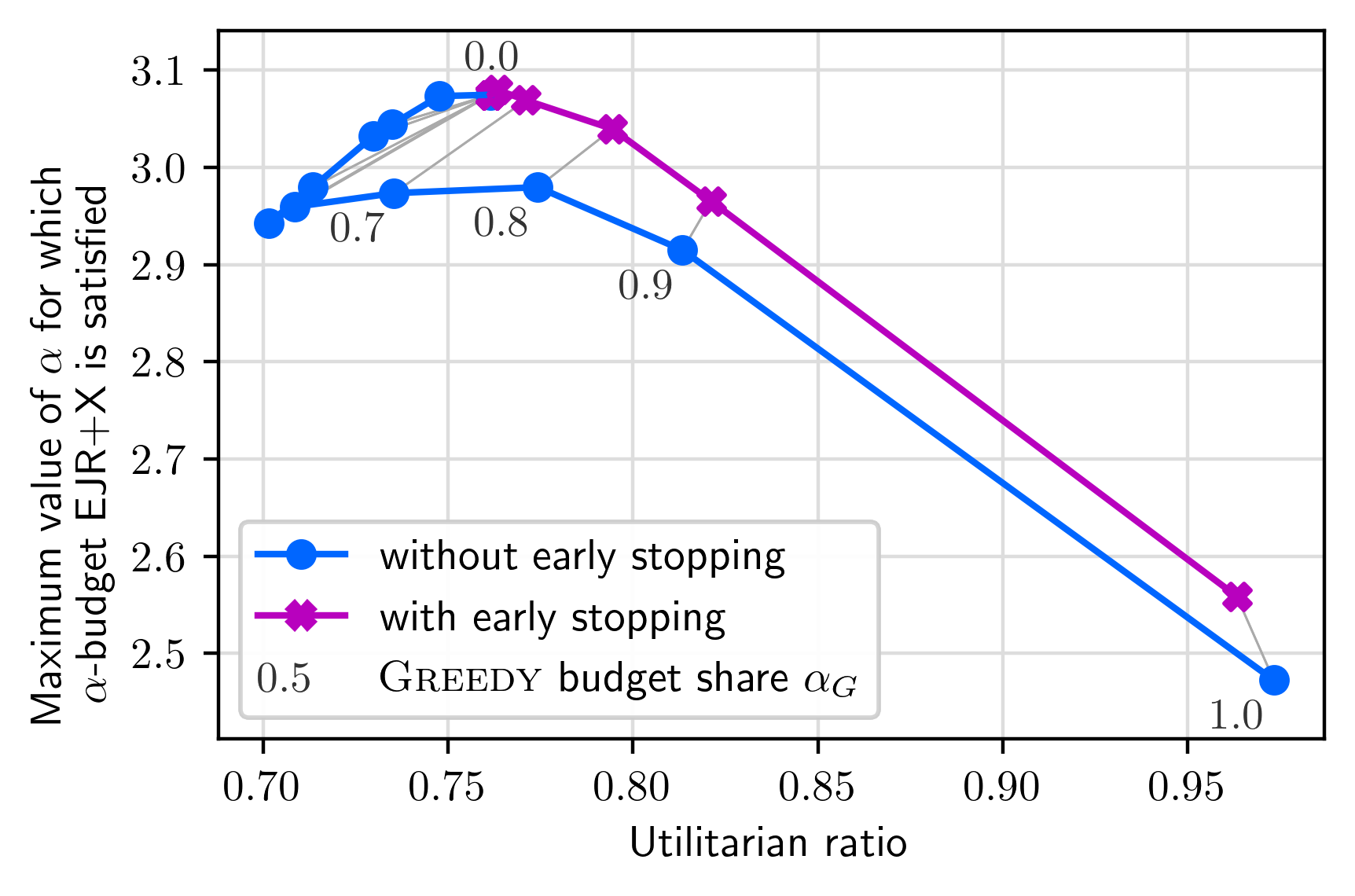}
        \caption{Proportionality over utilitarian ratio.}
        \label{fig:early_stopping_small}
    \end{subfigure}\hfill
    \begin{subfigure}[t]{.48\textwidth}
        \centering
        \includegraphics[scale=0.608]{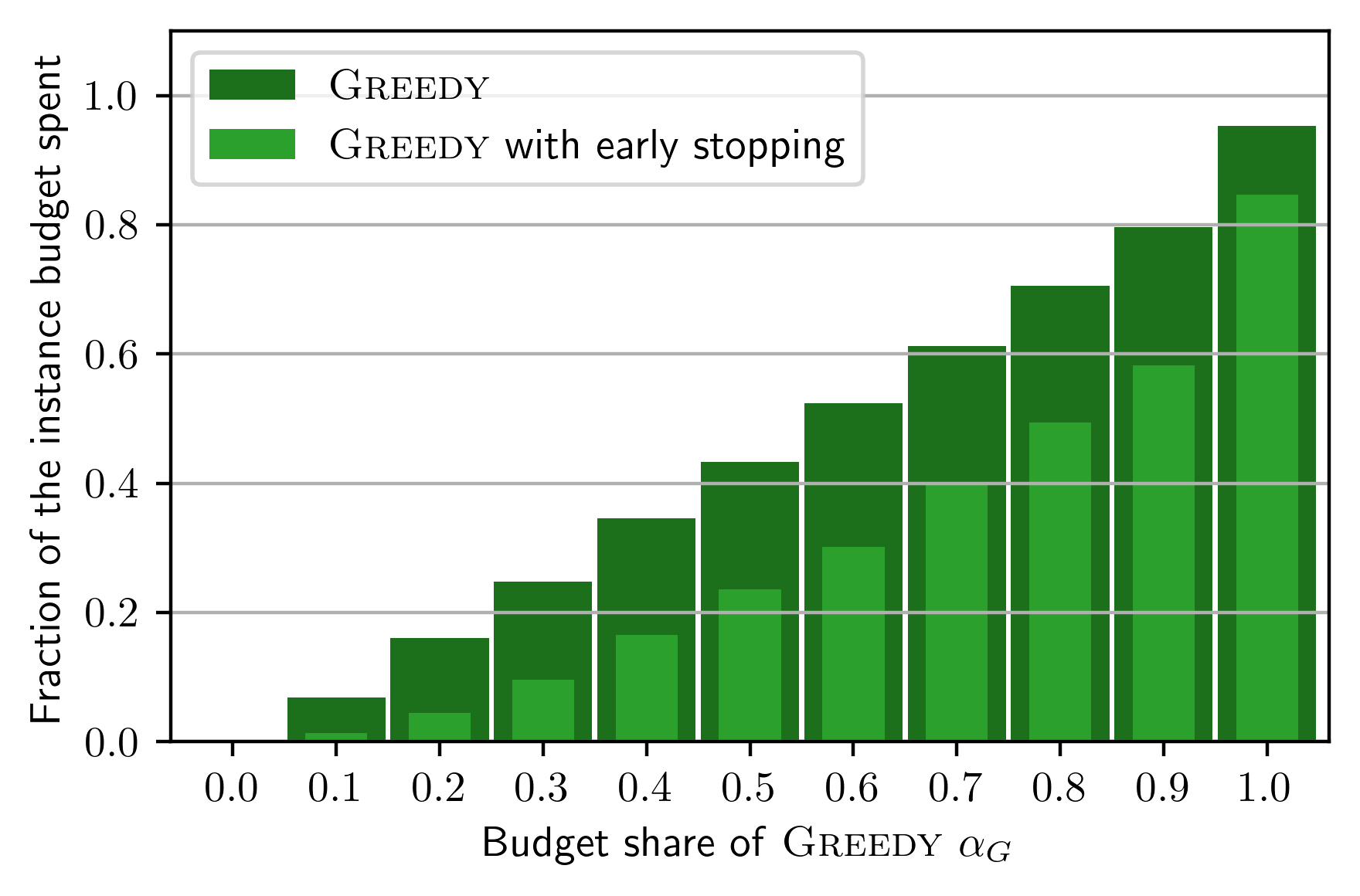}
        \caption{Budget spending of \alg{Greedy}.}
        \label{fig:early_stopping_spending}
    \end{subfigure}
    \caption{Experimental results of mixing \alg{Greedy} with and without early stopping with $\alg{MES}^{\pay{Value-Based}+}$, averaged over all instances with fewer than $20$ projects.}
    \label{fig:early_stopping_details}
\end{figure}

The aforementioned problems are mainly a consequence of using an exhaustive method like \alg{Greedy} as the first component in a mixed rule. In \Cref{sec:mixed}, we introduce a non-exhaustive version of \alg{Greedy} that iteratively selects the most approved projects until the next project to be added is not affordable. We call this variant ``\alg{Greedy} with early stopping'' and presented its effect on the experiments in \Cref{sec:exp_early_stopping}. \anton{Two new senetences, feel free to move:} Curiously, if the first instance of \alg{Greedy} in $\mix^\alg{G} = [\alg{Greedy}, \alg{Greedy}]$ is implemented with early stopping, then $\mix^\alg{G}$ is identical to the non-mixed \alg{Greedy} rule. Thus, using such an adaptation nullifies the effects discussed above.

Here, we focus on how early stopping affects mixed rules on small instances with fewer than 20 projects.
\Cref{fig:early_stopping_small} compares the performance of $\mix^{\pay{Value-Based}+}$ to its variant that uses \alg{Greedy} with early stopping. With the exception of $\alpha_\alg{G} = 1$, the variant with early stopping dominates the original mix for all \alg{Greedy} budget shares. Without early stopping, we observe a significant drop in welfare when assigning small fractions of the budget to \alg{Greedy}. In fact, the mixed rule only surpasses pure \alg{MES} in welfare once we allocate $80\,\%$ of the budget to \alg{Greedy}. In contrast, when using \alg{Greedy} with early stopping, the welfare remains essentially constant for greedy shares between $0$ and $0.6$. Furthermore, up to $60\,\%$ of the budget can be allocated to \alg{Greedy} without losing any proportionality, on average. A simple explanation could be that \alg{Greedy} with early stopping fails to spend any budget, because the cost of the most approved project exceeds the budget allocated to \alg{Greedy}. While this is certainly true for low values of $\alpha_\alg{G}$, it does not fully explain the behavior for larger values. \Cref{fig:early_stopping_spending} compares the average spending of \alg{Greedy} with and without early stopping for instances with fewer than $20$ projects. While \alg{Greedy}  spends significantly less when used with early stopping, the spending remains non-negligible. Notably, for $\alpha_\alg{G} = 0.6$, where $\mix^{\pay{Value-Based}+}$ with early stopping still matches the performance of pure \alg{MES}, the actual spending of \alg{Greedy} is around $30\,\%$ of the total budget.

Overall, introducing early stopping largely resolves the issues discussed in this section, albeit at the cost of reduced control over the exact budget share that \alg{Greedy} ultimately spends.

\subsection{Further Experimental Results for Other Rules} \label{app:other_methods_results}

\begin{figure}
    \centering
    \begin{subfigure}[t]{.48\textwidth}
        \centering
        \includegraphics[scale=0.62]{figures/EC/greedy_any.png}
        \caption{Proportionality over welfare. (\Cref{fig:results_other_methods} restated)}
        \label{fig:results_greedy_any_prop}
    \end{subfigure}\hfill
    \begin{subfigure}[t]{.48\textwidth}
        \centering
        \includegraphics[scale=0.62]{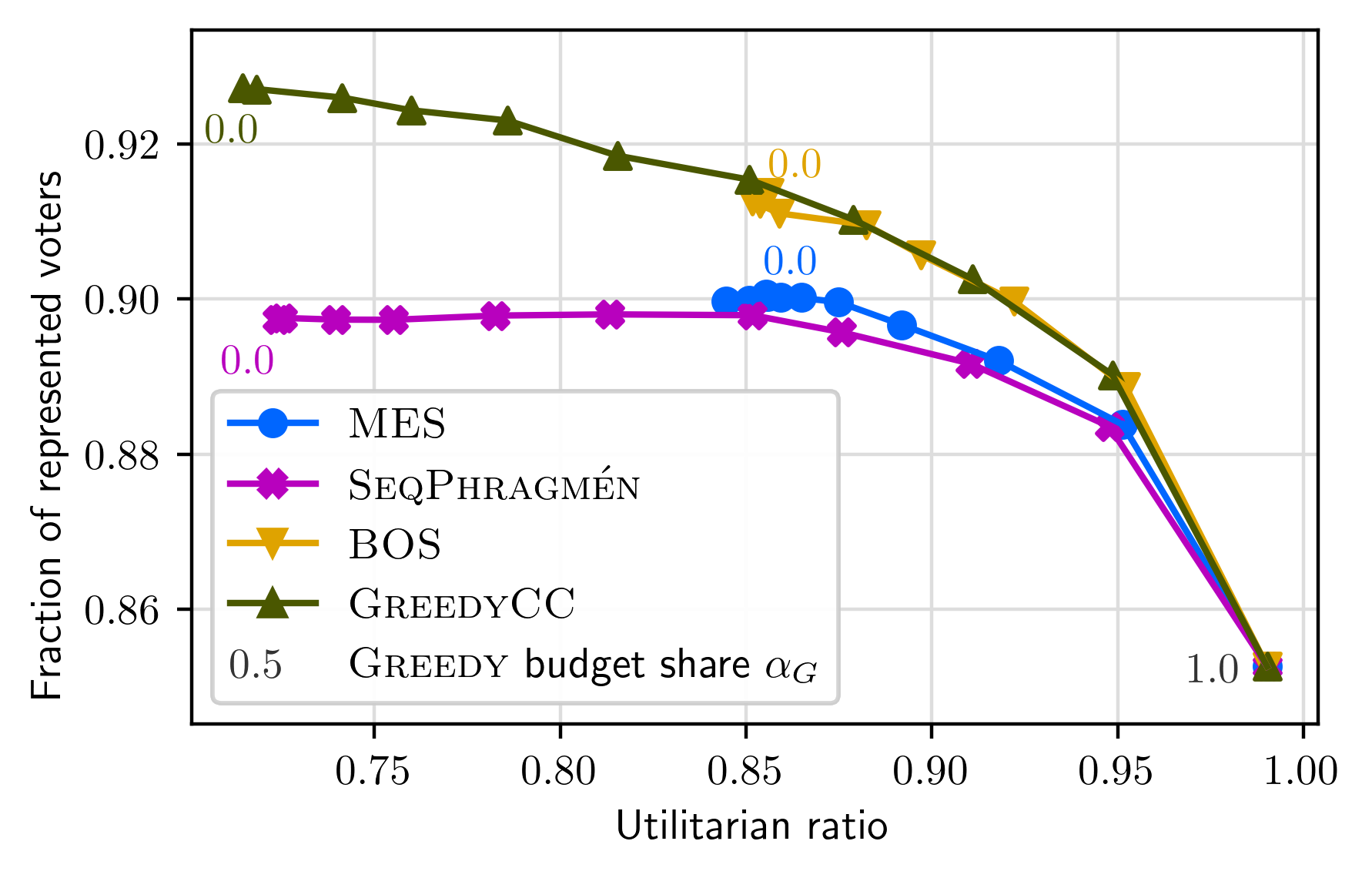}
        \caption{Representation over welfare.}
        \label{fig:results_greedy_any_rep}
    \end{subfigure}
    \caption{Comparison of experimental results when mixing \alg{Greedy} with \alg{MES}, \seqPhrag, \alg{BOS} and \alg{GreedyCC} as explained in \Cref{app:other_rules_theory}. Metrics are averaged over all instances with at least $20$ projects}
    \label{fig:results_greedy_any}
\end{figure}

In \Cref{sec:other_rules} we compare the performance of mixed rules containing \alg{MES}, \alg{BOS}, \seqPhrag and \alg{GreedyCC} with respect to proportional representation and utilitarian welfare.
However, when evaluating our mixed-method framework, it is most meaningful to assess each mixed rules using metric it is designed to optimize. In particular \alg{GreedyCC} is not trying to achieve proportional representation, but rather just to maximize the number of voters with non-zero satisfaction. We therefore introduce the following representation measure.

\begin{definition}
For a given instance $I = (B,P,A,c)$ and outcome $P^\ast \subseteq P$, we measure the \emph{fraction of represented voters} as the fraction of voters with non-zero satisfaction, i.e., $|\{i \in N \mid A_i \cap P^\ast \neq \emptyset\}|/n$.
\end{definition}

\Cref{fig:results_greedy_any_rep} illustrates the trade-off between this measure of representation and utilitarian welfare, when mixing \alg{Greedy} with \alg{MES}, \alg{BOS}, \seqPhrag and \alg{GreedyCC}, as explained in \Cref{app:other_rules_theory}. Unsurprisingly, \alg{GreedyCC} achieves the highest representation for all values of $\alpha_\alg{G}$. Notably, the curve for \alg{BOS} lies very close to that of \alg{GreedyCC}. In particular, for $\alpha_\alg{G} = 0.8$, \alg{BOS} attains representation levels comparable to $\alg{GreedyCC}$ while outperforming all methods with respect to proportionality and welfare. Since this point is also fairly close to the proportionality that pure \alg{MES} achieves, these results indicate that mixing even a small proportion of \alg{BOS} into \alg{Greedy} can substantially improve (proportional) representation in the outcome of a PB election.

\end{document}